\documentclass{lmcs} 
\pdfoutput=1

\usepackage{lastpage}
\lmcsdoi{21}{3}{24}
\lmcsheading{}{\pageref{LastPage}}{}{}%
{Feb.~15,~2023}{Sep.~03,~2025}{}


\usepackage{xcolor}
\usepackage{hyperref}
\usepackage{extarrows}
\definecolor{bluegray}{RGB}{160,200,200}
\usepackage[notion,quotation, electronic]{knowledge}
\knowledgeconfigure{quotation, protect quotation={tikzcd}}
\knowledgeconfigure{diagnose line=true, diagnose bar=true}
\definecolor{Blue Sapphire}{HTML}{003050} 
\definecolor{Gamboge}{HTML}{ee9b00}
\definecolor{Ruby Red}{HTML}{ab2226}

\IfKnowledgePaperModeTF{
}{
	\knowledgestyle{intro notion}{color={Ruby Red}, emphasize}
	\knowledgestyle{notion}{color={Blue Sapphire}}
	\hypersetup{
		colorlinks=true,
		breaklinks=true,
		linkcolor={Blue Sapphire}, 
		citecolor={Blue Sapphire}, 
		filecolor={Blue Sapphire}, 
		urlcolor={Blue Sapphire},
	}
}
\IfKnowledgeCompositionModeTF{
	\knowledgeconfigure{anchor point color={Ruby Red}, anchor point shape=corner}
	\knowledgestyle{intro unknown}{color={Gamboge}, emphasize}
	\knowledgestyle{intro unknown cont}{color={Gamboge}, emphasize}
	\knowledgestyle{kl unknown}{color={Gamboge}}
	\knowledgestyle{kl unknown cont}{color={Gamboge}}
}{
}

\knowledge{notion}
| lock-aware

\knowledge{notion}
| lock-sharing system
| lock-sharing systems
| LSS

\knowledge{notion}
| deadlock
| deadlocks
| leads to a deadlock in $\s$
| deadlocking run
| deadlocking
| deadlocked

\knowledge{notion}
| winning
| winning strategy
| wins

\knowledge{notion}
| control strategy
| strategy
| strategies
| Strategies
| Strategy

\knowledge{notion}
| \lss 

\knowledge{notion}
| neutral

\knowledge{notion}
| risky

\knowledge{notion}
| strong pattern
| strong patterns
| strong@pat

\knowledge{notion}
| weak pattern
| Weak patterns
| weak patterns
| weak@pat

\knowledge{notion}
 | lockable

\knowledge{notion}
| entirely synchronously

\knowledge{notion}
| $\PZ$

\knowledge{notion}
| solid
| $\PP$-solid
| solid process
| solid processes
| solid in $\PP$

\knowledge{notion}
| -fragile

\knowledge{notion}
| fragile
| $\PP$-fragile
| fragile process
| fragile processes
| fragile in $\PP$

\knowledge{notion}
| $Z$-deadlock scheme
| $Z$-deadlock scheme for $\PP$
| $B$-deadlock scheme
| $(B\cup Z)$-deadlock scheme

\knowledge{notion}
| deadlock scheme
| deadlock schemes

\knowledge{notion}
| full $Z$-deadlock scheme
| full $Z$-deadlock schemes
| full $(B\cup Z)$-deadlock scheme
| full $B$-deadlock scheme
| full

\knowledge{notion}
| full deadlock scheme
| full deadlock schemes
| full deadlock scheme for $\PP$

\knowledge{notion}
| sufficient deadlock scheme
| sufficient $Z$-deadlock scheme
| sufficient

\knowledge{notion}
| sufficient deadlock scheme@H
| sufficient@H

\knowledge{notion}
| lock graph
| $\Gp$

\knowledge{notion}
| strong edge
| strong edges
| strong@edge

\knowledge{notion}
| weak edge
| weak edges
| weak@edge

\knowledge{notion}
| admits a pattern
| pattern
| patterns

\knowledge{notion}
| simple cycle
| simple cycles
| simple
| simple path
| simple paths

\knowledge{notion}
| weak cycle
| weak@cycle
| weak cycles

\knowledge{notion}
| strong cycle
| strong cycles
| strong@cycle

\knowledge{notion}
| behavior@two
| behaviors@two

\knowledge{notion}
| local strategy
| local strategies

\knowledge{notion}
| deadlock avoidance control problem
| control problem

\knowledge{notion}
| controllable
| Controllable

\knowledge{notion}
| synchronously

\knowledge{notion}
| uncontrollable

\knowledge{notion}
| locally live

\knowledge{notion}
| exclusive@beh

\knowledge{notion}
| exclusive@sys
| Exclusive@sys

\knowledge{notion}
| exclusive@proc
| Exclusive@proc

\knowledge{notion}
| uses two locks
| 2LSS

\knowledge{notion}
| trimmed

\knowledge{notion}
| solid edge
| solid edges

\knowledge{notion}
| fragile edge
| fragile edges

\knowledge{notion}
| solid cycle
| solid cycles

\knowledge{notion}
| solo solid edge
| solo solid edges
| solo solid

\knowledge{notion}
| double solid edge
| double solid edges

\knowledge{notion}
| direct deadlock SCC
| direct deadlock

\knowledge{notion}
| deadlock SCC

\knowledge{notion}
| direct semi-deadlock
| direct semi-deadlock SCC

\knowledge{notion}
| semi-deadlock
| semi-deadlock SCC
| semi-deadlock SCCs
| Semi-deadlock SCCs

\knowledge{notion}
| unavoidable 

\knowledge{notion}
| nested-locking
| nested-locks

\knowledge{notion}
| nested-locking@locstrat
| nested-locks@locstrat

\knowledge{notion}
| nested-locking@strat
| nested-locks@strat

\knowledge{notion}
| nested-locking@proc
| nested-locks@proc

\knowledge{notion}
| behavior@nest
| behaviors@nest

\knowledge{notion}
| pattern@nest
| patterns@nest
| stair pattern
| stair patterns

\knowledge{notion}
| stair decomposition

\knowledge{notion}
| \Start

\knowledge{notion}
| init sequence
| init sequences
| initialization phase
| initialization phases

\knowledge{notion}
| deadlock avoidance control problem for LSS with initial lock ownership

\knowledge{notion}
| induces
| induced

\knowledge{notion}
| deadlock scheme@H
| full@H
| full deadlock scheme@H
| full $B$-deadlock scheme@H

\knowledge{notion}
| ($\star$)

\knowledge{notion}
| (Trim)

\usepackage[utf8]{inputenc}
\usepackage{mathtools}
\usepackage{amssymb}
\usepackage{amsthm}
\usepackage{thmtools, thm-restate}
\declaretheorem{theorem}
\usepackage{xspace}
\usepackage{tikz}
\usepackage{array}
\usepackage{changepage}
\usepackage{enumitem,hyperref}
\makeatletter
\def\namedlabel#1#2{\begingroup
	#2%
	\def\@currentlabel{#2}%
	\phantomsection\label{#1}\endgroup
}
\makeatother
\usepackage{xcolor}
\usepackage{algorithm}
\usepackage{algpseudocode}
\definecolor{bluegray}{RGB}{160,200,200}
\definecolor{light-gray}{gray}{0.75}
\definecolor{darkgreen}{RGB}{0,200,0}
\definecolor{lightpurple}{RGB}{220,0,220}
\usepackage[capitalise]{cleveref}
\usepackage{cite}

\usetikzlibrary{arrows.meta,arrows,calc,automata,shapes,positioning}
\tikzset{AUT style/.style={>=angle 60, initial text= ,initial distance=0.3cm,every edge/.append style={thick},every state/.style={thick,minimum size=15,inner sep=0.5}}}

\newtheorem{lemma}[theorem]{Lemma}
\newtheorem{corollary}[theorem]{Corollary}
\newtheorem{proposition}[theorem]{Proposition}
\newtheorem{claim}[theorem]{Claim}

\newtheorem{invariant}{Invariant}

\crefname{invariant}{Invariant}{Invariants}

\theoremstyle{definition}
\newtheorem{definition}[theorem]{Definition}
\newtheorem{remark}{Remark}
\newtheorem{example}{Example}

\renewcommand{\a}{\alpha}
\renewcommand{\b}{\beta}
\renewcommand{\d}{\delta}

\newcommand{\g}{\gamma}

\newcommand{\s}{\sigma}

\renewcommand{\S}{\Sigma}
\renewcommand{\epsilon}{\varepsilon}
\renewcommand{\phi}{\varphi}

\newcommand{\pow}[1]{2^{#1}}

\newcommand{\size}[1]{|#1|}
\newcommand{\set}[1]{\{ #1 \}}

\newcommand{\sigmatwo}{$\Sigma_2^{\textsc{P}}$\xspace}

\newcommand{\nexpt}{\textsc{NExpTime}\xspace}

\newcommand{\tss}{\Ss}
\newcommand{\tlss}{\textsf{"2LSS"}}
\knowledgenewrobustcmd{\lss}{\textsf{"LSS"}}

\newcommand{\run}{u}
\newcommand{\get}[1]{\mathtt{acq}_{#1}}
\newcommand{\rel}[1]{\mathtt{rel}_{#1}}
\newcommand{\eget}[2]{\mathtt{acq}^{#1}_{#2}}

\newcommand{\edge}[1]{\xrightarrow{#1}}
\newcommand{\ledge}[1]{\xleftarrow{#1}}
\newcommand{\dedge}[1]{\xleftrightarrow{#1}}

\newcommand{\aut}{\mathcal{A}}

\newcommand{\blackdiamond}{\mathalpha{\blacklozenge}}

\newcommand{\init}{\mathit{init}}

\newcommand{\df}{\mathit{ds}}

\newcommand\restrict[2]{{
		\left.\kern-\nulldelimiterspace
		#1
		\vphantom{\big|}
		\right|_{#2}
}}

\usepackage{logic9-thm}

\newcommand{\Runs}{\mathit{Runs}}

\knowledgenewrobustcmd{\eqh}{\cmdkl{\equiv_H}}
\newcommand{\neqh}{\not\equiv_H}

\newcommand{\act}[1]{\xlongrightarrow{#1}}

\newcommand{\lact}[1]{\xlongleftarrow{#1}}

\newcommand{\OpTp}{\mathit{Op}(T_p)}

\newcommand{\nop}{\mathit{nop}}

\newcommand{\Pat}{\PP}
\knowledgenewrobustcmd{\BTPP}{B_{\PP}}

\knowledgenewrobustcmd{\Blocks}{\cmdkl{\mathsf{Blocks}}}
\knowledgenewrobustcmd{\Owns}{\cmdkl{\mathsf{Owns}}}
\knowledgenewrobustcmd{\Towns}{\cmdkl{\mathsf{Owns}}}
\knowledgenewrobustcmd{\Twants}{\cmdkl{\mathsf{Blocks}}}

\renewcommand{\d}{\delta}

\newcommand{\Cinit}{C_{\init}}

\newcommand{\out}{\mathit{out}}

\newcommand{\Proc}{\mathit{Proc}}

\newcommand{\ds}{\mathit{ds}}

\newcommand{\lft}{\mathit{left}}
\newcommand{\rgt}{\mathit{right}}
\newcommand{\hungry}{\mathit{hungry}}
\newcommand{\think}{\mathit{think}}

\newcommand{\bP}{\overline{P}}
\newcommand{\bs}{\overline{s}}

\newcommand{\bb}{\overline{b}}

\newcommand{\bp}{\overline{p}}

\newcommand{\bx}{\overline{x}}
\newcommand{\bary}{\overline{y}}

\newcommand{\Ss}{\mathcal{S}}

\newcommand{\Iseq}{\mathit{IS}}

\newcommand{\greenedge}[1]{\stackrel{#1}{\Longrightarrow}}
\newcommand{\rededge}[1]{\stackrel{#1}{\dashrightarrow}}
\knowledgenewrobustcmd{\strongedge}[1]{\cmdkl{\greenedge{#1}}}
\knowledgenewrobustcmd{\weakedge}[1]{\cmdkl{\rededge{#1}}}
\newcommand{\pat}{\act{}}
\knowledgenewrobustcmd{\spat}{\cmdkl{\Longrightarrow}}
\knowledgenewrobustcmd{\wpat}{\cmdkl{\dashrightarrow}}
\knowledgenewrobustcmd{\Gp}{\cmdkl{G_{\PP}}}
\knowledgenewrobustcmd{\Gu}{\cmdkl{G_{u}}}
\knowledgenewrobustcmd{\BTu}{\cmdkl{B_{u}}}
\knowledgenewrobustcmd{\FTu}{\cmdkl{F_{u}}}
\knowledgenewrobustcmd{\PZ}{\cmdkl{\Proc_Z}}
\newcommand{\pproc}{{p\in\Proc}}

\newcommand{\down}{\mathit{down}}
\newcommand{\up}{\mathit{up}}
\newcommand{\equ}{\mathit{equality}}
\newcommand{\ver}{\mathit{vertical}}
\newcommand{\hor}{\mathit{horizontal}}
\newcommand{\domi}{\mathit{dominoes}}
\newcommand{\win}{\mathit{win}}

\newcommand{\Start}{\mathsf{Start}}

\begin{document}

\title{Distributed controller synthesis for deadlock avoidance}

\author[H.~Gimbert]{Hugo Gimbert\lmcsorcid{0000-0003-1227-9718}}[a]
\author[C.~Mascle]{Corto Mascle\lmcsorcid{0009-0007-7976-7480}}[b]
\author[A.~Muscholl]{Anca Muscholl\lmcsorcid{0000-0002-8214-204X}}[b]
\author[I.~Walukiewicz]{Igor Walukiewicz\lmcsorcid{0000-0001-8952-7201}}[a]

\address{Universit\'e de Bordeaux, CNRS, France}	
\email{hugo.gimbert@labri.fr, igw@labri.fr}  

\address{Universit\'e de Bordeaux, France and MPI-SWS Kaiserslautern, Germany}	
\email{cmascle@mpi-sws.org}  

\address{Universit\'e de Bordeaux, France}	
\email{anca@labri.fr}  





\begin{abstract}
  \noindent We consider the problem of distributed control for systems synchronizing over locks.
  The goal is to find a local controller for each of the processes so
  that global deadlocks of the system are avoided.
  Without restrictions this problem is shown to be undecidable, even for a fixed
  number of processes and locks.
  We identify two restrictions that help to recover decidability.
  The first one is that each process can use at most two locks.
  The control problem is shown to be $\S^P_2$-complete in this case, and even in \PTIME\ under some
  additional assumptions.
  The paradigmatic example of the dining philosophers satisfies these assumptions. 
  The second restriction is the nested usage of locks. 
  In this case the distributed control problem is shown to be \nexpt\-complete.
  The drinking philosophers problem falls in this case. 
\end{abstract}

\maketitle
%


\section{Introduction}

Automatic synthesis of distributed systems has a big potential since
such systems are difficult to
write, test, or verify. 
The state space and the number of different behaviors grow exponentially with
the number of processes. 
This is where distributed synthesis can be more useful than
centralized synthesis, because an equivalent, sequential system may be
too big to handle.
The other important point is that distributed synthesis produces by definition a distributed
system, while a central controller may not be implementable
on a given distributed architecture. 
Unfortunately, very few settings are known for which distributed synthesis is
decidable, and those we know of require very high complexity.

Distributed synthesis was first formulated in a synchronous setting by Pnueli and Rosner~\cite{PR90}.
Subsequent research showed that, essentially, the only decidable architectures are pipelines, where each process can send messages only
to the next process in the pipeline~\cite{KV01,MadThiag01,FinSch05}.
In addition, the complexity is non-elementary in the size of the
pipeline. 
These negative results motivated later a strand of work around distributed
controller synthesis  in the setting of Zielonka automata, 
in particular synthesis with so called causal memory.
Here the problem was shown decidable for co-graph action
alphabets~\cite{GLZ04}, and for tree architectures of processes~\cite{GGMW13,MW14}. 
Yet the complexity can be again non-elementary, e.g.~in the depth of
the tree representing the acyclic architecture.
Worse, it has been  recently established that  distributed
synthesis with causal memory is undecidable for
unconstrained process architectures~\cite{Gimbert22}. 
Distributed synthesis for (safe) Petri nets~\cite{FinkbeinerO17ic} has encountered a similar line of
limited advances, and due to~\cite{Gimbert22}, is undecidable in the
general case, too, since
it is inter-reducible to distributed synthesis for asynchronous
automata~\cite{Finkbeiner19concur}. 
This situation raised the question if there is any natural setting for distributed synthesis that
covers some standard examples of distributed systems, and is manageable
algorithmically. 

In this work we consider distributed systems with a weaker
synchronization mechanism, namely lock sharing.
Here each process can take or release a lock from a  pool of locks.
Locks are a classic concept in
distributed systems, and one of the most frequently used synchronization mechanism in concurrent
programs. 
We formulate our results in a  control setting rather than synthesis
-- this avoids the need for a specification formalism.
The objective is to find a local strategy for each process so that the global system
does not deadlock.
Note that local strategies are purely local: they do not involve any
information exchange as in the case of synthesis with causal memory
(Zielonka automata or Petri nets).
In this sense the synthesis problem resembles the Pnueli and Rosner
framework, but for the asynchronicity of processes.

For unrestricted lock-sharing systems we hit again an undecidability barrier, as for the
models discussed above.
Undecidability was known already for the verification of systems where
each process is modeled as a pushdown automaton~\cite{KahIvaGup05},
since unrestricted usage of locks  allows for inter-process
communication. 
Yet, we are able to find quite interesting restrictions making distributed control
synthesis for lock-sharing systems decidable, and even algorithmically manageable.
The first restriction is to limit the number of locks
available to each process to two.
The standard example is the dining philosophers problem, where each
philosopher has two locks
corresponding to the left and the right fork.
It is important to note that we do \emph{not} limit the \emph{total number} of locks in the system. 
We show that for such systems the complexity of the synthesis problem is at the second level of the polynomial
hierarchy. 
The problem gets even simpler when we restrict the local strategies such
that they cannot block the
process when all locks are available.
We call such strategies \emph{locally live}.
In this case we obtain an \NP-algorithm, and even a \PTIME\
algorithm when the access to locks is \emph{exclusive}.
The latter means that once a process tries to acquire some lock it
cannot switch to  
another action before getting it. In other words, a process that tries to get a lock is blocked as long as the lock is not available.
The second restriction is nested lock usage. 
This is a very common restriction in concurrent
programs~\cite{KahGup06lics}, sometimes enforced
syntactically by associating locks with program blocks.
Nested lock usage simply says that acquiring
and releasing locks should follow a stack discipline.
Verification of concurrent programs with nested locks has
been shown decidable in~\cite{KahIvaGup05,KahGup06lics}, and this
triggered further work on extensions of lock usage
policies~\cite{Kahlon09,BonnetCMV13,LammichMSW13}. 
In distributed computing, the drinking philosophers
setting~\cite{ChaMis84} is an example of nested lock usage.  
We show that in this case the distributed synthesis problem is
\NEXPTIME-complete, where the exponent in 
the algorithm depends only on the number of locks available to the process.
A decision procedure for the verification of such systems, based on
similar ideas on lock orderings, appeared already
in~\cite{KahIvaGup05}.
We study here a more general problem, namely distributed
control. 
Our results are stated for finite-state processes only, in order to
keep the setting simple, but they hold
 for pushdown processes as well.

As mentioned above, we formalize the distributed synthesis problem as a control
problem~\cite{RW89}.
A process is given as a transition graph where transitions can be
local actions, 
or acquire/release of a lock. 
Some transitions are "controllable", and some are not. 
A controller for a process decides which "controllable" transitions to allow, depending on the
local history.
In particular, the controller of a process does not see the states of other processes. 
Our techniques are based on analyzing patterns of taking and releasing locks.
In decidable cases there are finite sets of patterns characterizing
potential deadlocks.

The notion of patterns resembles locking
disciplines~\cite{ErnstLMST2016}, which are commonly
used to prevent deadlocks. 
An example of a locking discipline is ``take the left fork before the
right one'' in the dining philosophers problem.
Our results allow to check if a given locking discipline may result in a deadlock, and
in some cases even list all deadlock-avoiding locking disciplines.

To summarize, the main results of our work are:
\begin{itemize}
	\item $\S^P_2$-completeness of the deadlock avoidance control problem for systems
	where each process has access to at most 2 locks (\tlss\ for short).
	\item An \NP\ algorithm for \tlss\  with locally live strategies.
	\item A \PTIME\ algorithm for \tlss\ with locally live
          strategies and exclusive lock access.
	\item A \NEXPTIME\ algorithm and the matching lower bound for
          lock-sharing systems with nested lock usage.
	\item Undecidability of the deadlock avoidance control problem for 
	systems with unrestricted access to locks (with fixed number
        of  processes and  locks).
\end{itemize}

\paragraph*{Related work}
Distributed synthesis is an old idea motivated by Church's synthesis
problem~\cite{Church57}.
Actually, the logic CTL  has been proposed with distributed synthesis in
mind~\cite{ClaEme81}. 
Given this long history, there are relatively few results on distributed synthesis.
Three main frameworks have been considered: 
synchronous networks of input/output automata, asynchronous automata, Petri
games. 

The synchronous synthesis model has been proposed by Pnueli and
Rosner~\cite{PnuRos89,PR90}.
They established that controller synthesis is decidable  for pipeline
architectures and undecidable in general. The undecidability result holds
for very simple architectures with only two processes.
Subsequent work has shown that in terms of network shape
pipelines are essentially the only decidable case~\cite{KV01,MadThiag01,FinSch05}.
Several ways to circumvent undecidability have been considered. 
One was to restrict to local specifications, specifying the desired behavior of
each automaton in the network separately. 
Unfortunately, this does not extend the class of decidable architectures
substantially~\cite{MadThiag01}.
A furthergoing proposal was to consider only input-output specifications.
A characterization, still very restrictive, of decidable architectures for this
case is given in~\cite{GSZ-fmsd09}. 

The asynchronous (Zielonka) automaton setting was proposed as a reaction to these negative
results~\cite{GLZ04}.
The main hope was that causal memory helps to prevent undecidability
arising from partial information, since the synchronization of
processes in this model makes them share  information.
Causal memory indeed allowed to get new decidable cases:
co-graph action alphabets~\cite{GLZ04}, 
connectedly communicating systems~\cite{MTY05},
and tree architectures~\cite{GGMW13,MW14}.
There is also a weaker condition covering these three cases~\cite{gim17}.  
This line of research suffered however from a very recent result showing undecidability in the
general case~\cite{Gimbert22}. 

Distributed synthesis in the Petri net model, called Petri games, has been
proposed recently in~\cite{FinkbeinerO17ic}.  
The idea is that some tokens are controlled by the system and some by
the environment. 
Once again causal memory is used. 
Without restrictions this model is inter-reducible with the asynchronous
automata model~\cite{Finkbeiner19concur}, hence the undecidability result~\cite{Gimbert22} applies. 
The problem is \EXPTIME-complete for one environment
token and arbitrary many system tokens~\cite{FinkbeinerO17ic}.
This case stays decidable even for global safety specifications, such as
deadlock, but undecidable in general~\cite{FinGieHec22}.
As a way to circumvent the undecidability, bounded synthesis  has been
considered in~\cite{Fin15,HecMet19}, where the bound on the size of the resulting
controller is fixed in advance.
The approach is implemented in the tool \textsc{AdamSYNT}~\cite{GieHecYan21}.

The control formulation of the synthesis problem comes from the control theory
community~\cite{RW89}. 
It does not require to talk about a specification formalism, while retaining
most useful aspects of the problem. 
A frequently considered control objective is avoidance of undesirable states. 
In the distributed context, deadlock avoidance looks like an obvious candidate,
since it is one of the most basic desirable properties. 
The survey~\cite{Wal:21} discusses the relation between the
distributed control problem and  Church synthesis.
Some distributed versions of the control problem have been considered, also
hitting the undecidability barrier very quickly~\cite{RudWon92,Tri04,Thistle05,ArnWal07}.

We would like to mention two further results that do not fit into the main threads outlined
above.
In~\cite{WanLafKel09} the authors consider a different synthesis problem for 
distributed systems: they construct a centralized controller for a scheduler
that would guarantee absence of deadlocks. 
This is a very different approach to deadlock avoidance.
Another recent work~\cite{BerBolBou20} adds a new dimension to distributed
synthesis  by considering communication errors in a model
with synchronous processes that can exchange their causal memory.
The authors show decidability of the synthesis problem for 2 processes.



\paragraph*{Outline of the article}
In the next section we define systems with locks, strategies, and the control
problem. 
We introduce locally live strategies as well as the $2$-lock, exclusive,
and nested locking restrictions. 
This permits to state the main results of the article. 
The following three sections consider systems with the $2$-lock restriction.
First, we briefly give intuitions behind the $\S^p_2$-completeness
in the general case.
Section~\ref{sec:np} presents an \NP\ algorithm  for \tlss\ with locally live strategies.
Section~\ref{sec:exclusive} gives a \PTIME\ algorithm for the exclusive case with locally live strategies.
Next in Section~\ref{sec:nested} we consider systems with nested locks, and show that the problem is
\nexpt-complete in this case.
Finally, in Section~\ref{sec:undec} we prove that without any restrictions the problem is
undecidable.

This paper is an extended version of~\cite{GimbertMMW22}.

To help the reader we use the LaTeX package 
\href{https://ctan.org/pkg/knowledge?lang=en}{\texttt{knowledge}} that
hyperlinks definitions with their usage.


\section{Preliminaries}
\label{sec:prelim}

A \emph{lock-sharing system} is a parallel composition of processes sharing a pool of locks.
Processes do not communicate, but they may acquire  or release locks
from  the pool.
Some transitions of processes are uncontrollable, meaning that the environment
decides if such a transition is taken. 
The goal is to find a local strategy for each process so that the system never
deadlocks.
The challenge is that the strategies are purely local, in the sense that each
process only knows its previous actions.

\AP A \emph{process} $p$ is an automaton $\aut_p = (S_p,
\Sigma_p, T_p, \delta_p, \init_p)$ with a set of locks 
$T_p$  that it can acquire or release.
The transition function $\delta_p: S_p\times \S_p\stackrel{\cdot}{\to} \OpTp\times S_p$ associates
with a state from $S_p$ and an action from $\S_p$ an operation on
some lock and a new state; it is a
partial function. 
The lock operations consist in acquiring ($\get{t}$) or releasing ($\rel{t}$) some lock
$t$ from $T_p$, or doing nothing:
$\OpTp=\set{\get{t}, \rel{t}\mid t\in T_p}\cup\set{\nop}$.
Figure~\ref{fig:philosophers} gives an example.
For simplicity we write action names in our examples only for 
$\nop$, otherwise we just write the lock operation of the action.

\AP A \emph{local configuration} of process $p$ is a state from $S_p$ together
with the locks $p$ currently owns: $(s,B)\in S_p\times 2^{T_p}$. 
The initial configuration of $p$ is $(\init_p,\es)$, namely the
initial state and $p$ owns no  locks.
A transition between two local configurations $(s, B) \xrightarrow{(a, op)}_p (s',
B')$ exists when $\d_p(s,a)=(op,s')$ and one of the following holds:
\begin{itemize}
	\item $op = \nop$ and  $B = B'$;
	
	\item $op=\get{t}$,  $t \notin B$ and $B' = B\cup \set{t}$;
	
	\item $op=\rel{t}$, $t \in B$, and $B'= B \setminus
	\set{t}$.
\end{itemize}
\AP A \emph{local run} $(a_1,op_1)(a_2,op_2) \dots (a_n,op_n) $ of
$\aut_p$ is a finite sequence over
$\S_p\times\OpTp$ such that there exists a sequence of local configurations
$(\init_p,\es)=(s_0,B_0)\act{(a_1,op_1)}_p
(s_1,B_1)\act{(a_2,op_2)}_p\dots (s_n,B_n)$.
While the run is determined by the sequence of actions, we prefer to make
lock operations explicit.
We write  $\Runs_p$ for the set of local runs of $\aut_p$.
\AP We call a local run \intro{neutral} if it starts and ends with the same set of locks.

\AP A \intro{lock-sharing system} (\intro*\lss) $\Ss=((\aut_p)_{p\in\Proc},\S^s,\S^e,T)$ is a set of processes 
together with a partition of actions between \intro{controllable}
actions from $\S^s$ and
\intro{uncontrollable} actions from $\S^e$, and a set $T$ of locks.
We write $T=\bigcup_{p\in\Proc}T_p$, for the set of all locks. 
\kl{Controllable} and "uncontrollable" actions belong to the system and to
the environment, respectively.
We write $\S=\bigcup_{p\in\Proc}\S_p$ for the set of actions of all
processes and require that $(\S^s,\S^e)$ partitions $\S$.
The sets of states and action alphabets of processes are disjoint: $S_{p}\cap S_{q}=\es$
and $\S_{p}\cap \S_{q}=\es$ for all $p\not=q$.
The sets of locks are not disjoint, in general, since processes may
share locks.

\begin{exa}\label{ex:philo}
	The dining philosophers problem can be formulated as a control problem for
	a "lock-sharing system" $\Ss=((\aut_p)_{p\in\Proc},\S^s,\S^e,T)$.
	Let  $\Proc = \set{1,\ldots,n}$  and  $T = \set{t_1, \ldots, t_n}$ as the
	set of locks.
	For every $p\in\Proc$, process $\aut_p$ is as
	in Figure~\ref{fig:philosophers}, with the convention that $t_{n+1}=t_1$.
	Actions in $\S^s$ are marked by dashed arrows. 
	These are "controllable" actions.
	The remaining actions are in $\S^e$.
	Once the environment makes a philosopher $p$ hungry, $p$ has to
	get both the left ($t_p$) and the right ($t_{p+1}$) fork to eat. 
	She may however choose the order in which she takes them; actions $\lft$ and
	$\rgt$ are "controllable".  
	
		\begin{figure}
	\begin{tikzpicture}[xscale=1.4,yscale=1.3,AUT style]
		\node[state,initial, fill=blue!5!white] (0) at (0,0) {};
		\node[state, fill=blue!5!white] (1) at (1.5,0) {};
		\node[state, fill=blue!5!white] (2) at (3,0.4){}; 
		\node[state, fill=blue!5!white] (4) at (6,0) {};
		\node[state, fill=blue!5!white] (5) at (3,-0.4) {};
		\node[state, fill=blue!5!white] (2') at (4.5,0.4) {};
		\node[state, fill=blue!5!white] (5') at (4.5,-0.4){};
		\node[state, fill=blue!5!white] (6) [below of=5] {};

		\path[->] (0) edge node[above] {$\hungry$} (1);
		\path[->,loop above] (0) edge node[above] {$\think$} (0);
		\path[->, dashed] (1) edge node[above] {$\lft$} (2);
		\path[->, dashed] (1) edge node[below] {$\rgt$} (5);
		\path[->] (2') edge node[above] {$\get{t_{p+1}}$} (4);
		\path[->] (5') edge node[below] {$\get{t_{p}}$} (4);
		\path[->] (2) edge node[above] {$\get{t_{p}}$} (2');
		\path[->] (5) edge node[below] {$\get{t_{p+1}}$} (5');
		
		\path[->, bend left=35] (4) edge node[below]  {$\rel{t_p}$} (6);
		\path[->] (6) edge node[below]  {$\rel{t_{p+1}}$} (0);
		
	\end{tikzpicture}
	\caption{A dining philosopher $p$. Dashed transitions are "controllable".}
	\label{fig:philosophers}
\end{figure} 
\end{exa}	

A \emph{global configuration} of $\Ss$ is a tuple of local configurations
$C=(s_p, B_p)_{p\in\Proc}$ provided the sets $B_p$  are pairwise disjoint:
$B_{p}\cap B_{q}=\es$ for $p\not=q$. 
This is because a lock can be taken by at most one process at a time. 
The initial configuration is the tuple of initial configurations of
all processes.

The semantics of such systems is \emph{asynchronous}, as a step of computation is simply defined as one process taking a local transition:
$C\act{(a,op)}C'$ with $C=(s_p, B_p)_{p\in\Proc}$ and $C'=(s'_p, B'_p)_{p\in\Proc}$ if for some process $p$, $(s_p,B_p)\act{(a,op)}_p(s'_p,B'_p)$ and
$(s_q,B_q)=(s'_q,B'_q)$ for every $q\not=p$.
A global run is a sequence of transitions between global
configurations.
Since our systems are deterministic we usually identify a global run by the sequence
of transition labels.
Observe that any action name determines the process that executes it,
since the $\S_p$ are disjoint.
A global run $w$ \emph{determines a local run} of each process: $w|_p$
is the projection of $w$ on $\S_p$.

\AP A \intro{local strategy} $\s_p$ says which actions $p$ can take depending on its local run
so far. 
Moreover, it cannot  block environment actions.
Formally, for every $u\in \Runs_p$ define $\out(u)\subseteq \S_p$ as the set of
actions that are possible after $u$.
Then $\sigma_p : \Runs_p \to 2^{\Sigma_p}$ is such that $\s_p(u) \subseteq \out(u)$
  provided that $(\S^e\cap\out(u))\incl \s_p(u)$.
A \intro{control strategy} for a "lock-sharing system" is a tuple of local strategies, one for each
process: $\sigma = (\sigma_p)_{p \in \Proc}$.

A local run $u$ of $p$ \emph{respects} $\s_p$ if for every non-empty
prefix $v \, (a,op)$
of $u$, we have $a\in\s_p(v)$. 
Observe that local runs are affected only by the local strategy of
that process, there is no inter-process communication.
A global run $w$ respects $\s$ if for every process $p$, the local run $w|_p$
respects $\s_p$. 
We often say just $\s$-run, instead of ``run respecting $\s$''.

As an example consider the system for two philosophers from
Example~\ref{ex:philo}.
Suppose that both local strategies always say to take the $\lft$
transition.
So $\hungry^1,\lft^1,\eget{1}{t_1},\eget1{t_2}$ is a local run of process $1$ respecting
the strategy; similarly $\hungry^2,\lft^2,\eget2{t_2},\eget2{t_1}$ for process $2$.
(We use superscripts to indicate the process doing an action.)
The global run $\hungry^1,\hungry^2,\lft^1,\lft^2,\eget1{t_1},\eget2{t_2}$ respects
the strategy. 
It deadlocks, since  each philosopher needs a lock the other one owns.

\begin{defi}[Deadlock avoidance control problem]
	A $\s$-run $w$ \intro{leads to a deadlock in $\s$} if $w$ cannot be
	prolonged to a $\s$-run.
	A control strategy $\s$ is \intro{winning} if no $\s$-run leads to a deadlock in $\s$.
	The \intro{deadlock avoidance control problem} is to decide if for a given
	system there is some winning control strategy.  
\end{defi}

In this work we consider several variants of the deadlock avoidance control problem. 
Maybe surprisingly, we get more efficient algorithms when we exclude
strategies that can block a process by itself:
\begin{defi}[Locally live strategy]
	\label{def:locallylive}
	\AP A local strategy $\s_p$ for process $p$ is \intro{locally
          live} if every  $\s_p$-run $u$ of $p$
	can be prolonged: there is some $b \in\S_p$ and $op \in \OpTp$ such
        that  $u \, (b,op)$ is
        a $\s_p$-run, too.
	A strategy $\s$ is locally live if each of its associated local strategies
	is so.
\end{defi}
In other words, a "locally live" strategy guarantees that a
process does not block 
if it runs alone according to $\s_p$. 
Back to Example~\ref{ex:philo}:
a strategy always offering one of the $\lft$ or $\rgt$ actions is "locally live". 
A strategy that offers none of the two is not.
Observe that blocking one process after the hungry action is a very efficient strategy to
avoid a deadlock, but it is not the intended one. 
This is why we consider "locally live" to be a desirable property rather than a restriction.

Note that being "locally live" is not exactly equivalent to a strategy
always proposing at least one transition.
This is because with our definition, a process blocks if it tries to  acquire a lock that it already owns, or
to release a lock it does not own.
But it becomes equivalent thanks to the following:

\begin{rem}\label{rem:nice-get-and-rel}
We can assume w.l.o.g.~that \lss\ are ""lock-aware"": by this we mean that
every process knows from its local state which locks it holds,  and it
never tries to  acquire a lock that it
        already owns, or release a lock that it does not own.
  	Note that enforcing lock-awareness does not compromise the complexity
        results when  processes can access only a fixed number of locks.
	We will not use lock-awareness in Section~\ref{sec:nested}, where a
	process can access arbitrarily many locks (in nested fashion).
\end{rem}

Without any restrictions our synthesis problem is undecidable. 
The proof of
the theorem below is in Section~\ref{sec:undec}. 
\begin{restatable}{thm}{undecfourlocks}
	\label{thm:undecfour}
	The "deadlock avoidance control problem" for arbitrary \lss\
        is undecidable (even when the number of locks and processes is
        fixed). 
\end{restatable}

We propose then two interesting cases when the control problem becomes decidable.

In the first case each process accesses at most two different
locks. In the following definition, we require each process to use
exactly two locks, as it is more convenient to avoid case distinctions
on the number of locks used by a process. This is not more restrictive
as we can always add some dummy locks, which are never  used.

\begin{defi}[\tlss]\label{def:tlss}
	\AP A process $\aut_p = (S_p, \Sigma_p, T_p, \delta_p, \init_p)$ \intro{uses two
		locks} if $|T_p|=2$.
	A system  $\Ss=((\aut_p)_{p\in\Proc},\S^s,\S^e,T)$ is a \tlss\ if  every
	process uses two locks.
\end{defi}
Note that in the above definition we do not bound the total number of locks in
the system, just the number of locks per process.
The process from Figure~\ref{fig:philosophers} is a \tlss.
Our first main result says that the control problem is decidable for \tlss.
\begin{restatable}{thm}{twolssSigma}
	\label{thm:2LSSSigma}
	The "deadlock avoidance control problem" for \tlss\ is $\S^p_2$-complete.
\end{restatable}

The second main result says that restricting to "locally live"
strategies helps to obtain a quite tractable case:

\begin{restatable}{thm}{twolssNP}
	\label{thm:2LSSNP}
	The "deadlock avoidance control problem" for \tlss\ is in \NP\ when "strategies" are required to be "locally live".
\end{restatable}
We do not know if the above problem is in \PTIME.
We get a \PTIME\ algorithm under one more assumption:

\begin{defi}[Exclusive systems]\label{def:exclusive}
	A process $p$ is \intro(proc){exclusive}
	if for every state $s \in S_p$: if $s$ has an outgoing
	transition with some $\get{t}$ operation then all outgoing transitions have
	the same $\get{t}$ operation. 
	A system is \intro(sys){exclusive} if all its processes are. 
\end{defi}

\begin{exa}\label{ex:felxible-philo}
	The process from Figure~\ref{fig:philosophers} is "exclusive@@proc", while the one from
	Figure~\ref{fig:flexible-philosophers}  is not. 
	The latter has a state with one $\get{t_{p+1}}$ and one $\rel{t_p}$ outgoing transition. 
	Observe that in this state the process cannot block, and has the possibility to
	take a lock at the same time.
	"Exclusive@@sys" systems do not have such a possibility, so their analysis is much
	easier. 

	\begin{figure}
	\begin{tikzpicture}[xscale=1.4,yscale=1.4,AUT style]
		\node[state, fill=blue!5!white,initial] (0) at (0,0) {};
		\node[state, fill=blue!5!white] (1) at (1.5,0) {};
		\node[state, fill=blue!5!white] (2) at (3,0.4){};
		\node[state, fill=blue!5!white] (2') at (3,-0.4) {};
		\node[state, fill=blue!5!white] (3) at (4.5,0.4) {}; 
		\node[state, fill=blue!5!white] (3') at (4.5,-0.4){};
		\node[state, fill=blue!5!white] (4) at (6,0) {};
		\node[state, fill=blue!5!white] (5) [below of= 2', yshift= -6mm, xshift=1.2cm] {};
		
		\path[->] (0) edge node[above] {$\hungry$} (1);
		\path[->,loop above] (0) edge node[above] {$\think$} (0);
		\path[->, dashed] (1) edge node[above] {$\lft$} (2);
		\path[->, dashed] (1) edge node[below] {$\rgt$} (2');
		\path[->] (2') edge node[above] {$\get{t_{p+1}}$} (3');
		\path[->] (3') edge node[below] {$\get{t_{p}}$} (4);
		\path[->] (2) edge node[above] {$\get{t_{p}}$} (3);
		\path[->] (3) edge node[above] {$\get{t_{p+1}}$} (4);
		
		\path[->, bend right=70] (2) edge node[above] {$\rel{t_p}$} (1);
		\path[->, bend left=70] (2') edge node[below] {$\rel{t_{p+1}}$} (1);
		\path[->, bend left=35] (4) edge node[right, xshift=5] {$\rel{t_p}$} (5);
		\path[->, bend left=20] (5) edge node[below] {$\rel{t_{p+1}}$} (0);
	\end{tikzpicture}
	\caption{A flexible philosopher $p$. She can release a fork if the other fork is not available.}
	\label{fig:flexible-philosophers}
\end{figure} 
\end{exa}

\begin{restatable}{thm}{exclusiveP}
	\label{thm:exclusive}
	The "deadlock avoidance control problem" for "exclusive@@sys" \tlss\ is in \PTIME, when
	"strategies" are required to be "locally live".
\end{restatable}

Without local liveness, the problem  for
"exclusive@@sys" \tlss\ remains $\S^p_2$-hard.

The second case we consider is a common restriction on the usage of locks:
\begin{defi}[Nested-locking]
	\label{def:nested}
	A local run is \intro{nested-locking} if the order of acquiring and
	releasing locks in the run respects a stack discipline, i.e., the only lock a process can release is the last one it acquired. 
	
	\AP A process is ""nested-locking@@proc"" if all its local runs are, and an \lss\ is nested-locking if all its processes are.
\end{defi}

Note that none of the processes in Figures~\ref{fig:philosophers} and
\ref{fig:flexible-philosophers} are  "nested-locking@@proc".
However, both can be made  "nested-locking@@proc" by
remembering in the local state in which order the locks were obtained.
With this information one can easily determine if an \lss\ is
"nested-locking".

\begin{restatable}{thm}{nestedNEXP}
	\label{thm:nested}
	The "deadlock avoidance control problem" for "nested-locking" \lss\ is \NEXPTIME-complete. 
\end{restatable}


\section{Two locks per process}
\label{sec:two}

We describe how to solve the "deadlock avoidance control problem" for
\tlss, so for systems where every process uses at most two locks.
We present the three results announced in the previous section, namely, 
Theorems~\ref{thm:2LSSSigma},~\ref{thm:2LSSNP}, and~\ref{thm:exclusive}.

The general case, treated in Theorem~\ref{thm:2LSSSigma}, puts no restriction on
"strategies" or on the system, besides being a \tlss.
The main idea is that each "winning strategy" can be decomposed into
"local strategies", each summarized by an object of polynomial size,
called its behavior.
We show that from a computational complexity perspective we cannot do better
than guessing these behaviors to solve the problem. 

The next case is when we require "strategies" to be "locally live". 
With such "strategies", a process can only block if all
locks it asks for are taken forever. 
This simplifies the analysis and enables us to reason on a graph
because of the two-locks restriction.

Finally, we consider the restriction of the deadlock avoidance problem to \emph{"exclusive@@sys"} systems,
still with "locally live" "strategies".
Here, whenever a process can execute an action acquiring a lock it is the only
thing it can do. 
This means that a process gets blocked whenever it tries to get a
certain lock that  is not available. 
Recall that the system in Figure~\ref{fig:flexible-philosophers} is
not "exclusive@@sys", whereas the one in Figure~\ref{fig:philosophers}
is so. 

Throughout this section we fix a \tlss\ $\Ss=((\aut_p)_{p\in\Proc},\S^s,\S^e, T)$ over the
set of processes $\Proc$. 
We also assume that the \tlss\ is "lock-aware" (cf.~Remark~\ref{rem:nice-get-and-rel}).
We also fix a "control strategy" $\s=(\s_p)_{p \in \Proc}$.

The three following subsections present the three cases.

\subsection{The general case of \tlss}

We will use summaries of local runs through so-called 
\emph{patterns}, that describe the most recent lock operations.
We will see later that this information is sufficient to decide if the
"strategy" is "winning" (Lemma~\ref{lem:patterncharacterization}).
Informally, a \emph{pattern of a local run} of process $p$ in a \tlss\ describes which of
the 
four following situations are possible for $p$ at the end of its run:
\begin{itemize}
	\item $p$ owns both locks; 
	
	\item $p$ owns no lock;
	
	\item $p$ owns exactly one of its locks, say $t$, and either
	\begin{itemize}
		\item its last operation on locks was $\get{t}$; or
		\item the last operation on locks was $\rel{t'}$ with
                  $t \not= t'$.
	\end{itemize}
\end{itemize}

Before defining patterns formally we introduce the runs for which we
need them, which are runs that lead potentially to deadlocks:
\begin{defi}[Risky run]
  Consider a local $\s$-run $u$ of a process $p$.
  We say that  $u$ is $\s$-\intro{risky}
  if after executing $u$ all transitions allowed by $\s$ are $\get{}$
  transitions\footnote{A particular case is where after $u$ no
    transitions are possible at all.}.
  We simply write risky when $\s$ is clear from the context. 

\AP   We write $\intro*\Owns_{p,\s}(u)$ for the set of locks owned by $p$ after $u$, or
  simply $\Owns_{p}(u)$ when $\s$ is clear from context. 
  We write  $\intro*\Blocks_{p,\s}(u)=\set{t : \get{t}\in\s_p(u)}$, or 
  simply $\Blocks_p(u)$ when $\s$ is clear from context. 
\end{defi}
Note that if a $\s$-run $u$ is "risky" and the "strategy" $\s$ is
"locally live", then 
$\Blocks_p(u) \not=\es$; if $\s$ is not "locally live" then
$\Blocks_p(u)$ can be empty.
If the run is not "risky" then the process can do some local
action or a release action.

We can now define patterns formally.

\begin{defi}[Patterns]
	\label{def:patterns2locks}
	Consider a \textbf{"risky"} local $\s$-run $u$ of process $p$. 
        We say that $u$ has a \intro{strong pattern}
        $\Owns_p(u)\intro*\spat \Blocks_p(u)$  if $\Owns_p(u)\neq\es$ and the last
        operation on locks in $u$ is a release. Otherwise we say  that $u$ has a
        \intro{weak pattern} $\Owns_p(u) \intro*\wpat \Blocks_p(u)$.
	We also  write $\Owns_p(u)\pat\Blocks_p(u)$ if we do not  specify
        if a pattern is "strong@@pat" or "weak@@pat".
	
	\AP We say that $\s$ \intro{admits a pattern} $\Owns_p \pat \Blocks_p$ ($\Owns_p
        \spat \Blocks_p$, $\Owns_p \wpat \Blocks_p$, resp.)
   	for process $p$ if there exists some "risky" $\s$-run $u$ of $p$ with $\Owns_p=\Owns_p(u)$,
        $\Blocks_p=\Blocks_p(u)$ and this kind of
    pattern  (strong, weak, resp.). 
     
	We write $\PP^\s_p$ for the set of patterns for $p$  admitted
        by $\s$. 
	We write $\PP^\s=(\PP^\s_p)_{p\in\Proc}$ 
    and denote $\PP^\s$ as the \intro(two){behavior} of $\s$.
\end{defi}

We will refer to "patterns" of process $p$ as $\Owns_p \pat \Blocks_p$,
in order to stress the name of the process, and
we always assume that $\Owns_p \cap\Blocks_p=\es$.
Since in a  \tlss\ any process uses two locks, a "strong pattern"
$\Owns_p \spat \Blocks_p$ for
$p$ is such that
$\Owns_p=\set{t_1}$, and $\Blocks_p$ is either $\set{t_2}$ or $\es$, where $t_1,t_2$
are the two locks used by $p$.
For example, the \tlss\ in Figures~\ref{fig:philosophers} and
\ref{fig:flexible-philosophers}  admit only "weak patterns".
Consider now the \tlss\ in Figure~\ref{fig:S2hardness}.
If the strategy of process $p_i$ takes only the lower branch, then its
"patterns" are $\es \wpat \set{x_i}$, $\set{x_i} \wpat
\set{\overline{x_i}}$ and $\set{x_i} \spat \es$.
If the strategy allows both branches then we add another "strong
pattern", $\set{\overline{x_i}} \spat \es$.

The next lemma characterizes winning strategies in terms of
patterns.

\begin{lem}
\label{lem:patterncharacterization}
Let $\s = (\s_p)_{p \in \Proc}$ be a "strategy" and
$\PP^\s=(\PP^\s_p)$ its "behavior@@two".
Then $\sigma$ is \textbf{not} "winning" if and only if for every $p$ there is
some "pattern" $\Owns_p \pat \Blocks_p$ in $\PP^\s_p$ such that all
conditions below hold:
\begin{itemize}
\item $\bigcup_{p \in \Proc} \Blocks_p \subseteq \bigcup_{p \in \Proc} \Owns_p$,
		
\item the sets $\Owns_p$ are pairwise disjoint,
		
\item there exists a total order $<$ on $T$ such that for all $p$,
  if $p$ admits a "strong pattern" $\set{t}\spat \Blocks_p$ then
  $t<t'$, where $t'$ is the other lock used by $p$.
\end{itemize}
\end{lem}

\begin{proof}
	Suppose that $\s$ is not "winning", let $\run$ be a global $\s$-run ending in a "deadlock", 
	and for each process $p$ let $\run_p$ be the corresponding local run.
	
	For every $p$, the local run $\run_p$ has to be "risky", otherwise $\run_p$
	could be extended into a longer run consistent with $\s$. 
	Thus $\run_p$ has a "pattern" 
	$\Owns_p \pat \Blocks_p$ in $\PP^\s_p$. 
	
	We check that these "patterns" meet all requirements of the lemma.
	Clearly as we are in a "deadlock", the only actions available
        to each process acquire locks that are already
	taken, hence the first condition is satisfied.
	Furthermore, no two processes can own the same lock, implying the second
	condition.
	Finally, let $<$ be a total order on locks compatible with the order in $u$
	between the last operation on each lock, that is: $t < t'$
	if the last operation on $t$ in $\run$ is before the last one on $t'$. 
	If one of $t, t'$ is untouched throughout the run then the order is taken arbitrarily.
	
	Consider a process $p$ using locks $t,t'$ and such that $u_p$ has a "strong@@pat" "pattern" $\set{t}\spat \Blocks_p$.
	So $\run_p$ is of the form $\run_1 (a, \get{t}) \run_2 (b,\rel{t'}) \run_3$ 
	with no action on $t$ in $\run_2$ or $\run_3$. 
	Hence $t < t'$ since the last action on $t$ is before the last action on $t'$.

	We now prove the  other direction of the lemma.
	Suppose that for each $p$ there is a "pattern" $\Owns_p \pat \Blocks_p$ in $\PP^\s_p$ such that 
	those "patterns" satisfy all three conditions of the lemma. 
	Let $<$ be a total order on locks witnessing the third condition.
	
	By definition, for all $p$ there exists a "risky" local run
        $\run_p$ with $\Owns_p=\Owns_p(u_p)$ and $\Blocks_p=\Blocks_p(u_p)$.
	We show now the existence of a global run $u$ with $u_p =
        u|_p$ for every $p \in \Proc$.
	We start by executing 	one by one, in some arbitrary order, all the $\run_p$
	such that $\Owns_p = \emptyset$.
	After executing each such run, all locks are free, hence we can execute the
	next one. 
	At the end all locks are still free.
	
	For all $p$ such that $\Owns_p = \set{t}$ and $\Owns_p \wpat \Blocks_p$ is
	weak, we can write $\run_p$ as $\run^p_1 (a, \get{t}) \run^p_2$ with $\run^p_1$ "neutral" and $\run^p_2$ not containing any operation on locks. 
	We can execute $\run^p_1$, which again leaves all locks free as it is "neutral".
	
	Next we consider all the processes $p$ where $u_p$ has a "strong
	pattern" $\set{t_p}\spat \Blocks_p$.
	We execute all runs $u_p$ according to the order $<$. 
	This is possible, as for each such $p$ we have $t_p < t'_p$,
        where $t'_p$ is the other lock used by $p$.
	The order $<$ guarantees that before executing $\run_p$ all
        locks $t \ge t_p$ are free.
	In particular, since $t_p$ and $t'_p$ are free, we can execute
        $\run_p$. 
	
	At this point all locks are free except for locks
        $t_p$ of processes $p$ with a "strong
	pattern" $\set{t_p}\spat \Blocks_p$.
	We now come back to the $u_p$ with "weak patterns". 
	We execute the remaining parts of $u_p$, namely $(a,\get{t})\run^p_2$
	as above. As $\run_2^p$ contains no operation on locks, we only need $t$ to be free to execute this run.
	As all $\Owns_q$ are disjoint, and all locks taken at that point belong to some other $\Owns_q$, $t$ is free, hence all such runs can be executed.
	
	Finally, the remaining runs $u_p$ are the ones such that 
	$\Owns_p = \set{t,t'}$ contains both locks of $p$. 
	As all $\Owns_p$ are disjoint, both $t,t'$ are free, hence
        $\run_p$ can be executed. 
	
	We have executed all local runs, therefore we reach a configuration where
	all processes need some lock from 
	$\bigcup_{p \in \Proc} \Blocks_p$ to keep running, and all locks in 
	$\bigcup_{p \in \Proc} \Owns_p$ are taken.
	As $\bigcup_{p \in \Proc} \Blocks_p \subseteq \bigcup_{p \in \Proc} \Owns_p$, we have reached a "deadlock".
\end{proof}

Thanks to Lemma~\ref{lem:patterncharacterization}, in order to decide if there is a "winning strategy" for a
given system it is enough to come up with a set of "patterns" $\Pat_p$ 
for each process $p$ and show two properties:
\begin{itemize}
	\item there exists a "strategy" $\s$ such that
          $\PP^\s_p \subseteq \Pat_p$ for each process $p$; 
	\item the sets of "patterns" $\Pat_p$ do not meet the conditions given by
	Lemma~\ref{lem:patterncharacterization}.
\end{itemize}

Note that in the first condition we only require an inclusion because by the
previous lemma, the less "patterns" a "strategy" 
allows, the less likely it is to create a "deadlock".

We start by showing that given a set of "patterns" for each process,
we can check the first condition in polynomial time.

\begin{lem}
	\label{lem:polycheckpatterns}
	Given a "behavior@@two" $(\PP_p)_{p \in \Proc}$,
	it is decidable in \PTIME\ whether there exists a "strategy" 
	$\s$ such that for every $p$ we have $\PP^\s_p
        \subseteq \Pat_p$. 
\end{lem}

\begin{proof}
	First of all recall that we only need to check for each $p$
        that there exists a "local strategy" $\s_p$ that does not
        allow any "risky" run of $p$ with "pattern" not in $\Pat_p$. 
	
	Let $p \in \Proc$ and $\aut_p=(S_p,\S_p,T_p,\delta_p,
        \init_p)$ be its transition system.
        Recall that we assume that  $\aut_p$ is "lock-aware".
        We can do a bit more: in a  state where $p$ owns lock $t_1$,
        we store an additional  
	bit of information saying whether $p$ released its other lock $t_2$
	since the last acquisition of $t_1$. 
	This way, the "risky" nature of a local run and its
        "pattern" depend  only on 
	the state in which the run ends and the outgoing transitions. 
	For instance if a state has no outgoing transitions and is such that
	when reaching it $p$ holds $t_1$ and released $t_2$ since acquiring it,
	then the "pattern" of runs ending there is 
	$\set{t_1} \spat \emptyset$.

        A local state is called bad if all its outgoing transitions have
        acquire operations, and there is no subset of outgoing
        transitions that includes all "uncontrollable" such
        transitions and that yields only "patterns" in $\Pat_p$.
        Otherwise, the state is called good.

        Clearly, a "strategy" $\s$ satisfies $\PP^\s_p
        \subseteq \Pat_p$ iff all states reached by $\s_p$-runs are good.

	To know whether there exists a local strategy $\s_p$ such that
        all its "patterns" are in $\Pat_p$ we proceed as follows.
        We iteratively delete bad states and all their ingoing 
	transitions. If one of those transitions is "uncontrollable"
	we declare its source state as bad (as reaching that state would allow the 
	environment to take that transition, leading us to a bad state).
	Note that deleting transitions may create more bad states by reducing the 
	choice of the system.
	If we end up deleting $\init_p$, we conclude that there is no suitable "local strategy". 
	Otherwise the subsystem we obtain  has only good states,
       and it corresponds to a "strategy" $\s_p$ as desired.
\end{proof}

\begin{prop}\label{p:SigmaP}
The "deadlock avoidance control problem" for \tlss\ is decidable in
\sigmatwo.
\end{prop}

\begin{proof}
	The algorithm first guesses a set of "patterns"
	$\Pat_p$ for each   process $p$.
        Note that the overall size of $\Pat$ is polynomial in $|\Proc|$.
  By Lemma~\ref{lem:polycheckpatterns}, we can then check in polynomial time if
  there exists a "strategy" $\s=(\s_p)_{p\in \Proc}$ with $\s_p$ admitting only
  "patterns" in $\Pat_p$.
  By Lemma~\ref{lem:patterncharacterization} we can determine in \coNP\
  whether $\s$ is winning.

	For the correctness of the algorithm observe that if there exists a "winning
	strategy" $\s$ then it suffices to guess its "behavior@@two" $\Pat^\s$.
  Conversely suppose the algorithm  guessed a "behavior@@two" not meeting the requirements of
  Lemma~\ref{lem:patterncharacterization}.
	Then the "strategy" obtained by Lemma~\ref{lem:polycheckpatterns} is "winning".
\end{proof}

\twolssSigma*
 
\begin{proof}
	The upper bound follows immediately from Proposition~\ref{p:SigmaP}.
	
	For the lower bound we reduce from $\exists\forall$-SAT. 
	Suppose that we are given a formula in $3$-disjunctive normal form
        $\bigvee_{k=1}^s \a_k$, so
	 each $\a_k$ is a conjunction of three literals $\ell^k_1 \land  \ell^k_2
	\land \ell^k_3$ over a set of variables $\set{x_1, \ldots, x_n, y_1, \ldots,
	y_m}$. 
	The question is whether the formula 
	$\phi=\exists x_1 \ldots \exists x_n \forall y_1 \ldots \forall y_m, \bigvee_{k=1}^s \a_k$ is 
	true.
	
	We construct a \tlss\ for which there is a "winning strategy" iff the formula is
	true. 
	The \tlss\ will use locks:
	\begin{equation*}
		\set{t_k \mid 1 \leq k \leq s} \cup
		\set{x_i, \bar{x_i} \mid 1 \leq i \leq n} \cup \set{y_j, \bar{y_j} \mid 1 \leq
		j \leq m}\ .
	\end{equation*}

	For each $1 \leq i \leq n$ we have a process $p_i$ for each
        existentially quantified variable, 
	as depicted in \cref{fig:S2hardness}. 
	In that process the system has to take both $x_i$ and $\bar{x_i}$, and then may release one of them before being blocked in a state with no outgoing transitions. 
	Similarly, for each universally quantified variable we have a
        process $q_j$,  $1 \leq j \leq m$, in which the
	environment has to take $y_j$ or $\bar{y_j}$, and then it blocks.
	
	For each clause $\a_k$ we have a process $p(\a_k)$ which just has one 
	transition acquiring lock $t_k$ towards a state with a local loop on it.
	Hence to block all those processes the environment needs to have all $t_k$
	taken by other processes. 
	
	The environment can block all processes $p(\a_k)$  with the last type of processes.
	For each clause $\a_k$ and each literal $\ell$ of $\a_k$ there is a process
	$p(\a_k,\ell)$.
	There the process has to acquire $t_k$ and then $\ell$ before entering a state with a self-loop.
	In order to block all processes $p(\a_k)$, each $t_k$ has to be taken by a process $p(\a_k,\ell)$ 
	for some literal $\ell$ of $\a_k$.
        For  process $p(\a_k,\ell)$ to be
        blocked, 
	lock $\ell$ has to be taken before, by some $p_i$ or $q_j$.
	
	A "strategy" for the system amounts to choosing whether $p_i$ should release $x_i$
	or $\bar{x}_i$, for each $i=1,\dots,n$. It may also choose to release neither.
	Since the environment has a global view of the system, it can afterwards
	choose one of $y_j, \bar{y_j}$ in process $q_j$, for each $j=1,\dots,m$.
	Those choices represent a valuation, a lock remaining free
        corresponds to the  literal being true.

        If the formula $\phi$ is true, then the system chooses the
        valuation of the $x_i$'s in order to make $\phi$ true.
        As soon as processes $p_i,q_j$ have reached their final state,
        we  also have a valuation for the $y_j$'s.
        At this point there is at least one clause $\a_k$ true, so with
        all its literals $\ell^k_1,\ell^k_2,\ell^k_3$ true.
        Observe that among the 4 processes $p(\a_k)$ and
        $p(\a_k,\ell^k_1)$, $p(\a_k,\ell^k_2)$, $p(\a_k,\ell^k_3)$  at least one
       can reach its self-loop, namely the one that acquires $t_k$
       first.
       Hence, the system does not deadlock.
       Note also that no winning strategy here can be "locally live",
       because of processes $p_i$ and $q_j$.

       Otherwise, if the formula $\phi$ is not true, then for each
       choice of the system for the $x_i$'s, the environment can chose
       afterwards a suitable valuation of the $y_j$'s that falsifies
       $\phi$ (``afterwards'' means that we look at a suitable
       scheduling of the acquire actions).
       For such a valuation, for every $\a_k$ there is some literal
       $\ell^k$ of $\a_k$ that is false.
       Consider the scheduling that lets $p(\a_k,\ell^k)$ acquire $t_k$
       first.
       Since $t_k$ is taken, this implies that $p(\a_k,\ell^k)$ is
       blocked.
       Also, $p(\a_k)$ is blocked because of $t_k$.
       The other two processes $p(\a_k,\ell)$ with $\ell\not=\ell^k$ are
       also blocked because of $t_k$.
       So overall the entire system is blocked.
   \end{proof} 
       
        		\begin{figure}
	\begin{tikzpicture}[xscale=2,yscale=1.3,AUT style]
		\node[state,initial, fill=blue!5!white] (0) at (0,0) {};
		\node[state, fill=blue!5!white] (01) at (0.7,0) {};
		\node[state, fill=blue!5!white] (02) at (1.4,0) {};
		\node[state, fill=blue!5!white] (1) at (2.1,0.35) {};
		\node[state, fill=blue!5!white] (2) at (2.1,-0.35) {};
		\node[state, fill=blue!5!white] (3) at (3,0) {};
		\node (xi) at (2.1,1.2) {\large $p_i$};
		
		\path[->,dashed] (02) edge (1);
		\path[->,dashed] (02) edge (2);
		\path[->] (0) edge node[above] {$\get{x_i}$} (01);
		\path[->] (01) edge node[above] {$\get{\bar{x_i}}$} (02);
		\path[->] (1) edge node[above] {$\rel{x_i}$} (3);
		\path[->] (2) edge node[below] {$\rel{\bar{x_i}}$} (3);
		
		\node[state,initial, fill=blue!5!white] (0) at (3.8,0) {};
		\node[state, fill=blue!5!white] (1) at (4.4,0.35) {};
		\node[state, fill=blue!5!white] (2) at (4.4,-0.35) {};
		\node[state, fill=blue!5!white] (3) at (5.3,0) {};
		\node (yj) at (4.8,1.2) {\large$q_j$};
		
		\path[->] (0) edge (1);
		\path[->] (0) edge (2);
		\path[->] (1) edge node[above] {$\get{y_j}$} (3);
		\path[->] (2) edge node[below] {$\get{\bar{y_j}}$} (3);
		
		\node[state,initial, fill=blue!5!white] (0) at (3,-2.2) {};
		\node[state, fill=blue!5!white] (1) at (4,-2.2) {};
		\node (pak) at (4.6,-1.2) {$p(\a_k)$};
		
		\path[->] (0) edge node[above] {$\get{t_k}$} (1);
		\path[->, loop above] (1) edge (1);
		
		\node[state,initial, fill=blue!5!white] (0) at (0,-2.2) {};
		\node[state, fill=blue!5!white] (1) at (0.8,-2.2) {};
		\node[state, fill=blue!5!white] (3) at (1.6,-2.2) {};
		\node (pakl) at (1.8,-1.2) {\large$p(\a_k,\ell)$};
		
		\path[->] (0) edge node[above] {$\get{t_k}$} (1);
		\path[->] (1) edge node[above] {$\get{\ell}$} (3);
		\path[->, loop above] (3) edge (3);
	\end{tikzpicture}
	\caption{Processes used in Theorem~\ref{thm:2LSSSigma}. 
		Transitions of the system are dashed. All unlabeled
		transitions carry $\nop$
		as operation. Processes $p_i$ and $q_j$ handle
		existentially and universally
		quantified variables, resp.; process $p(\a_k,\ell)$ handles
		literal $\ell$ in clause $\a_k$, and process $p(\a_k)$
		handles clause $\a_k$.}
	\label{fig:S2hardness}
\end{figure} 
 
\subsection{Locally live strategies}
\label{sec:np}

We now consider the case of \tlss\ with "locally live" "strategies". 
Such a strategy ensures that no process blocks when running alone.
Hence a process can only block if all its available transitions need to acquire a
lock, but all these locks are taken.
This restriction prevents a construction like the one used to obtain the lower bound of
Theorem~\ref{thm:2LSSSigma}.

In the last subsection we were guessing a "behavior@@two" of a strategy and then
checking in \coNP\ if the condition from Lemma~\ref{lem:patterncharacterization}
does not hold. 
Here we show that this check can be done in \PTIME.

The argument is  quite lengthy and requires a precise analysis of the
graph representing the guessed "behavior@@two".
We represent a "behavior@@two" as a lock graph $\Gp$, with vertices
corresponding to locks and edges to patterns. 
Then, thanks to local liveness, instead of
Lemma~\ref{lem:patterncharacterization} we get Lemma~\ref{lem:winningiffnods}
characterizing when a strategy is not winning by the existence of a subgraph of 
$\Gp$, called (full) deadlock scheme. 
The main body of the proof is a polynomial time algorithm to decide
the existence
of full deadlock schemes. 

As we are in a "locally live" framework, some "patterns" of 
local runs are impossible. 
We do not have patterns of the form $\Owns_p\to \es$ as a local run can block only
because it requires some locks that are taken.
This leaves  two possible types of patterns, $\set{t_1}\to\set{t_2}$
and $\es\to \Blocks_p$ for some non-empty $\Blocks_p \subseteq \set{t_1,t_2}$.
The set of patterns of the first type defines a graph: 
an edge labeled by $p$ from $t_1$ to $t_2$ represents the pattern
$\set{t_1}\to\set{t_2}$ of  process $p$.
Recall that this corresponds to a local run ending in a situation when
$p$ holds $t_1$ and all actions need to acquire $t_2$.
The second type of patterns will be incorporated later in form of
fragile processes.

We define \emph{weak} and \emph{strong} patterns and cycles, as well as \emph{solid} and \emph{fragile} processes.
We are from the point of view of the controller: we want to obtain strong patterns and solid processes, as they make deadlocks less likely.

\begin{defi}[Lock graph $\intro*\Gp$]\label{def:graph}
	For a "behavior@@two" $\PP=(\PP_p)_\pproc$, we define a labeled graph
	$G_\PP=\struct{T,E_\PP}$, called  \intro{lock graph}, whose
        nodes are locks and with two types of edges, weak or strong.
        Edges are labeled by processes.
	
	\AP There is a ""weak edge"" $t_1 \intro*\weakedge{p} t_2$ in $\Gp$ whenever there is a "weak
	pattern" $\set{t_1} \wpat \set{t_2}$ in $\PP_p$.
	There is a ""strong edge"" $t_1 \intro*\strongedge{p} t_2$ whenever there is a "strong pattern"
$\set{t_1} \spat \set{t_2}$ in $\PP_p$ \textbf{and} there is no "weak pattern" $\set{t_1}
\wpat \set{t_2}$ in $\PP_p$.
	We write  $t_1 \act{p} t_2$ when the type of the edge is
	irrelevant.	
	
	\AP A path (resp.~cycle) in $\Gp$ is called \intro{simple} if all its edges are labeled by different 
	processes.
	A cycle is \intro(cycle){weak} if it contains some "weak edge", and \intro(cycle){strong} 
	otherwise.
\end{defi}

The next definition provides some notions for patterns of
the form $\es \pat \Blocks_p$.

\begin{defi}[solid/fragile]
	For a "behavior@@two" $\PP=(\PP_p)_\pproc$, 
	a process $p$ is called \intro{solid in $\PP$} (or just
        solid, if $\PP$ is clear from the context) if there is no
        "pattern" of the form $\es \pat \Blocks_p$  in 
	$\PP_p$; otherwise it is called \intro{fragile in $\PP$} (or just
        fragile).
	
	\AP A process $p$ is called $Z$\intro{-fragile} if there is some "pattern" $\es\to B$
	in $\PP_p$ with $B\incl Z$.
	Note that a process is "fragile" if and only if it is
        $Z$"-fragile" for some $Z \subseteq T$. 
	
	\AP A \intro{solid edge} of $\Gp$ is one that is labeled by a "solid" process. A \intro{solid cycle} is one that only has "solid edges".
\end{defi}

What the previous definition says is that a
"solid" process needs to take a lock to be blocked, whereas a
"fragile" one can be
blocked without owning a lock.
So we take into account only "solid" processes in the deadlock schemes
defined next: 

\begin{defi}[$Z$-deadlock scheme]\label{def:z-ds}
	Consider a "behavior@@two" $\PP=(\PP_p)_\pproc$,  and  the associated "lock graph" $\Gp$.
	Let $Z \incl T$ be a set of locks.
	We set $\intro*\PZ$ as the set of those processes that can
        access only locks in 	$Z$. 
	
	\AP A \intro{$Z$-deadlock scheme for $\PP$} is a partial function 
        $\ds_{Z} : \PZ \stackrel{\cdot}{\to} E_{G_\PP}$ 
        such that all conditions below are satisfied:
	\begin{enumerate}
		\item For all $p \in \PZ$, if $\ds_{Z}(p)$ is defined
                  then it is a $p$-labeled edge of $\Gp$.
		
		\item If $p \in \PZ$ is "solid" then $\ds_{Z}(p)$ is defined.
		
		\item For all $t \in Z$ there exists a unique $p \in \PZ$ such that
		$\ds_{Z}(p)$ is an outgoing edge of $t$.
		
		\item The subgraph of $\Gp$ restricted to $\ds_{Z}(\PZ)$
		does not contain any "strong cycle".
              \end{enumerate}
              A ""deadlock scheme"" for $\PP$ is a "$Z$-deadlock
              scheme" for $\PP$ for some set $Z$.
\end{defi}

The idea underlying the previous definition is that a $Z$-"deadlock scheme"
witnesses a way to reach a configuration in which all locks of $Z$ are
taken, and all processes from $\PZ$  are blocked.
Each "solid" process from $\PZ$ is mapped to an edge telling which lock it
holds in the "deadlock" configuration and which one it needs in order to
advance.

For every lock in $Z$ there is a unique outgoing edge in $\ds_Z$,
corresponding to the process owning that lock.
Note that this implies that the subgraph induced by $\ds_Z$ is a union
of cycles, with some non-branching paths going into these cycles.

The fourth condition excluding "strong cycles" is required as to be able
to schedule the local runs according to the edges of the "deadlock
scheme" into a global run.

A $Z$-"deadlock scheme" is not a full witness for "deadlock" because
"fragile" processes are missing.
The next definition takes care of "fragile" processes.
Note that $\df_Z(p)$ is always undefined if $p \notin \Proc_Z$.

\begin{defi}[Full deadlock scheme]
	\label{def:sufficientdeadlockscheme}
	A \intro{full $Z$-deadlock scheme} for a "behavior@@two" $\PP$ is
	a $Z$-"deadlock scheme" $\df_Z$ for $\PP$ for some $Z \subseteq
        T$ such that for every process $p \in
	\Proc$ either $\df_Z(p)$ is defined, or $p$ is 
	$Z$"-fragile".
        A ""full deadlock scheme"" for $\PP$ is a "full $Z$-deadlock
        scheme" for $\PP$, for some set $Z \subseteq T$.
\end{defi}

We now prove an analogous result to
Lemma~\ref{lem:patterncharacterization}: a "strategy" is not "winning" if
and only if its lock graph admits a "full deadlock scheme".  
The existence of a "winning strategy" will be established by
non-deterministically guessing a "behavior@@two", verifying that
there exists a "strategy" respecting it, computing the corresponding
"lock graph" and then checking that it has no  "full deadlock
scheme".
The most involved step is the last one.

\begin{lem}
	\label{lem:winningiffnods}
	Consider a "locally live" "control strategy" $\sigma$ and $\PP^\s=(\PP^\s_p)_\pproc$ the "behavior@@two" of $\s$.
The 	"strategy" $\s$ is \textbf{not} "winning" if and only if there is a "full deadlock scheme"
	for $\PP^\s$.
\end{lem}

\begin{proof}
	Suppose that $\sigma$ is not "winning".
	Then by Lemma~\ref{lem:patterncharacterization}, there exist "patterns"
	$\Owns_p\pat \Blocks_p \in \PP^\s_p$, one  for each $p$, such that:	 
	
	\begin{itemize}
		\item $\bigcup_{p \in \Proc} \Blocks_p \subseteq \bigcup_{p \in \Proc} \Owns_p$,
		
		\item the sets $\Owns_p$ are pairwise disjoint,
		
		\item there exists a total order $\leq$ on $T$ such that for all $p$, if
		$\Owns_p \pat \Blocks_p$ is a "strong pattern" of the
                form $\set{t} \spat \Blocks_p$ then $t\leq
		t'$ where $t,t'$ are the two locks used by $p$.
	\end{itemize}
	
	Let $Z = \bigcup_{p \in \Proc} \Owns_p$.
	For every process $p \in \PZ$, define $\df_Z(p)$ as $t_1 \act{p}
        t_2$ if $\Owns_p=\set{t_1}$ and $\Blocks_p=\set{t_2}$.
Note that $\df_Z(p)$ is undefined if $\Owns_p=\es$.

Moreover, there are no other possible cases above, as $\s$ is "locally live" and thus
	$\Blocks_p$ cannot be empty. 
	
	We show that $\ds_Z$ is a "full $Z$-deadlock scheme" for $\PP^\s$ by checking
	the four conditions from Definition~\ref{def:z-ds}. The first condition holds by definition of $\ds_Z$.
	For the second condition let  $p \in \PZ$ and suppose $p$ is "solid".
Thus,  $\Owns_p$ is not empty, hence $\df(p)$ is defined.
	For the third condition let $t \in Z$.
As $Z$ is the disjoint union of the sets $\Owns_p$ there
	exists a unique $p\in \PZ$ such that $t \in \Owns_p$, so a
        unique edge $\df(p)$ outgoing from $t$. 
	For the last condition note that for all "strong edges"
        $t \strongedge{p} t'$ the pattern $\Owns_{p} \spat \Blocks_p$ must be strong as well, hence $t \leq t'$. 
As $\leq $ is a total order on locks, there cannot be any "strong cycle".
	
	Finally, suppose that $p \notin \PZ$ or $\df(p)$ is undefined.
In both cases $\Owns_p = \es$, thus $p$ is $\Blocks_p$"-fragile", and
hence $Z$"-fragile" as $\Blocks_p \subseteq Z$.  
	As a consequence, $\df$ is a "full $Z$-deadlock scheme" for $\PP^\s$. 
	
	For the other direction, suppose we have a "full $Z$-deadlock
        scheme" $\df$ for $\PP^\s$, for some set $Z$ of locks.
	For each process $p\in\Proc$ we can find a 
	"pattern" $\Owns_p \pat \Blocks_p \in \PP^\s_p$ as follows:
	\begin{itemize}
		\item  If $\df(p)$ is undefined or $p \notin \PZ$ then $p$ is $Z$"-fragile".
		In this case we choose $\Blocks_p \subseteq Z$ such
                that $\es \pat \Blocks_p \in \PP^\s_p$ and set $\Owns_p=\es$. 
		
		\item If $\df(p)= t_1 \act{p} t_2$
		then there exists a "pattern" $\set{t_1} \pat
                \set{t_2} \in \PP^\s_p$ with $\set{t_1,t_2}
                \subseteq Z$.
                We set $\Owns_p=\set{t_1}$ and $\Blocks_p=\set{t_2}$.
	\end{itemize}

	We check now the conditions of
        Lemma~\ref{lem:patterncharacterization}.
        
	As all locks of $Z$ have exactly one outgoing edge in
        $\df(\PZ)$, and as all $\Owns_p$ with $p \notin \PZ$ or
        $\df(p)$ undefined are empty, the sets $\Owns_p$ are pairwise
        disjoint.
        Moreover,  $\bigcup_{p\in \Proc}\Blocks_p \subseteq Z
        \subseteq \bigcup_{p\in \Proc} \Owns_p$. 
	
	It remains to check the last condition.
      Consider a "strong pattern"
      $\Owns_p \spat \Blocks_p$ with $\Owns_p=\set{t}$.
      Since $\s$ is
        "locally live" we have that $\Blocks_p=\set{t'}$, where $t,t'$
        are the two locks used by $p$.
        	As $\df(\PZ)$ does not contain any "strong cycle",
	we can pick a total order $\leq$ on locks such that for every "strong edge" $t_1
	\strongedge{p} t_2$ belonging to $\df(\PZ)$, we have $t_1 <
        t_2$.
In particular, $t<t'$, which finishes the proof.
\end{proof}

From now on we fix a "behavior@@two" $\PP$ and its lock graph $\Gp$.
We will show how to decide if there is a "full deadlock scheme" for $\PP$ in
\PTIME.
For this we need to be able to certify in \PTIME\ that there is no
$Z$-"deadlock scheme" for $\PP$, as in
Definition~\ref{def:z-ds}. 
Our approach will be to eliminate edges from $\Gp$ and try to construct a "$Z$-deadlock
scheme" on increasingly larger sets $Z$ of locks. 
We will show that this process either yields a set $Z$ that provides a "full
deadlock scheme" for $\PP$, or it fails, and in this case there is no "full
deadlock scheme" for $\PP$.

The next lemma provides a condition that allows to extend a
"$Z$-deadlock scheme" towards a "full deadlock scheme" for $\PP$, if
one exists.
This lemma is a basic ingredient to construct a "$Z$-deadlock scheme"
for increasingly larger sets $Z$ of locks. 

\begin{lem}
	\label{lem:soundpartialds}
	Let $Z \subseteq T$ be such that there is
	no "solid edge" from $Z$ to 
	$T\setminus Z$ in $\Gp$. 
	Suppose that $\ds_Z : \PZ \stackrel{\cdot}{\to} E$ is a $Z$-"deadlock scheme" for $\PP$.	
	If there exists some "full deadlock scheme" for $\PP$ then there is one 
	which is equal to $\ds_Z$ over $\PZ$.
\end{lem}

\begin{proof}
	Suppose that $\ds$ is a "full deadlock scheme" for $\PP$, so $\ds$ is a
	$B$-"deadlock scheme" for some $B\incl T$ such that for every $p\in \Proc$
	either $\ds(p)$ is defined or $p$ is $B$"-fragile" in $\PP$. 
	We construct a $(B\cup Z)$-"deadlock scheme" $\ds'$ which is
	equal to $\ds_Z$ over $\PZ$. 
	Then we show that $\ds'$ is a "full $(B\cup Z)$-deadlock scheme" for $\PP$.

	For every process $p \in \Proc$, set $\ds'(p)$ as:
		\begin{itemize}
		\item $\ds_Z(p)$ if $p \in \PZ$,

		\item $\ds(p)$ if $p \notin \PZ$ and $p$ does not
                  label any edge of $\Gp$ from $Z$ to $T \setminus Z$.
	\end{itemize} 
	First we check that $\ds'$ is a $(B\cup Z)$-"deadlock scheme".
        The first condition of a "deadlock scheme" is satisfied by
        construction. 
	Recall that we assume that there are no "solid edges" from $Z$
        to $T \setminus Z$.
        In particular, all processes $p$ such that $\ds'(p)$ is
        undefined are "fragile", so the second condition is satisfied
        as well.
	By  definition of $Z$-"deadlock scheme" there is a unique outgoing
        edge of $\ds_Z$ from every lock in $Z$. 
	A lock $t \in B \setminus Z$ has exactly one outgoing edge in
	$\ds(\Proc)$, and this edge in conserved in $\ds'$.
        Thus, the third condition is satisfied, too.
	Finally, there cannot be any "strong cycle" in $\ds'(\Proc)$ as there
	are none within $Z$, nor in $B \setminus Z$, and there are no edges
        from $Z$ to $T \setminus Z$ in $\ds'$.

	It remains to show that $\ds'$ is a "full  $(B\cup Z)$-deadlock scheme" for $\PP$.
	Let $p\in\Proc$ be an arbitrary process.
	We make a case distinction on the locks of $p$.
	The first case is when  both locks are in $Z$.
	If $p$ is "solid" then $\ds'(p)=\ds_Z(p)$ is defined.
	If $p$ is "fragile" then it is $Z$"-fragile", so also $(B\cup
        Z)$"-fragile". 
	The second case is when one lock is in $B\setminus Z$ and the
        other one in $B\cup 	Z$.
	If $p$ is "solid" then $\ds(p)$ must be defined because $\ds$ is a "full $B$-deadlock
	scheme".  
	We must have $\ds'(p)=\ds(p)$ as there are no "solid" edges from $Z$ to
	$T\setminus Z$.
	If $p$ is "fragile" then $p$ is $B$"-fragile", so also $(B\cup Z)$"-fragile".
	The final case is when one lock of $p$ is not in $B\cup Z$. 
	Since $\ds$ is a "full $B$-deadlock scheme", $p$ must be $B$"-fragile", so also $(B\cup Z)$"-fragile".
\end{proof}

Recall that we have fixed a "behavior@@two" $\PP$, and that
$\Gp=(T,E_\PP)$ is its "lock graph".  
We will describe in the following several polynomial-time algorithms operating on a subgraph $H
= (T, E_H)$ of $\Gp$, so $E_H \subseteq E_\PP$, and a set $Z$ of locks. 

We will say that $H$ has a \intro(H){deadlock scheme} 
to mean that  there is a "deadlock scheme" using only edges in $H$.
The notion of "full" is the same as for $\Gp$.

Each of the four algorithms introduced below will either eliminate
some edges from $H$ or extend 
$Z$, while maintaining the following three invariants: 

\begin{invariant}\label{Inv1}
	$\Gp$ has a "full deadlock scheme" for $\PP$ if and only if $H$ does.
\end{invariant}

\begin{invariant}\label{Inv2}
	There are no "solid edges" from $Z$ to $T \setminus Z$  in  $H$.
\end{invariant}

\begin{invariant}\label{Inv3}
	There exists a $Z$-"deadlock scheme" for $\PP$ in $\Gp$.
\end{invariant}

\label{page:inv}

Invariant 1 expresses that the edges we removed from $\Gp$ to get $H$
were not essential for finding a "full deadlock scheme for
$\PP$". Invariant 2, along with Lemma~\ref{lem:soundpartialds}, will
guarantee  that we can always extend a "$Z$-deadlock scheme"  to a
full one, if one exists. Invariant 3 maintains the existence of a
"$Z$-deadlock scheme", while $Z$ is growing.

Our algorithm will extend $Z$ as much as possible while maintaining
the three invariants. 
In the end we either obtain a "full $Z$-deadlock scheme" for $\PP$,
or a "$Z$-deadlock scheme" that is not full, but cannot be extended
anymore.
In the second case we show that no "full deadlock scheme" exists.

We may also at some point observe contradictions in the edges of $H$
that exclude the existence of any "full deadlock scheme" for $H$, in
which case we can conclude immediately thanks to Invariant 1.

We start with $H=\Gp$ and $Z = \emptyset$. All invariants are clearly
satisfied.

Our first two algorithms will analyze "solid" edges in $H$, since any
"$Z$-deadlock scheme" is defined over "solid" processes.
The first algorithm will possibly remove some edges, and the second
one will look for cycles and possibly enlarge $Z$.
The third algorithm will extend $Z$ by locks that can reach it.
Finally, the fourth algorithm will also add to $Z$ "weak cycles" that
are outside of $Z$.

\begin{defi}[Double and solo solid edges]
	Consider a "solid process" $p$. 
	We say that there is a \intro{double solid edge} $t_1\dedge{p}t_2$ in $H$ if both
	$t_1\edge{p}t_2$ and $t_1\ledge{p}t_2$ exist in $H$.  
	We say that $t_1\edge{p}t_2$ in $H$ is a \intro{solo solid edge} if there is no
	$t_1\ledge{p}t_2$ in $H$.
\end{defi}

Algorithm 1 below looks for a "solo solid edge" $t_1 \edge{p} t_2$
in $H$
and erases all other outgoing edges from $t_1$. It will be proven
correct exploiting the following property:

\begin{quotation}
 \emph{""($\star$)""}  If $t_1 \edge{p} t_2$ is a "solo solid edge"
in $H$, then any "deadlock scheme" $\ds_H$ in $H$ is such that
$\ds_H(p)=t_1 \edge{p} t_2$. 
\end{quotation}

The argument behind Property "($\star$)" is that a "deadlock scheme"
needs to map every "solid" process to one of the two possible edges of
the lock graph. So if there is only one (remaining) edge labeled
by $p$, this edge is needed and cannot be deleted.

We repeat this algorithm until no edges are removed.
If some call of the algorithm fails then there can be no "full deadlock
scheme@@H" for $\PP$ in $H$. 
Otherwise the resulting $H$ satisfies the property:
\begin{quotation}
	\emph{""(Trim)""} if a lock $t$ in $T\setminus Z$ has an outgoing "solo solid edge" then it has no other
	outgoing edges. 
\end{quotation}

\AP We denote $H$ as \intro{trimmed} if it satisfies property "(Trim)".

\begin{algorithm}
	\caption{Trimming the graph for one "solo solid edge"}\label{alg:trim}
	\begin{algorithmic}[1]
		\State Look for $t\in T\setminus Z$ with a "solo solid edge" $t \edge{p} t' \in
		E_H$ and some other outgoing edges.
		\State \textbf{if} there is no such edge then stop and report
                success.
		\For{every edge $t\edge{q} t'' \in E_H$ from $t$ with $q \neq p$}
		\If{$q$ is "solid" and $t\ledge{q} t'' \notin E_H$}
		\State \Return{``$H$ has no "deadlock scheme@@H" for $\PP$''}
		\Else
		\State delete $t\edge{q} t''$ from $E_H$
		\EndIf
		\EndFor
	\end{algorithmic}
\end{algorithm}

\begin{lem}
	\label{lem:trim}
	Suppose $(H,Z)$ satisfies \cref{Inv1,Inv2,Inv3}.
	If~\cref{alg:trim} fails then $H$ has no "full deadlock
        scheme@@H" for $\PP$. 
	After a successful execution of the algorithm all the invariants are still
	satisfied.
	If a successful execution does not remove any edge from $H$ then $H$ satisfies "(Trim)".
\end{lem}

\begin{proof}
	Let $H'$ be the graph after an execution of \cref{alg:trim}. 
	Observe that the algorithm does not change $Z$.
	If $H=H'$ then "(Trim)" holds.
	If the algorithm fails then there is a lock with two outgoing "solo solid edges". 
	In this case it is impossible to find a "full deadlock
        scheme@@H" in  $H$, because of Property "($\star$)" above and since a "deadlock scheme" has exactly
        one outgoing edge from each lock.
	
	Finally, if the algorithm succeeds but $H'$ is smaller than $H$, we must show
	that all the invariants on page~\pageref{page:inv} hold.
	Since the algorithm does not change $Z$, Invariants~\ref{Inv2} and~\ref{Inv3}
	continue to hold.
	For Invariant~\ref{Inv1}
	we use Property "($\star$)" and the fact that a "deadlock scheme" has
        a unique outgoing edge from each lock to conclude that any "full deadlock scheme@@H" in $H$ is also a "full deadlock scheme@@H" in $H'$. 
	For the other direction, a "full deadlock scheme@@H" in $H'$ is also "full@@H" in $H$,
	as $H'$ is a subgraph of $H$ with the same set of vertices. 
\end{proof}

Algorithm 2 below searches for "simple" cycles formed by "solid edges" and eventually adds them to
$Z$. 
If such a cycle is "weak@@cycle" then it can be added to $Z$.
If the cycle is "strong@@cycle", it may still be the case that
 its reversal is "weak@@cycle" (see $p_1, p_2, p_3$ in Figure~\ref{fig:alg1and2}).
More precisely it may be the case that for every "solid edge"
$t_i\act{p_i}t_{i+1}$  in the
cycle there is also a reverse edge $t_i\lact{p_i}t_{i+1}$ (which is "solid" by
definition, since $p_i$ is so). 
If the reversed cycle is also "strong@@cycle" then there is no $H$-"deadlock scheme".
Otherwise, it is "weak@@cycle" and it can be added to $Z$.
We  will show that the result still satisfies the invariants thanks to
property  "(Trim)".

\begin{algorithm}
	\caption{Find a "simple" "solid cycle".}\label{alg:solid-cycles}
\begin{algorithmic}[1]
	\State Look for a "simple cycle" of "solid edges" $t_1 \edge{p_1} t_2 \cdots \edge{p_k}
		t_{k+1} = t_1$  not intersecting $Z$ and with all $t_i$ distinct.
	\State \textbf{if} there is no such cycle, stop and report success.
	\If{all the edges on the cycle are "strong@@edge"}
		\If{for some $j$ there is no reverse edge $t_{j} \ledge{p_j} t_{j+1} \in E_H$}
			\State \Return{``$H$ has no "deadlock scheme"
                          for $\PP$''}
		\ElsIf{all edges $t_{j} \ledge{p_j} t_{j+1}$ are strong}
			\State \Return{``$H$ has no "deadlock scheme"
                          for $\PP$''}
		\EndIf
	\EndIf
	\State $Z \gets Z \cup \set{t_1, \ldots, t_k}$
		\State For every $t_i$ remove from $E_H$ all edges
                outgoing from $t_i$ except for $t_i \edge{p_i} t_{i+1}$.
                \label{line:remove-edges}
	\If{some "solid process" $p$ has no edge in $H$}\label{line:solo-test}
		\State \Return{``$H$ has no "deadlock scheme" for $\PP$''}
	\EndIf	
	\Repeat 
		\State Apply \cref{alg:trim}
	\Until{no more edges are removed from $H$}
\end{algorithmic}
\end{algorithm}	

Figure~\ref{fig:alg1and2} presents a case where \cref{alg:solid-cycles} detects an inconsistency in the "solid edges",
proving the non-existence of a "deadlock scheme".

	 	\begin{figure}
	\begin{tikzpicture}
		\tikzset{enclosed/.style={draw, circle, inner sep=0pt, minimum size=.15cm, fill=black}}
		
		\node[enclosed, label={left, yshift=.2cm: $t_1$}] (1) at (0.5,0) {};
		\node[enclosed, label={above, xshift=.2cm: $t_2$}] (2) at (2.5,1) {};
		\node[enclosed, label={above: $t_3$}] (3) at (0.5,2.2) {};
		\node[enclosed, label={below: $t_4$}] (4) at (4,1) {};
		\node[enclosed, label={below: $t_5$}] (5) at (6,0) {};
		\node[enclosed, label={left, yshift=.2cm: $t_6$}] (6) at (6,2.2) {};
		\node[enclosed, label={above: $t_7$}] (7) at (2.2,-1) {};
		\node[enclosed, label={above: $t_8$}] (8) at (4.4,-1) {};

		\node[draw,text width=6cm] at (10,1) {This graph does not have
			a "full deadlock scheme@@H" (all processes are
			"solid", "weak edges" are displayed in red). However a first
			execution of \cref{alg:trim} has no effect as all edges are double.};
		
		\path[->, -stealth, thick, bend right=20] (1) edge node[above] {$p_1$} (2);
		\path[->, -stealth, thick, bend right=20] (2) edge (1);
		\path[->, -stealth, thick, bend right=20, color = red] (2) edge node[below] {\color{black}$p_2$} (3);
		\path[->, -stealth, thick, bend right=20] (3) edge (2);
		\path[->, -stealth, thick, bend right=20] (3) edge node[right] {$p_3$} (1);
		\path[->, -stealth, thick, bend right=20] (1) edge (3);
		
		\path[->, -stealth, thick, bend right=20] (4) edge node[above] {$p_4$} (5);
		\path[->, -stealth, thick, bend right=20] (5) edge (4);
		\path[->, -stealth, thick, bend right=20, color = red] (5) edge node[left, xshift=1mm] {\color{black}$p_5$} (6);
		\path[->, -stealth, thick, bend right=20] (6) edge (5);
		\path[->, -stealth, thick, bend right=20] (6) edge node[below] {$p_6$} (4);
		\path[->, -stealth, thick, bend right=20] (4) edge (6);
		
		\path[->, -stealth, thick, bend right=20] (1) edge node[above] {$p_7$} (7);
		\path[->, -stealth, thick, bend right=20] (7) edge (1);
		\path[->, -stealth, thick, bend right=20] (7) edge node[above] {$p_8$} (8);
		\path[->, -stealth, thick, bend right=20] (8) edge (7);
		\path[->, -stealth, thick, bend right=20] (8) edge node[above] {$p_9$} (5);
		\path[->, -stealth, thick, bend right=20] (5) edge (8);
		
		\node[enclosed, label={left, yshift=.2cm: $t_1$}] (1) at (0.5,-4) {};
		\node[enclosed, label={above, xshift=.2cm: $t_2$}] (2) at (2.5,-3) {};
		\node[enclosed, label={above: $t_3$}] (3) at (0.5,-1.8) {};
		\node[enclosed, label={below: $t_4$}] (4) at (4,-3) {};
		\node[enclosed, label={below: $t_5$}] (5) at (6,-4) {};
		\node[enclosed, label={left, yshift=.2cm: $t_6$}] (6) at (6,-1.8) {};
		\node[enclosed, label={above: $t_7$}] (7) at (2.2,-5) {};
		\node[enclosed, label={above: $t_8$}] (8) at (4.4,-5) {};
		
		\node[draw,text width=6cm] at (10,-3) {We apply \cref{alg:solid-cycles}, which finds "solid cycles", erases all other edges going out of those cycles, and makes sure that those cycles are "weak@@cycle".};
		
		\path[->, -stealth, thick, bend right=20] (1) edge node[above] {$p_1$} (2);
		\path[->, -stealth, thick, bend right=20, color = red] (2) edge node[below] {\color{black}$p_2$} (3);
		\path[->, -stealth, thick, bend right=20] (3) edge node[right] {$p_3$} (1);
		\path[->, -stealth, bend right=20, opacity=.4] (2) edge (1);
		\path[->, -stealth, bend right=20, opacity=.4] (3) edge (2);
		\path[->, -stealth, bend right=20, opacity=.4] (1) edge (3);
		
		\path[->, -stealth, thick, bend right=20] (4) edge node[above] {$p_4$} (5);
		\path[->, -stealth, thick, bend right=20, color = red] (5) edge node[left, xshift=1mm] {\color{black}$p_5$} (6);
		\path[->, -stealth, thick, bend right=20] (6) edge node[below] {$p_6$} (4);
		\path[->, -stealth, bend right=20, opacity=.4] (5) edge (4);
		\path[->, -stealth, bend right=20, opacity=.4] (6) edge (5);
		\path[->, -stealth, bend right=20, opacity=.4] (4) edge (6);
		
		\path[->, -stealth, thick, bend right=20] (7) edge node[below] {$p_7$} (1);
		\path[->, -stealth, thick, bend right=20] (7) edge node[above] {$p_8$} (8);
		\path[->, -stealth, thick, bend right=20] (8) edge (7);
		\path[->, -stealth, thick, bend right=20] (8) edge node[above] {$p_9$} (5);
		
		\node[draw,text width=6cm] at (10,-7) {We now apply
			\cref{alg:trim} again. It detects that  $t_8 \edge{p_9} t_5$
			is a "solo solid edge" and it erases the other outgoing edge
			$t_8 \edge{p_8} t_7$. It then concludes that there is no
			"full deadlock scheme@@H" as $t_7$ has two outgoing  "solo solid edges".};		
		
		\node[enclosed, label={left, yshift=.2cm: $t_1$}] (1) at (0.5,-8) {};
		\node[enclosed, label={above, xshift=.2cm: $t_2$}] (2) at (2.5,-7) {};
		\node[enclosed, label={above: $t_3$}] (3) at (0.5,-5.8) {};
		\node[enclosed, label={below: $t_4$}] (4) at (4,-7) {};
		\node[enclosed, label={below: $t_5$}] (5) at (6,-8) {};
		\node[enclosed, label={left, yshift=.2cm: $t_6$}] (6) at (6,-5.8) {};
		\node[enclosed, label={above: $t_7$}] (7) at (2.2,-9) {};
		\node[enclosed, label={above: $t_8$}] (8) at (4.4,-9) {};
		\node (9) at (3,-8.5) {\color{orange} !};
		
		\path[->, -stealth, thick, bend right=20] (1) edge node[above] {$p_1$} (2);
		\path[->, -stealth, thick, bend right=20, color = red] (2) edge node[below] {\color{black}$p_2$} (3);
		\path[->, -stealth, thick, bend right=20] (3) edge node[right] {$p_3$} (1);
		\path[->, -stealth, bend right=20, opacity=.4] (2) edge (1);
		\path[->, -stealth, bend right=20, opacity=.4] (3) edge (2);
		\path[->, -stealth,  bend right=20, opacity=.4] (1) edge (3);
		
		\path[->, -stealth, thick, bend right=20] (4) edge node[above] {$p_4$} (5);
		\path[->, -stealth, thick, bend right=20, color = red] (5) edge node[left, xshift=1mm] {\color{black}$p_5$} (6);
		\path[->, -stealth, thick, bend right=20] (6) edge node[below] {$p_6$} (4);
		\path[->, -stealth, bend right=20, opacity=.4] (5) edge (4);
		\path[->, -stealth, bend right=20, opacity=.4] (6) edge (5);
		\path[->, -stealth, bend right=20, opacity=.4] (4) edge (6);
		
		\path[->, -stealth, thick, bend right=20] (7) edge node[below] {$p_7$} (1);
		\path[->, -stealth, thick, bend right=20] (7) edge node[above] {$p_8$} (8);
		\path[->, -stealth, thick, bend right=20] (8) edge node[above] {$p_9$} (5);
		
		\draw[thick, orange] (2.2,-9) circle (0.6);
	\end{tikzpicture}
	\caption{An example of application of \cref{alg:solid-cycles}.}
	\label{fig:alg1and2}
\end{figure}

\begin{lem}\label{lem:solid-cycles}
	Suppose $(H,Z)$ satisfies the \cref{Inv1,Inv2,Inv3} and $H$ is "trimmed".
	If the execution of~\cref{alg:solid-cycles} does not fail then the resulting
	$H$ and $Z$ also satisfy all invariants and "(Trim)".
	If the execution fails then $H$ has no "full deadlock scheme@@H" for $\PP$.
\end{lem}

\begin{proof}
	Suppose that the algorithm finds a "simple cycle" $t_1 \edge{p_1} t_2 \cdots \edge{p_k}
	t_{k+1} = t_1$ where all $p_i$ are "solid processes", 
	and all $t_i$ are distinct. 
	By definition of a "simple cycle", all $p_i$ are distinct as well.
	If there is a "full deadlock scheme@@H" for $H$ then it should assign either $t_{i}
	\edge{p_i} t_{i+1}$ or $t_{i} \ledge{p_i} t_{i+1}$ to $p_i$,
        because $p_i$ is "solid".

	We examine the cases when the algorithm fails.
	The first reason for failure may appear when all the edges on the cycle are "strong@@edge".
	If for some $j$ there is no reverse edge $t_{j} \ledge{p_j} t_{j+1}$ in $E_H$
	then a "full deadlock scheme@@H" for $H$, call it $\ds_H$, should assign the
	edge $t_{j} \edge{p_j} t_{j+1}$ to  $p_j$, because the edge is
        "solo solid" (recall Property "($\star$)"). 
	As a consequence, as $\ds_H$ has to give each $t_i$ at most
        one outgoing edge and all edges of the cycle are "solid", all
        the edges in the cycle should be in the image of $\ds_H$.  
	But this is forbidden by the last condition in the definition of "deadlock scheme",
        as the cycle is "strong@@cycle".
        
	When there are reverse edges $t_{i} \ledge{p_i} t_{i+1}
	\in E_H$ for all $i$, the algorithm fails if all of them are "strong@@edge".
	Indeed, there cannot exist any "full deadlock scheme@@H" for
        $H$ in this case either, because either the cycle or its
        reverse would need to be in the image of $\ds_H$, but both are
        "strong@@cycle".
        
	The last reason for failure is when there is some "solid process" $p$ and
	all the $p$-labeled edges were removed by the algorithm.
	These must be edges	of the form $t_i\act{p} t$ that are
        not on the cycle, for some $i=1,\dots,k$ and $p \not= p_i$.
	Those edges cannot belong to a "deadlock scheme" as it has to contain the cycle
	in one direction or the other and thus cannot contain other outgoing edges from
	that cycle.
	As a "deadlock scheme" cannot assign any edge to $p$, and $p$
        is "solid", there cannot exist any "full deadlock scheme@@H" in that case.

	If the algorithm does not fail then either the cycle $t_1 \edge{p_1} t_2 \cdots \edge{p_k}
	t_{k+1} = t_1$ is "weak@@cycle", or its reverse is. 
	Thanks to Lemma~\ref{lem:trim}, we only need to show that all three invariants hold after 
	line~\ref{line:remove-edges}.
	Let $(H',Z')$ be the values at that point.
	So $Z'=Z\cup\set{t_1, \ldots, t_k}$, and $H'$ is $H$ after removing edges in
	line~\ref{line:remove-edges}.
	We show now that all invariants on page~\pageref{page:inv}
        continue to hold. 
	
	For Invariant~\ref{Inv2}, we 
	observe that thanks to "(Trim)" for every lock in $Z'$ there is
        exactly one outgoing edge in 	$H'$.
	So there is no "solid edge" from $Z'$ to $T\setminus Z'$ as
        there was none from $Z$ to $T \setminus Z$ and the only
        "solid" edge of $H'$ outgoing from $t_i$ is $t_i \act{p_i} t_{i+1}$.
	
	For Invariant~\ref{Inv3}, we extend our $Z$-"deadlock scheme"
        to a  $Z'$-"deadlock scheme":
	we choose the cycle found by the algorithm or its reversal, depending on which one is "weak@@cycle". 
	For every $p_i$ we define $\df_{Z'}(p_i)$ to be the edge in the chosen cycle. 
	For all $p\in \Proc_{Z'}\setminus \PZ$ other
	than $p_1,\dots,p_k$, $\df_{Z'}(p)$ is undefined.
	We must show that such any such $p$ is "fragile".
        If both locks used by $p$ are among the $\set{t_1,\dots,t_k}$
        then $p$ must be "fragile" because Algorithm~2 does not fail at
        line~\ref{line:solo-test}. 
	The other case is where $p$ has one lock $t$ in $Z$, and the other, 
	$t'$ in $Z' \setminus Z$.
        If $p$ was "solid", then given that the algorithm does not
        fail at line~\ref{line:solo-test}, there must be some
        ("solid") edge labeled by $p$ in $H'$.
	However, by Invariant~\ref{Inv2} for $H$, an edge from $t$ to
	$t'$ cannot be "solid".
	Moreover, the edge from $t'$ to $t$ is removed at line 11.
	Therefore, $p$ is "fragile".

	For Invariant~\ref{Inv1} suppose that $H'$ has a "full deadlock scheme@@H" for $\PP$. 
	Then this is also a "full deadlock scheme@@H" for $H$ as well, as $H'$ is a
	subgraph of $H$ over the same set of locks.
	For the other direction consider a "full $B$-deadlock
        scheme@@H" $\ds_H$ in $H$, for some $B \subseteq T$. 
	By Lemma~\ref{lem:soundpartialds}, as we showed that
        Invariant~\ref{Inv2} is 
	maintained for $Z'$, we can assume that $Z' \subseteq B$ and $\ds_H$ is equal to 
	$\ds_{Z'}$ on $\Proc_{Z'}$.
	We define a "deadlock scheme" $\ds_{H'}$ for $H'$ as follows.
	If $\ds_H(p)$ is undefined then $\ds_{H'}(p)$ is undefined, too.
	Otherwise, if the source vertex of $\ds_H(p)$ is not in $Z'$ then
	$\ds_{H'}(p)=\ds_H(p)$.
	This edge is guaranteed to exist also in $H'$ because only some edges
        outgoing from the $t_i$ were removed.
	If the two locks of $p$ are both in $Z'$ let
	$\ds_{H'}(p)=\ds_H(p)=\ds_{Z'}(p)$.
	The remaining case is when $\ds_H(p)$ is an edge $t\edge{p}
        t'$ with $t\in Z'$ and $t'\notin Z'$.
        Note that $t,t'$ are both in $B$.
        If $t \in Z$ then this would contradict Condition~3  in the
        definition of "deadlock scheme", as $p \notin \Proc_Z$.
        Hence $t=t_i$ for some $i$, and $p$ is "fragile" as the only
        "solid" edge leaving $t_i$ in $H'$ is $t_i \act{p_i} t_{i+1}$.
	We let $\df_{H'}(p)$ be undefined in this case, and
        Condition~2 of "deadlock scheme" is satisfied.

        We establish now that $\ds_{H'}$ is a "full $B$-deadlock
        scheme@@H"  in $H'$.
        All we need to check is that any process $p$ with
        $\df_{H'}(p)$ undefined is $B$"-fragile".
        If $\df_{H}(p)$ was already undefined then we get that $p$ is
        $Z$"-fragile", so $B$"-fragile" as well.
        If $\df_{H}(p)$ was defined, but $\df_{H'}(p)$ is not, then since
         both locks of $p$ are in $B$ and $p$ is "fragile", we obtain
         that $p$ is $B$"-fragile".
         This concludes the proof.
\end{proof}

	\begin{lem}
		If \cref{alg:solid-cycles} succeeds but does not increase $Z$ nor decrease $H$ then $(H,Z)$ satisfies
		three properties:
		\begin{description}
			\item[\namedlabel{H1}{H1}] $H$ is "trimmed".
			
			\item[\namedlabel{H2}{H2}] $H$ has no "solid
                          cycle" that intersects $T\setminus Z$.
			\item[\namedlabel{H3}{H3}] Every "solid process" has an edge in $H$.
		\end{description}
	\end{lem}

	\begin{proof}
Property~\ref{H1} is satisfied because $H$ was not modified by \cref{alg:trim}.
		
\cref{Inv2} is satisfied by Lemma~\ref{lem:solid-cycles}, 
hence any "solid" simple cycle intersecting $T \setminus Z$ in $H$
must lie entirely in $T \setminus Z$.
Moreover, it is easy to see that  if there is some "solid" cycle in $H$ intersecting $T
\setminus Z$, then there exists also a "simple" one.
In this case \cref{alg:solid-cycles} would not have stopped in line 2,
and thus would have either failed or increased $Z$. There is therefore
no  "solid" cycle intersecting $T\setminus Z$ in $H$, hence property
\ref{H2} is also satisfied. 
		
Finally, Property~\ref{H3} is satisfied because \cref{alg:solid-cycles}
did not fail at line 12-13.  
 	\end{proof}

The next algorithms will not modify $H$ anymore and only increase $Z$.
Therefore, all three properties stated in the previous lemma
will continue to hold.

\begin{defi}
Given a pair $(H,Z)$ consisting of a subgraph $H$ of $\Gp$ and
a set $Z \subseteq T$ of locks we define the following
equivalence relation on $T$:  $t_1 \intro*\eqh t_2$  if $t_1,t_2 \in T \setminus Z$ and there is a path of "double
	solid edges" in $H$ between $t_1$ and $t_2$.
\end{defi}

Intuitively, once we have "trimmed" the graph and eliminated "simple cycles" of "solid edges" with
\cref{alg:solid-cycles},
the equivalence classes of $\eqh$ are ``trees'' made of "double solid
edges" (c.f.~Lemma~\ref{lem:unique-simple-path} below) with no
outgoing edges (except for singletons,
c.f.~Lemma~\ref{lem:no-outgoing-edges-from-eqh}). 

\begin{lem}\label{lem:no-outgoing-edges-from-eqh}
	If $H$ satisfies property \ref{H1} and $t_1\edge{p} t_2$ is in $H$ for a "solid process"
	$p$ then either the $\eqh$-equivalence class of $t_1$ is a singleton, or
	$t_1\ledge{p} t_2$ is in $H$, hence  $t_1 \eqh t_2$.
\end{lem}
\begin{proof}
	If the $\eqh$-equivalence class of $t_1$ is not a singleton
        then $t_1 \notin Z$ and there is a "double solid edge" from $t_1$. 
	By property \ref{H1}, there cannot be any outgoing "solo
        solid edge" from $t_1$, so $t_1\ledge{p} t_2$ must be in $H$,
        too. 
\end{proof}

\begin{lem}\label{lem:unique-simple-path}
  Suppose that $H$ satisfies properties \ref{H1} and \ref{H2}.
  Let $t_1, t_2 \in T \setminus Z$.
	If $t_1\eqh t_2$ then $H$ has a unique "simple path" of
	"solid edges"  from $t_1$ to $t_2$.
\end{lem}

\begin{proof}
  If $t_1 = t_2$ then any non-empty "simple path" of "solid edges"
  from $t_1$ to $t_2$ would contradict property \ref{H2}, hence the
  empty path is the only "simple path" from $t_1$ to $t_2$. If
  $t_1 \neq t_2$ then by definition of $\eqh$ there is a path of
  "double solid edges" from $t_1$ to $t_2$, hence there is such a "simple
  path" from $t_1$ to $t_2$.
	
  Suppose there exist two distinct "simple paths" from $t_1$ to $t_2$,
  then by Lemma~\ref{lem:no-outgoing-edges-from-eqh} all the locks on
  those paths are in the $\eqh$-equivalence class of $t_1$ and
  $t_2$.
  Hence as $t_1 \notin Z$, there is a cycle of "double solid
  edges" intersecting $H \setminus Z$, contradicting property
  \ref{H2}.
\end{proof}

Our third algorithm looks for an edge $t_1 \edge{p} t_2$ with $t_1 \notin Z$ and
$t_2 \in Z$, and adds the full $\eqh$-equivalence class $C$ of $t_1$  to $Z$.
This step will be shown correct by showing that  a $Z$-"deadlock scheme" 
extends to a $(Z \cup C)$-"deadlock scheme" by orienting edges in $C$
towards $Z$, as 
displayed in the example in Figure~\ref{fig:alg3}.

\begin{algorithm}
	\caption{Extending $Z$ by locks that can reach it.}\label{alg:reach}
	\begin{algorithmic}[1]
		\While{there exists $t_1 \edge{p} t_{2} \in E_H$ with $t_1 \notin Z$ and $t_2 \in Z$}
			\State $Z \gets Z \cup \set{t \in T \mid t \eqh t_1}$
		\EndWhile
	\end{algorithmic}
\end{algorithm}

\begin{lem}\label{lem:reach}
Suppose that $H$ satisfies properties \ref{H1}, \ref{H2} and \ref{H3},
and $(H,Z)$ satisfies \cref{Inv1,Inv2,Inv3}.
After executing \cref{alg:reach}, the resulting $H$ and $Z$ also
satisfy all these properties, and $H$ has no edges from $T\setminus Z$ to $Z$.
\end{lem}
\begin{proof}
Let $(H',Z')$ be the pair obtained by applying \cref{alg:reach}. 
Invariant~\ref{Inv1}, and properties \ref{H1} and \ref{H3} continue to  hold because $H'=H$.
Also property \ref{H2} continues to hold, because $Z \subseteq Z'$.

It remains to show that Invariant~\ref{Inv2} (no "solid" edges
from $Z$ to $T \setminus Z$) and Invariant~\ref{Inv3} (existence of
$Z$-"deadlock scheme") are preserved.
	 
Let $Z_{m+1}$ be the value	of $Z$ at the end of the $m$-th iteration. 
So $Z_{m+1}=Z_m\cup\set{t \in T \mid t \eqh t_1}$, where $t_1 \edge{p} t_{2}$
is the edge found in the guard of the while statement.
We verify that $Z_{m+1}$ satisfies Invariants~\ref{Inv2} and~\ref{Inv3} if
$Z_m$ does. 

	For Invariant~\ref{Inv2}, Lemma~\ref{lem:no-outgoing-edges-from-eqh} says that
	there are no outgoing "solid edges" from the $\eqh$-equivalence class of $t_1$,
	unless that class is a singleton.
	If it is a singleton, there are no outgoing solid edges from $t_1$ or $t_1
	\edge{p} t_2$ is the only outgoing edge of $t_1$. 
	In both cases, there are no solid edges from $Z_{m+1}$ to $T \setminus
	Z_{m+1}$ in $H$.  

	For Invariant~\ref{Inv3} we extend a $Z_m$-"deadlock scheme" $\ds_m$
        to a $Z_{m+1}$-"deadlock scheme" $\ds_{m+1}$.
	If the two locks of some process $q$ are both in $Z_m$ then $\df_{m+1}(q)=\df_m(q)$.
	We set $\df_{m+1}(p)$ to be the edge $t_1 \edge{p} t_{2}$
        found by the algorithm, so here $t_1 \in Z_{m+1} \setminus Z_m$ and
        $t_2 \in Z_m$.
	Let $C$ be the $\eqh$-equivalence class of $t_1$: $C=\set{t \in T \mid t \eqh t_1}$.
	By Lemma~\ref{lem:unique-simple-path} there is a unique "simple path" from $t\in
	C$ to $t_1$. 
	Let $t\edge{q}t'$ be the first edge on this path. 
	We set $\df_{m+1}(q)$ to be this edge.
        We let $\df_{m+1}(q)$ be undefined for all remaining processes $q$.
	
	We verify now that $\df_{m+1}$  is a $Z_{m+1}$-"deadlock scheme".
	By construction every lock in $C$ has a unique outgoing edge in
	$\df_{m+1}$, hence every lock in $Z_{m+1}$ does so. 
	It is also immediate that $\df_{m+1}(\Proc_{Z_{m+1}})$ does not contain a "strong cycle" as it
	would need to be already the case for $\df_m$ and $Z_m$: every
        lock of $C$ has exactly one outgoing edge in $\df_{m+1}$ and
        the path obtained by following those edges from an element of
        $C$ leads to $Z_m$.

        It remains to show that $\ds_{m+1}$ is defined for every
        "solid" process $q\in \Proc_{Z_{m+1}}$.
	Suppose by contradiction that $\df_{m+1}(q)$ is not defined by the procedure.
	If both locks of $q$ are in $Z_m$ then $\df_{m+1}(q)$ must be defined because
	$\df_m(q)$ is.
	If $q = p$, the process labeling the transition chosen by the algorithm,
	then $\df_{m+1}(q)$ is defined. 
	In the remaining case both locks of $q$, say $t,t'$, are in $C$.
	If neither $t\edge{q} t'$ is on the shortest path from $t$ to $t_1$, nor
	is $t\ledge{q} t'$ on the shortest path from $t'$ to $t_1$ then there must
	be a "solid" cycle in $C$. 
	But this is impossible as we assumed that there are no "solid cycles" intersecting $T\setminus Z$ (property \ref{H2}) and $Z \subseteq Z_m$.
	Hence $\df_{m+1}(q)$ is defined, and $\df_{m+1}$ is a $Z_{m+1}$-"deadlock scheme".
	
	All what is left to prove is that $H$ has no edges from $T\setminus Z$ to $Z$, which is immediate as otherwise \cref{alg:reach} would not have stopped.
\end{proof}

\begin{figure}
\begin{tikzpicture}[scale=0.85]
\tikzset{enclosed/.style={draw, circle, inner sep=0pt, minimum size=.15cm, fill=black}}

\node (A) at (9,3) {};
\node (B) at (8,3) {};
\node (C) at (11,4) {};
\node (D) at (11,0) {};

\pgfdeclarelayer{background}
\pgfdeclarelayer{foreground}
\pgfsetlayers{background,main,foreground}

\begin{pgfonlayer}{background}
\draw[color=white, fill=blue!5!white] (C) .. controls (A) and (B) .. (D);
\draw[dashed, thick] (C) .. controls (A) and (B) .. (D);
\end{pgfonlayer}

\begin{pgfonlayer}{foreground}
\node[enclosed, label={left, yshift=.2cm: $t_1$}] (1) at (1,2.2) {};
\node[enclosed, label={above, xshift=.2cm: $t_2$}] (2) at (3,2) {};
\node[enclosed, label={above: $t_3$}] (3) at (5,1) {};
\node[enclosed, label={right: $t_4$}] (4) at (6,-1) {};
\node[enclosed, label={above: $t_5$}] (5) at (7,2) {};
\node[enclosed, label={left: $t_6$}] (6) at (2.5,0) {};
\node[enclosed, label={right: $t_7$}] (7) at (10,2) {};
\node (Z) at (10,3) {\Large$Z$};

\path[->, -stealth, thick, bend right=30] (1) edge node[above=.05cm] {$p_1$} (2);
\path[->, -stealth, thick, bend right=30] (2) edge (1);
\path[->, -stealth, thick, bend right=30] (3) edge node[above=.1cm] {$p_4$} (5);
\path[->, -stealth, thick, bend right=30] (5) edge (3);
\path[->, -stealth, thick, bend right=30] (2) edge node[above=.1cm] {$p_2$} (3);
\path[->, -stealth, thick, bend right=30] (3) edge (2);
\path[->, -stealth, thick, bend right=30] (4) edge node[left=.05cm] {$p_3$} (3);
\path[->, -stealth, thick, bend right=30] (3) edge (4);
\path[->, -stealth, thick, bend right=30] (6) edge node[left=0cm] {$p_5$} (2);
\path[->, -stealth, thick, bend right=30] (2) edge (6);
\path[->, -stealth, thick, bend right=30, dotted] (5) edge (7);
\end{pgfonlayer}

\node (A) at (0,-2) {};
\node (B) at (-1,-2) {};
\node (C) at (1, -1) {};
\node (D) at (2,-6) {};

\begin{pgfonlayer}{background}
\draw[dashed, thick , fill=blue!5!white] (11,-2.5) .. controls (9,-2) and (3,-1.5) ..  (1, -2) .. controls (-1,-2.5) and (-2,-5.5) .. (2,-6) .. controls (6,-6.5) and (9,-7.5) .. (11,-7);
\end{pgfonlayer}

\begin{pgfonlayer}{foreground}
\node[enclosed, label={left, yshift=.2cm: $t_1$}] (1) at (1,-3.3) {};, opacity=.4
\node[enclosed, label={above, xshift=.2cm: $t_2$}] (2) at (3,-3.5) {};
\node[enclosed, label={above: $t_3$}] (3) at (5,-4.5) {};
\node[enclosed, label={right: $t_4$}] (4) at (6,-6.5) {};
\node[enclosed, label={above: $t_5$}] (5) at (7,-3.5) {};
\node[enclosed, label={left: $t_6$}] (6) at (2.5,-5.5) {};
\node[enclosed, label={right: $t_7$}] (7) at (10,-3.5) {};
\node (Z) at (3,-2.3) {\Large$Z$};

\path[->, -stealth, thick, bend right=30] (1) edge node[above=.05cm] {$p_1$} (2);
\path[->, -stealth, bend right=30, opacity=.4] (2) edge (1);
\path[->, -stealth, thick, bend right=30] (3) edge node[above=.1cm] {$p_4$} (5);
\path[->, -stealth, bend right=30, opacity=.4] (5) edge (3);
\path[->, -stealth, thick, bend right=30] (2) edge node[above=.1cm] {$p_2$} (3);
\path[->, -stealth, bend right=30, opacity=.4] (3) edge (2);
\path[->, -stealth, thick, bend right=30] (4) edge node[left=.05cm] {$p_3$} (3);
\path[->, -stealth, bend right=30, opacity=.4] (3) edge (4);
\path[->, -stealth, thick, bend right=30] (6) edge node[left=0cm] {$p_5$} (2);
\path[->, -stealth, bend right=30, opacity=.4] (2) edge (6);
\path[->, -stealth, thick, bend right=30, dotted] (5) edge (7);
\end{pgfonlayer}
\end{tikzpicture}
\caption{Illustration of \cref{alg:reach}. A deadlock in $Z$ can be extended to all these processes by orienting all edges/processes towards $Z$ (black arrows in the bottom part).}
\label{fig:alg3}
\end{figure}

Our last algorithm looks for "weak cycles" in the remaining graph. 
If it finds one, it adds to $Z$ not only all locks in the cycle but
also their $\eqh$-equivalence classes.

\begin{algorithm}
	\caption{Incorporating "weak cycles"}\label{alg:other-cycles}
	\begin{algorithmic}[1]
	\If{there exists a "weak cycle" $t_1 \edge{p_1} t_2 \cdots
          \edge{p_k} t_{k+1}=t_1$ with $t_k \edge{p_k} t_1$
          "weak@@edge" and  $t_i \notin Z$ for some $i$, }	
				\State $Z \gets Z \cup \bigcup_{i=1}^k \set{t \mid t\eqh t_i}$\label{line:alg4:adding}
	\EndIf
	\end{algorithmic}
\end{algorithm}

\begin{lem}
	\label{lem:other-cycles}
	Suppose $H$ satisfies \ref{H1}, \ref{H2} and \ref{H3}, $(H,Z)$ satisfies \cref{Inv1,Inv2,Inv3}, and
	moreover there are no edges from $T\setminus Z$ to $Z$.
	After an execution of~\cref{alg:other-cycles},  $H$ still
        satisfies \ref{H1}, \ref{H2} and \ref{H3}, and the resulting $(H,Z)$
	satisfies \cref{Inv1,Inv2,Inv3}.
\end{lem}
\begin{proof}
	Let $(H',Z')$ be the pair obtained after execution of \cref{alg:reach}. 
	Observe that $H'=H$, hence 	Invariant~\ref{Inv1} holds.
	For the same reason \ref{H1} and \ref{H3} are still
        satisfied. Furthermore, as $Z \subseteq Z'$, so is \ref{H2}.
	It remains to verify Invariants~\ref{Inv2} and~\ref{Inv3}.

	Consider the "weak cycle" found by the algorithm
	$t_1 \edge{p_1} t_2 \cdots \edge{p_k} t_{k+1}=t_1$, and note
        that $t_i \notin Z$ for all $i$ because $H$ has no edges from
        $T \setminus Z$ to $Z$.
	Let $Z'=Z \cup \bigcup_{i=1}^k \set{t \mid t\eqh t_i}$ as
	in line~\ref{line:alg4:adding}.
		
	Towards showing Invariant~\ref{Inv2} consider some $t_i$ on the  cycle.
	Lemma~\ref{lem:no-outgoing-edges-from-eqh} says that
	there are no outgoing solid edges from the $\eqh$-equivalence class of $t_i$,
	unless that class is a singleton.
	If this class is a singleton, there are no outgoing "solid edges" from $t_i$ or $t_i
	\edge{p} t_{i+1}$ is the only outgoing edge of $t_i$. 
	In both cases, there are no "solid edges" from $Z'$ to $T \setminus
	Z'$ in $H$.  
	
	For Invariant~\ref{Inv3} we extend a $Z$-"deadlock scheme"
        $\ds_{Z}$ to a $Z'$-"deadlock scheme" $\ds_{Z'}$.
	For every lock $t\in Z'\setminus Z$ let
	$j$ be the biggest index among 	$1,\dots,k$ with $t\eqh t_j$. 
	If $t=t_j$ then set $\df_{Z'}(p_j)$ to be the edge $t_j\edge{p_j}t_{j+1}$.
	Otherwise, take the unique path from $t$ to $t_j$ in the $\eqh$-equivalence class of the
	two locks; this is possible thanks to Lemma~\ref{lem:unique-simple-path}.
	If the path starts with $t\edge{p}t'$ then set $\df_{Z'}(p)$
        to this edge.
        For all remaining processes $p$ we let $\df_{Z'}(p)$ be undefined.
		
	We show now that $\ds_{Z'}$ is a $Z'$-"deadlock scheme".
	First, note that there is an outgoing $\ds_{Z'}$ edge from every lock in $Z'$ by
	definition, and this edge is unique.
	
	Next we show that $\ds_{Z'}(p)$ is defined for every "solid process" $p$.
	This is clear if the two locks, $t$ and $t'$, of $p$ are in $Z$. 
	If both locks are in $Z' \setminus Z$ then either $t\eqh t'$ or there is a "solo solid edge"
	between the two, say $t\edge{p}t'$.
	In the latter case this is the only edge from $t$, as $H$ is "trimmed". 
	As the $\eqh$-equivalence class of $t$ is then a singleton, this must be an edge on the
	cycle and $\ds_{Z'}(p)$ is defined to be this edge.
	Suppose now that $t\eqh t'$ and $\ds_{Z'}(p)$ is not defined. Let $j$ be the biggest index among
	$1,\dots,k$ such that $t\eqh t_j$.
	If neither $t\edge{p} t'$ is on the shortest path from $t$ to $t_j$, nor
	$t\ledge{p} t'$ is on the shortest path from $t'$ to $t_j$ then there must
	be a cycle in $C$. 
	But this is impossible as we assumed that there are no "solid cycles" intersecting $T\setminus Z$ in $H$ (Property~\ref{H2}).
	The last case is when one of the locks of $p$ is in $Z$ and
        the other in
	$Z'\setminus Z$. 
	There is no "solid edge" leaving $Z$ by Invariant~\ref{Inv2}.
	There is no "solid edge" entering $Z$ by the assumption of the lemma.
	So $p$ is a "solid process" labeling no edge in $H$ which contradicts \ref{H3}.

	The last thing to verify for a $Z'$-"deadlock scheme" is that there is no "strong cycle" in $\ds_{Z'}(\Proc_{Z'})$.
	We first check that $\ds_{Z'}(\Proc_{Z'})$ contains $t_k\edge{p_k} t_1$. 
	This is because $t_k$ is necessarily the last one from its $\eqh$-equivalence class. 
	A "strong cycle" cannot contain locks from $Z$ as there are no edges entering
	$Z$ in $\ds_{Z'}$. 
	Let $t'_1\edge{p'_1}t'_2\dots\edge{p'_l}t'_{l+1}=t'_1$ be a hypothetical "strong
	cycle" in $Z'\setminus Z$ using transitions in $\df_{Z'}$.

	Consider $x$ such that $t'_1\eqh t'_j$ for $j\leq x$ but $t'_1\neqh
	t'_{x+1}$. 
	By definition of $\df_{Z'}$ we must have that $t'_{x}$ is the last lock among
	$t_1,\dots,t_k$ equivalent to $t'_1$, say it is $t_{y}$.
	As each lock only has one outgoing transition in the image of $\df_{Z'}$, and as there is a path from $t_y$ to $t_k$ in that image, $t_k$ must be on that cycle, and thus the "weak edge" $t_k \edge{p_k} t_1$ as well, contradicting the assumption that this is a "strong cycle". 
\end{proof}

We conclude with our complete algorithm (if one of our sub-algorithms
returns a result, then the entire algorithm stops):

\begin{algorithm}
	\caption{Algorithm to check the existence of a "full deadlock
          scheme" for $\PP^\s$}\label{alg:completealgo}
	\begin{algorithmic}[1]
		\State $H \gets G_{\PP^\s}$
		\State $Z \gets \emptyset$
		\Repeat 
			\State apply \cref{alg:trim}
		\Until{no more edges are removed from $H$}
		\Repeat \Comment{$H$ is "trimmed"}
			\State apply \cref{alg:solid-cycles}
		\Until{no more edges are removed from $H$}
		\Repeat \Comment{from now on $H$ satisfies properties \ref{H1}, \ref{H2} and \ref{H3}}
			\State apply \cref{alg:reach}\Comment{no edges from $T\setminus Z$ to $Z$}
			\State apply \cref{alg:other-cycles}
		\Until{$Z$ does not increase anymore}
		\If{there is a process $p\notin \PZ$ that is not $Z$"-fragile"}\label{line:last-test}
			\State \Return{``$\sigma$ is "winning"''}\label{line:last-fail}
		\Else
			\State \Return{``$\sigma$ is not "winning"''}
		\EndIf
	\end{algorithmic}
\end{algorithm}

\begin{lem}
	\label{lem:correctcompletealgo}
	\cref{alg:completealgo} terminates in polynomial time, and
        return ``$\s$ winning'' if and only
	if no "full deadlock scheme" for $\PP^\s$ exists.
\end{lem}
\begin{proof}
  Let $\PP=\PP^\s$.
	Suppose that the algorithm fails before reaching the end.
	If this happens before line~\ref{line:last-test} then using 
	Lemma~\ref{lem:trim}, Lemma~\ref{lem:solid-cycles} and Invariant~\ref{Inv1} we obtain that $\Gp$ does not have any "full
	deadlock scheme" for $\PP$.
	If the algorithm fails at line~\ref{line:last-fail} then there
        exists a process 
	$p \notin \Proc_Z$ that is not $Z$"-fragile".
	Suppose towards a contradiction that $H$ has a "full deadlock
        scheme@@H" $\ds_H$, and assume that $\ds_H$ is a $B$-"deadlock
        scheme" for some $B \subseteq T$.
        By Lemma~\ref{lem:soundpartialds} we can assume that $Z
        \subseteq B$ and $\ds_H$ is equal to  	$\ds_{Z}$ on
        $\Proc_{Z}$.
        Observe that one of the locks of $p$ must belong to $B$, by
        the definition of "full deadlock scheme".
	So there must exist some outgoing edge from a lock of $p$, say $t$, in
        $\ds_H(\Proc_B)$.
        Since $\ds_Z(p)$ was undefined, and $\ds_H,\ds_Z$ coincide on
        $\Proc_Z$, the lock $t$ cannot belong to $Z$.
        
	By definition, every lock with an incoming edge in $\ds_H$ must also have an
	outgoing edge in $\ds_H$. 
	Following these edges we get a cycle in the image of $\ds_H$. 
	During the last iteration of lines 9-12, $Z$ did not increase, hence by Lemma~\ref{lem:reach} there are no edges from $T \setminus Z$ to $Z$.
	This cycle is therefore outside $Z$. It has to be a "weak cycle" by definition of
	a "deadlock scheme", which is a contradiction because~\cref{alg:other-cycles}
	did not increase $Z$ in its last application. 

	If the algorithm reaches the end then by \cref{Inv3} we know
        that a $Z$-"deadlock scheme" for $\PP$, say $\ds_Z$, exists.
	We construct a "full deadlock scheme" $(Z,\ds)$ in
        $\Gp$ as follows.
	First, we set $\ds(p)=\ds_Z(p)$ for all $p\in \PZ$.
	For $p\notin\PZ$, as the algorithm did not fail at
	lines 13-14, $p$ is $Z$"-fragile", and we let $\ds(p)$ undefined.
	
	Finally, this algorithm runs in polynomial time as all steps of all loops in the algorithms either decrease $H$ or increase $Z$. Furthermore, the condition on line 13 is easily verifiable by checking in the "behavior@@two" $(\PP_p^\s)_{p \in \Proc}$ of $\sigma$ whether there exists $\emptyset \pat B \in \PP_p$ such that $B \subseteq Z$.
\end{proof}

\twolssNP*
	
	\begin{proof}
		We start by guessing a "behavior@@two"
		$\PP=(\PP_p)_\pproc$ such that no $\PP_p$ contains any "pattern" of the form $\Owns_{p} \pat \emptyset$.
		Its  size is polynomial in the number of processes. 
		We can check in polynomial time  that there	exists
                a "strategy" respecting the
		"patterns" in $\PP$ by Lemma~\ref{lem:polycheckpatterns}.  
		Note that if there is one, by the requirement we made on $\PP$ it must be "locally live".

		If yes, then we  compute the "lock graph" $\Gp$ for $\PP$
		and check if there is a "full deadlock scheme"
                for $\PP$ in polynomial time by Lemma~\ref{lem:correctcompletealgo}.
		
		By Lemma~\ref{lem:winningiffnods}, this algorithm answers yes if and 
		only if the system has a  "locally live" "strategy"
                that avoids "deadlocks".
                More formally, if there exists a "winning", "locally
                live" "strategy" $\s$ then it suffices to guess the
                "behavior@@two" $\PP^\s$ and the Algorithm
                \ref{alg:completealgo} will return ``$\s$ is winning''.
                For the other direction, assume that we guess a
                "behavior@@two" $\PP$ such that no "pattern"
                $\Owns_p \pat \es$ belongs to $\PP_p$, for any $p$.
               Assume also that 
                Lemma~\ref{lem:polycheckpatterns} tells that there exists some "strategy" $\s$ such that
                $\PP^{\s}_p\subseteq \PP_p$ for every $p$.
                If Algorithm
                \ref{alg:completealgo} returns ``$\s$ is winning'' then
                by Lemma~\ref{lem:correctcompletealgo} we know that there
                is no "full deadlock scheme" for $\PP$, so there
                cannot be any for $\PP^\s$ either.
	\end{proof}


\subsection{Exclusive \tlss}
\label{sec:exclusive}

In this section we study "exclusive@@sys" \tlss. 
These systems enjoy enough properties to be able to decide 
the deadlock avoidance control problem with "locally live" strategies 
in polynomial time.

Recall that in an "exclusive@@sys" system, if a state has an outgoing $\get{t}$ transition,
then all its outgoing transitions are labeled with $\get{t}$. 
So in such a state the process is necessarily blocked until $t$
becomes available.

Behaviors of exclusive systems have some special properties, see
Lemma~\ref{lem:guardededges}.
First, whenever a "strategy" has a "strong pattern" 
$\set{t_1} \spat \set{t_2}$ for a process $p$, 
it also allows a reverse "weak pattern" $\set{t_2} \wpat \set{t_1}$. 
This will imply that the "strong cycle" condition in our "deadlock
schemes" can be satisfied automatically, because any "strong cycle" can be replaced by a reverse cycle of "weak edges".
Second, all processes that have some pattern are "fragile".

The above observations simplify the analysis of the "lock graph". 
First, we get a much simplified \NP\ argument
(Proposition~\ref{prop:sigma-winning}). 
This allows us to eliminate  guessing and obtain a \PTIME\
algorithm (Proposition~\ref{prop:winiffnoBTu}).

Throughout this section we fix an "exclusive@@sys" \tlss\ $\tss$, and consider only
"locally live" "strategies".
As we have seen in the previous section, whether or not a "strategy" $\s$ is
"winning" is determined by its "behavior@@two" $\PP^\s$. 
More precisely, $\s$ is "winning" if and only if $\PP^\s$ does not admit a "full deadlock
scheme", see Lemma~\ref{lem:winningiffnods}.
In this section we show that the latter property can be decided in \PTIME\ for
"exclusive@@sys" \tlss.

\begin{defi}
  We call a "behavior@@two" $\PP$ ""exclusive@@beh"" if
	\begin{itemize}
		\item  	whenever $\PP_p$ contains $\set{t_1}\spat\set{t_2}$ 
		then it contains either  $\set{t_1}\wpat\set{t_2}$ or $\set{t_2}\wpat\set{t_1}$, and 
		\item whenever  $\PP_p$ contains $\set{t_1}\pat\set{t_2}$ then $p$ is $\set{t_1,t_2}$"-fragile" in $\PP_p$.
	\end{itemize}
\end{defi}

\begin{rem}
	\label{rem:cycles-exclusive}
	Assume that we have a "strong cycle" $t_1 \strongedge{p_1} t_2
        \strongedge{p_2} \cdots \strongedge{p_k} t_{k+1} = t_1$ in the
        "lock graph" $\Gp$ of an "exclusive@@beh" "behavior@@two"
        $\PP$.
        Then by definition of "strong edges", every $p_i$ has "pattern" $\set{t_i} \spat \set{t_{i+1}}$ but not
        $\set{t_i} \wpat \set{t_{i+1}}$. Then by definition of "exclusive@@beh" "behavior@@two", they all have a pattern $\set{t_{i+1}} \wpat \set{t_i}$, hence there is a weak cycle $t_1=t_{k+1} \weakedge{p_k} \cdots  \weakedge{p_2} t_2 \weakedge{p_1} t_1$.
\end{rem}

\begin{lem}
	\label{lem:guardededges}
	\label{lem:EGallfragile}
	If $\s$ is a "locally live" strategy in an "exclusive@@sys" \tlss\, and 
	$\PP^\s=(\PP_p)_\pproc$ is its "behavior@@two", then $\PP^\s$ is
	"exclusive@@beh". 
\end{lem}

\begin{proof}
	Consider the first statement.
	Suppose there is a "strong pattern" $\set{t_1} \spat
        \set{t_2}$ in $\PP_p$, then there exists a local $\s$-run of $p$  of the form
        \[\run = \run_1 (a_1,
          \get{t_1}) \run_2 (a_1, \rel{t_2})\run_3 (a_3, \get{t_2})\,,\]
        with no $\rel{t_1}$ in $\run_2$ or $\run_3$. Hence, there is a point
	in the run at which $p$ holds both locks. 	
        In consequence, there must be two acquires in $\run$ with no release
	in-between. As the process is "exclusive@@proc", the state
        from which the second lock  is taken  is such that
        all outgoing transitions take this lock.
        Thus there is a "weak pattern" from the first lock taken to the
	second one. 
	For the second statement, suppose that $\set{t_1} \pat \set{t_2}$ is
        in $\PP_p$.
	Thus there exists a
	$\s$-run of $p$ making it acquire $t_1$, so there must be some
        $\s$-run $u$ of the form $u=u_1
	(a, \get{t_i}) u_2$ for some $i \in \set{1,2}$ and $u_1$ containing only
	local actions. As $\tss$ is "exclusive@@sys", this means that $u_1$ makes $p$ reach a
	configuration where all outgoing transitions acquire $t_i$, and $p$ owns
	no lock. Since $\s$ is "locally live" this means that $u_1$
        has the pattern $\es \pat \set{t_i}$, so $p$ is
        $\set{t_i}$"-fragile", hence also
	$\set{t_1,t_2}$"-fragile". 
\end{proof}

Now consider a decomposition of the lock graph $\Gp$ of a given
"behavior@@two" $\PP$ into strongly connected
components (SCC for short).
An SCC of $\Gp$ is a \intro{direct deadlock} if it contains a "simple cycle". 
A \intro{deadlock SCC} is a "direct deadlock SCC" or an SCC from which
a "direct deadlock SCC" can be reached.

Figure~\ref{fig:semideadlock} illustrates these concepts: 
the left graph has a "direct deadlock SCC" formed by the three locks at the top.
The two remaining locks form a "deadlock SCC", because there is a path
towards a direct deadlock SCC.
Observe that the two locks at the bottom are not a "direct deadlock SCC" because
there is only one process between the two locks and thus no "simple cycle"
within the SCC.  

\AP Let $\intro*\BTPP \subseteq T$ be the set of all locks appearing in some
"deadlock SCC" of $\Gp$.

\begin{prop}\label{prop:sigma-winning}
	Consider an  "exclusive@@beh" "behavior@@two" $\PP$. 
	There is a "full deadlock scheme" for $\PP$ if and only
        if  all processes in $\Proc$ are $\BTPP$"-fragile".
\end{prop}

The proof of Proposition~\ref{prop:sigma-winning} follows from the
lemmas below.
In all these lemmas we assume that $\PP$ is an  "exclusive@@beh"
"behavior@@two". 

\begin{lem}
	If all processes are $\BTPP$"-fragile" then there is a "full deadlock
	scheme" for $\PP$.
\end{lem}

\begin{proof}
	We construct a "deadlock scheme" for $\Gp$ as follows.
	For all "direct deadlock" SCCs we select a simple cycle inside.
	By Remark~\ref{rem:cycles-exclusive} and
        Lemma~\ref{lem:guardededges}, this cycle is "weak@@cycle" or
        has a reverse "weak cycle". We select a direction in which the
        cycle is weak, and for all $t$ in the cycle we set $p_t$ as
        the process labeling the edge outgoing from $t$ in the cycle.
	
	While there is some edge $t \edge{p} t'$ in $\Gp$ such that
        $p_t$ is not yet defined but $p_{t'}$ is, we set $p_t = p$. 
	When this ends we have defined $p_t$ for all locks $t \in \BTPP$. 
         We define $\ds$ as
	$\ds(p_t) = t \edge{p_t} t'$ for all $t \in \BTPP$.
        For all other $p \in \Proc$, $\ds(p)$ is undefined. 

        We show now that $\ds$ is a "full deadlock scheme" for $\PP$.
	Clearly, for all $p \in \Proc$, if $\ds(p)$ is defined
	then it is a $p$-labeled edge of $\Gp$. 
	Furthermore, as all processes are $\BTPP$"-fragile", in
        particular all processes $p$ 
	with $\ds(p)$ undefined are $\BTPP$"-fragile".
	It is also clear that all locks of $\BTPP$ have a unique outgoing edge.
	Finally, by construction we ensured that $\ds$ has no "strong cycle".
\end{proof}

\begin{lem}\label{lem:BT-biggest}
	Any "full $Z$-deadlock scheme" for $\PP$ is such that
	$Z \subseteq \BTPP$.
\end{lem}

\begin{proof}
  Suppose that $\ds_Z$ "full $Z$-deadlock scheme" for $\PP$.
  If there is some $t \in Z \setminus \BTPP$, 
	then there exists $p$ such that $\ds_Z(p) = t \edge{p} t'$, for
        some $t' \in Z$. 
	By definition of $\BTPP$, there are no edges from 
	$T \setminus \BTPP$ to $\BTPP$
	in $\Gp$, hence $t' \in Z\setminus \BTPP$. 
	By iterating this process we eventually find a "simple cycle" 
	in $\Gp$ outside of $\BTPP$, 
	which is impossible, as this cycle should be part of a 
	"direct deadlock SCC", and thus included in $\BTPP$.
\end{proof}

\begin{lem}
	If some process $p$ is not $\BTPP$"-fragile" then there is no "full
	deadlock scheme" for \nolinebreak $\PP$.
\end{lem}

\begin{proof}
	Suppose there exists $p$ that is not $\BTPP$"-fragile". 
	Towards a contradiction assume that there is some
	"full $Z$-deadlock scheme" $\ds_Z$ for $\PP$, for some $Z$.
	
	As $p$ is not $\BTPP$"-fragile", then by
	Lemma~\ref{lem:BT-biggest} it is not $Z$"-fragile" either.
	Hence, $\ds(p)$ is an edge $t_1 \edge{p} t_2$ in $\Gp$, with $t_1, t_2 \in Z$, and thus $t_1, t_2 \in \BTPP$. 	
	By Lemma~\ref{lem:EGallfragile}, $p$ is $\set{t_1,t_2}$"-fragile", 
	and therefore also $\BTPP$"-fragile", yielding a contradiction.
\end{proof}
This concludes the proof of Proposition~\ref{prop:sigma-winning}.
\medskip

\paragraph{Deciding the existence of a winning strategy for
  "exclusive@@sys"  systems.}

Until now we have assumed that we were given a "strategy" $\s$, and we
described how to check if it is "winning", by constructing $\BTPP$ and
checking that every process is $\BTPP$"-fragile", where $\PP=\PP^\s$.
Now we want to decide if there is any "winning" "strategy". 
We use the insights above, but we cannot simply enumerate all
"exclusive@@beh" "behaviors@@two", as they are exponentially many. 

We say below that a "strategy" $\s$ for $p$ ""induces"" the edge $t_1
\act{p} t_2$ if $\s_p$ admits the "pattern" $\set{t_1} \pat \set{t_2}$.

For every  process $p$ and every set of edges between two locks of $p$ we
can check if there is a "strategy" for $p$ inducing only edges within
this set, as a consequence of Lemma~\ref{lem:polycheckpatterns}.

We call an edge $t_1 \edge{p} t_2$ \intro{unavoidable} if it is
"induced" by every "locally live" "strategy" of $p$.

Let $\intro*\Gu$ be the graph whose nodes are locks and whose edges
are the "unavoidable" edges.
We will compute a set of locks $\BTu$ in a similar way as $\BTPP$ in the previous section except that
we will use slightly more general basic SCCs of $\Gu$.

\AP A \intro{direct semi-deadlock SCC} of $\Gu$ is either a "direct deadlock SCC", or an SCC containing only double edges, with two locks $t_1$ and $t_2$ such that
for some process $p$ using $t_1$ and $t_2$, every strategy for $p$ induces at
least one edge between $t_1$ and $t_2$.
Then a \intro{semi-deadlock SCC} of $\Gu$ is either a "direct semi-deadlock
SCC" or an SCC from which a "direct semi-deadlock SCC" can be reached.

Let $\intro*\BTu$ be the set of locks appearing in "semi-deadlock SCCs".

\begin{figure}
	\begin{tikzpicture}
		
		\tikzset{enclosed/.style={draw, circle, inner sep=2pt, minimum size=.5cm, fill=black}}
		
		\node[enclosed, state,fill=black,minimum size = .5mm] at (5,1.3) (L) {};
		\node[enclosed,state,fill=black,minimum size = .5mm] at (4,-2) (H) {};
		\node[enclosed,state,fill=black,minimum size = .5mm] at (6,-2) (I) {};
		\node[enclosed,state,fill=black,minimum size = .5mm] at (4,0) (J) {};
		\node[enclosed,state,fill=black,minimum size = .5mm] at (6,0) (K) {};
		
		\node[enclosed,state,fill=black,minimum size = .5mm] at (8,-2) (D) {};
		\node[enclosed,state,fill=black,minimum size = .5mm] at (10,-2) (E) {};
		\node[enclosed,state,fill=black,minimum size = .5mm] at (8,0) (F) {};
		\node[enclosed,state,fill=black,minimum size = .5mm] at (10,0) (G) {};
		
		\path[->, thick, >= angle 60, bend left =20] (D) edge node[below] {$p_3$} (E);
		\path[->, thick, >= angle 60, bend left =20] (E) edge (D);
		\path[->, thick, >= angle 60, bend left =20] (F) edge node[below] {$p_1$} (G);
		\path[->, thick, >= angle 60, bend left =20] (G) edge (F);
		\path[->, thick, >= angle 60, bend left =20] (D) edge node[right=-3pt] {$p_4$} (F);
		\path[->, thick, >= angle 60, bend left =20] (F) edge (D);
		\path[->, >= angle 60, bend left =20, color=blue] (E) edge node[right=-2pt] {\color{black}$p_2$} (G);
		\path[->, >= angle 60, bend left =20, color=blue] (G) edge (E);
		
		\path[->, thick, >= angle 60, bend left =20] (H) edge node[below] {$p_3$} (I);
		\path[->, thick, >= angle 60, bend left =20] (I) edge (H);
		\path[->, thick, >= angle 60, bend right =20] (J) edge node[above] {$p_1$} (K);
		\path[->, thick, >= angle 60, bend right =20] (K) edge node[below left=-2pt] {$p_6$} (L);
		\path[->, thick, >= angle 60, bend right =20] (L) edge node[below right=-2pt] {$p_5$} (J);
		\path[->, >= angle 60, bend left =20] (H) edge node[right=-3pt] {$p_4$} (J);
		\path[->, >= angle 60, bend left =20] (I) edge node[right=-2pt] {$p_2$} (K);
	\end{tikzpicture}
	\caption{"Semi-deadlock SCCs": the blue double edge is not in $\Gu$, but every strategy of the system will induce one of those two edges.}
	\label{fig:semideadlock}
\end{figure}

In the  graph on the right of Figure~\ref{fig:semideadlock} the black edges are in
$\Gu$, the double blue ones are not, but indicate that every
"strategy" $\s$ of
process $p_2$ induces one of the two blue edges in $G_{\PP^\s}$.
The four locks do not form a "direct deadlock SCC" of $\Gu$ as there is no "simple
cycle" (without the blue edges, which do not belong to $\Gu$).
However they do form a "direct semi-deadlock" SCC, as $p_2$ will induce an edge
no matter its strategy, forming a "simple cycle".

\begin{prop}
	\label{prop:winiffnoBTu}
	There is a "winning" "strategy" for deadlock avoidance iff there exists some process $p$ and a "local
	strategy" $\sigma_p$ that prevents $p$ from acquiring a lock from $\BTu$.
\end{prop}

\begin{proof}
One direction is easy: if all "strategies" make all processes acquire a lock 
from $\BTu$ then there is no "winning" "strategy". 
Let $\s$ be a "strategy", $\PP=\PP^\s$ its "behavior@@two" and $\Gp$ its "lock graph".
Note that $\Gu$ is a subgraph of $\Gp$, hence 
every SCC in $\Gp$ is a superset of an SCC in $\Gu$. 
Observe that if an SCC in $\Gp$ contains a 
"direct semi-deadlock SCC" of $\Gu$ then it is a "direct deadlock SCC". 
Indeed, if an SCC in $\Gu$ is a "direct semi-deadlock" but not a
"direct deadlock" one then $\s$ adds one of edges between the locks of
$p$, say edge $t_1 \edge{p} t_2$, to this SCC in $\Gp$.
As $t_1, t_2$ are in that SCC of $\Gu$, 
there is a "simple" path from $t_2$ to $t_1$ not involving $p$. 
Hence, a "direct semi-deadlock SCC" becomes a "direct deadlock SCC".
This implies $\BTu \incl \BTPP$.

Let $p \in \Proc$, as there is a $\s$-run of $p$ acquiring a lock of $\BTu$,
either $p$ is $\BTu$"-fragile" (and thus $\BTPP$"-fragile") or there is an edge 
labeled by $p$ towards $\BTu$, meaning that both locks of $p$ are in $\BTPP$ 
and thus that $p$ is $\BTPP$"-fragile" by Lemma~\ref{lem:EGallfragile}.  
As a consequence, all processes are $\BTPP$"-fragile". 
We conclude by Proposition~\ref{prop:sigma-winning}.

In the other direction we suppose that there exists a process $p$ and 
a strategy $\sigma_p$ forbidding $p$ to acquire any lock of $\BTu$.
We construct a strategy $\s$ such that $p$ is not $\BTPP$"-fragile".
This will show that $\s$ is winning by Proposition~\ref{prop:sigma-winning}.

\AP Let $\intro*\FTu=T\setminus \BTu$ be the set of locks not in $\BTu$. 
By definition of $\BTu$, in $\Gu$ no node of $\FTu$ can reach a "direct semi-deadlock SCC". 
In particular, there is no "direct semi-deadlock SCC" in 
$\Gu$ restricted to $\FTu$.
We construct a "strategy" $\s$ such that, when restricted to $\FTu$, the SCCs of
$\Gp$ and $\Gu$ are the same, where $\PP=\PP^\s$.

Let us linearly order the SCC  of $\Gu$ restricted to $\FTu$ in such a way that if a
component $C_1$ can reach a component $C_2$ then $C_1$ is before $C_2$ in the
order. 

We use strategy $\s_p$ for $p$.
For every process $q \neq p$ we have one of the two cases: (i) either
there is a "local strategy" $\s_q$ inducing only the edges that are already in
$\Gu$; or (ii) every "local strategy" induces some edge that is not in $\Gu$.
In the second case there are no $q$-labeled edges in $\Gu$, and for each of the
two possible edges there is a "local strategy" inducing only this edge.

For a process $q$ from the first case we take a "local strategy" $\s_q$ that induces
only the edges present in $\Gu$.

For a process $q$ from the second case,
\begin{itemize}
	\item If both locks of $q$ are in $\BTu$ then take any "local strategy" for $q$.
	\item If one of the locks of $q$ is in $\BTu$ 
	and the other in $\FTu$ then choose a "strategy"
	inducing an edge from the lock in $\BTu$  to the lock in $\FTu$. 
	\item If both locks of $q$ are in $\FTu$  then choose a "strategy" inducing
          an edge from a smaller to a bigger SCC of $\Gu$.
\end{itemize}

In the last case, both locks cannot be in the same SCC of $\Gu$:  As they are in $\FTu$, this would have to be an SCC with no simple cycles, i.e., a tree of double edges. But then the existence of $q$ implies that this is a "direct semi-deadlock SCC", which contradicts the fact that those locks are in $\FTu$.

Consider the graph $\Gp$ of the resulting strategy $\s$.
Restricted to $\FTu$ this graph has the same SCCs as $\Gu$. 
Moreover, there are no extra edges in $\Gp$ added to any SCC 
included in $\FTu$, and there are no edges
from $\FTu$ to $\BTu$. As a result, we have $\BTu = \BTPP$.
As $p$ acquires no lock from $\BTu$, it is not $\BTu$"-fragile" and thus not
$\BTPP$"-fragile" either.
  \end{proof}

\exclusiveP*

\begin{proof}
  First we need to compute the "unavoidable" edges.
  An edge $t_1 \act{p} t_2$ is avoidable iff there exists some
  "locally live"
  strategy $\s_p$ that does not admit the pattern $\set{t_1} \pat \set{t_2}$.
  Recall that we assume that $\Aa_p$ is "lock-aware".
  Then the above means that we look for a "locally live" strategy
  $\s_p$ that avoids all states in $\Aa_p$ where $p$ owns $t_1$ and
  needs to acquire $t_2$.
  This question reduces to a usual safety game.

  Next we have to determine  which SCCs of $\Gu$ are a "direct
  deadlock", which amounts to check the existence of a "simple cycle".
  Observe that an SCC does not contain such a cycle iff it is a tree of
  double edges, which is easy to check in \PTIME.
 Knowing whether an SCC is a "direct semi-deadlock" or a
 "semi-deadlock" can also be done in \PTIME.

 Finally we have to check the condition from
 Proposition~\ref{prop:winiffnoBTu}, so the existence of a "locally
 live" strategy $\s_p$ that prevents $p$ to take a look from $\BTu$.
 As above, this amounts to a safety game.
\end{proof}

\section{Nested locks}\label{sec:nested}
	
	We consider now  "nested-locking" "LSS", in which the system
        has to ensure that locks are acquired and released 
	in a stack-like manner.
	So a process can release only the last lock it has acquired.

          Throughout the section the action associated with each operation
  is omitted, to simplify the presentation.
        
Local runs of "nested-locking" "LSS" have a natural decomposition
into staircases:

\begin{defi}
  A \intro{stair decomposition} of a local run $u$ is of the form
  \[u=u_1 \,\get{t_1} u_2 \, \get{t_2} \cdots u_{k} \, \get{t_k} u_{k+1}\]
  
  \noindent where $u_1,\dots,u_{k+1}$ are "neutral" runs,
  and no $u_i$ uses locks  from $\set{t_1,\dots,t_{i-1}}$. 
\end{defi}

\begin{lem}
	\label{lem:NTSSproperty}
	Every "nested-locking" local run $\run$
	has a unique "stair decomposition". 
\end{lem}

\begin{proof}
	We set 
	$\run = \run_1 \get{t_1} \run_2 \get{t_2} \cdots \run_{k}
	\get{t_k} \run_{k+1}$
	such that $\set{t_1, \ldots, t_k}$ is the set of
	locks held by the process, call it $p$, at the end of the run $u$, and the distinguished	
	$\get{t_i}$ are the last acquisitions of these locks in $\run$.
        All properties are immediate.
\end{proof}
	
We now define patterns of "risky" local runs that will serve as witnesses of 
reachable "deadlocks", in a similar manner as in Definition~\ref{def:patterns2locks}.	

\begin{defi}\label{def:nested-pattrerns}
  Consider a local "risky" $\s$-run $u$ of process $p$, and its "stair decomposition" $u=u_1 \,\get{t_1} u_2 \, \get{t_2}
	\cdots u_{k} \,  \get{t_k} u_{k+1}$.
We associate with $u$ a \intro{stair
	pattern} $(\intro*\Towns_p,\intro*\Twants_p,\preceq^p)$, where $\Towns_p=\set{t_1,\dots,t_k}$, 
	$\Twants_p$ is the set of locks requested by outgoing
        transitions allowed by $\s$ in the 
	state reached by $u$,
	and $\preceq^p$ is the smallest partial order on $T_p$ satisfying: 

  \begin{quote}
    For all $1 \le i \le k$ and all $t \in T_p$, if the last operation
    on $t$ in $u$ is after the last $\get{t_i}$  then $t_i \preceq^p t$.
  \end{quote}

	A \intro(nest){behavior} of $\s$ is a
	family of sets of "stair patterns" 
  $(\PP^\s_p)_{p\in\Proc}$, where $\PP^\s_p$ is the set of "stair patterns" of
  local "risky" $\s$-runs of $p$.
\end{defi}

\begin{exa}
	Consider the local run displayed in Figure~\ref{fig:stair}.
	It is "nested-locking" and "risky", hence we can define its "stair pattern" $(\set{t_1,t_2,t_4}, (t_2<t_1<t_3<t_5<t_4), \set{t_3,t_5})$.
	This pattern describes the set of locks held at the end, the order on the last operations on each lock appearing in the run, and the set of locks that can be acquired at the end. 
\end{exa}

	\begin{figure}[ht]
	\begin{tikzpicture}[AUT style]
		\node[state, minimum size=10pt, fill=blue!5!white] (A) {};
		\node[state, right=of A, minimum size=10pt, fill=blue!5!white] (B) {};
		\node[state, right=of B, minimum size=10pt, fill=blue!5!white] (C) {};
		\node[state, right=of C, minimum size=10pt, fill=blue!5!white] (D) {};
		\node[state, right=of D, minimum size=10pt, fill=blue!5!white] (E) {};
		\node[state, right=of E, minimum size=10pt, fill=blue!5!white] (F) {};
		\node[state, right=of F, minimum size=10pt, fill=blue!5!white] (I) {};
		\node[state, right=of I, minimum size=10pt, fill=blue!5!white] (J) {};
		\node[state, right=of J, minimum size=10pt, fill=blue!5!white] (G) {};
		\node[state, right=of G, minimum size=10pt, fill=blue!5!white] (H) {};
		\node[state, above right=of H, minimum size=10pt, fill=blue!5!white] (X) {};
		\node[state, below right=of H, minimum size=10pt, fill=blue!5!white] (Y) {};
		
		\path[->, thick] (A) edge node[above] {$\get{t_1}$} (B);
		\path[->, thick] (B) edge node[above] {$\rel{t_1}$} (C);
		\path[->, thick] (C) edge node[above] {$\get{t_2}$} node[below] {\color{blue} $\blackdiamond$} (D);
		\path[->, thick] (D) edge node[above]  {$\get{t_1}$} node[below] {\color{blue} $\blackdiamond$} (E);
		\path[->, thick] (E) edge node[above] {$\get{t_3}$} (F);
		\path[->, thick] (J) edge node[above] {$\rel{t_3}$} (G);
		\path[->, thick] (G) edge node[above] {$\get{t_4}$} node[below] {\color{blue} $\blackdiamond$}   (H);
		\path[->, thick] (F) edge node[above] {$\get{t_5}$} (I);
		\path[->, thick] (I) edge node[above] {$\rel{t_5}$} (J);
		\path[->, dotted] (H) edge node[above left] {$\get{t_3}$} (X);
		\path[->, dotted] (H) edge node[below left] {$\get{t_5}$} (Y);
		
		\path[draw,very thick] (0,-3) -- node[above] {$t_1$} (2,-3) -- (2,-2.5) -- node[above] {} (4,-2.5) -- (4,-2) -- node[above] {$t_3,t_5$} (6,-2) -- (6,-1.5) -- (8,-1.5);
		
		\path[draw, dashed] (2,-2.5) -- (2,-3.5);
		\path[draw, dashed] (4,-2) -- (4,-3.5);
		\path[draw, dashed] (6,-1.5) -- (6,-3.5);
		
		\node (1) at (2,-4) {$t_2$};
		\node (2) at (4,-4) {$t_1$};
		\node (3) at (6,-4) {$t_4$};
	\end{tikzpicture}
	\caption{Example of a "nested-locking" local run. The dotted arrows are the available transitions at the end of the run.
		Blue diamonds mark transitions taking a lock
		that is not released later in the run.
		The lower part shows the "stair pattern" of this run (without the $\Blocks$ part). Steps represents the points at which a lock is taken and not released later. On each step we write the set of locks used in the corresponding section of the run.}
	\label{fig:stair}
\end{figure}

\begin{lem}
  \label{lem:winiffnopatterns}
  A "control strategy" $\sigma$ with "behavior@@nest" $(\PP^\s_p)_{p \in \Proc}$ is 
  \textbf{not} "winning" if and only if for every $p\in\Proc$ there is some "stair 
  pattern" $(\Towns_p, \Twants_p, \preceq^p) \in \PP^\s_p$ such that:
  \begin{itemize}
  	\item $\bigcup_{p \in \Proc} \Twants_p \subseteq \bigcup_{p \in \Proc} 
  	\Towns_p$,
    \item the sets $\Towns_p$ are pairwise disjoint, 
    \item there exists a total order $\preceq$ on the set of all locks
      that is    compatible with all $\preceq^p$. 
  \end{itemize}
\end{lem}

\begin{proof}
	Suppose $\sigma$ is not "winning", and let $w$ be a run leading to a "deadlock". 
	For all $p$ let $\Towns_p$ be the set of locks owned by $p$ after $w$. 
	Let $\run^p=w|_p$ be the local run of $p$ in $w$. 
	Since $w$ leads to a "deadlock" every $\run^p$ is "risky". 
	For every $p$, consider the "stair pattern" $(\Towns_p,\Twants_p, \preceq^p)$ of $\run^p$. 
	By definition, this is a "pattern@@nest" from $\PP^\s_p$.
	
	We need to show that these patterns satisfy the requirements of the lemma.
	Since the configuration reached after $w$ is a "deadlock", every process waits
	for locks that are already taken so $\bigcup_p \Twants_p \subseteq \bigcup_{p} \Towns_p$, proving the first condition.
	Moreover, the sets $\Towns_p$ are pairwise disjoint. 
	
	For the last requirement of the lemma consider some order $\preceq$ on $T$ satisfying:
	$t\preceq t'$ if the last operation on $t$ appears before the
	last operation on $t'$ in $w$. 	
	Let $p \in \Proc$, let 
	$\run^p=\run^p_1 \, \get{t^p_1} \run^p_2 \, \get{t^p_2} \cdots \run^p_{k}\, \get{t^p_k} \run^p_{k+1}$ be the "stair decomposition" of $\run^p$. 
	As $p$ never releases  $t^p_i$, the distinguished $\get{t^p_i}$,
	is the last operation on $t^p_i$ in the global run.
	Consequently, for all $t$ we have $t^p_i \preceq t$ whenever $t$ is used in
	$\run^p_{i+1} \get{t^p_{i+1}} \cdots \run^p_{k} \get{t^p_k} \run^p_{k+1}$.
	Hence, $\preceq$ is compatible with all $\preceq^p$. 
	
	For the converse implication, suppose that there are "patterns@@nest" satisfying
	all the conditions of the lemma. 
	We need to construct a run $w$ ending in a "deadlock".
	For every process $p$ we have a "stair pattern" $(\Towns_p,\Twants_p, \preceq^p)$ coming
	from a local $\s$-run $u^p$ of $p$, with
	$u^p=u^p_1\, \get{t^p_1} u^p_2 \, \get{t^p_2} \cdots u^p_{k}\, \get{t^p_k} u^p_{k+1}$
	as "stair decomposition".
	There is also a linear order $\preceq$ compatible with all $\preceq^p$.
	Let $\prec$ be its strict part.
	Let $t_1,\dots, t_k$ be the sequence of locks from
	$\bigcup_{p}\Towns_p$  listed according to $\prec$.
	Let $\set{p_1, \ldots, p_n} = \Proc$.
	We claim that we can get a suitable global run $w$ as $u^{p_1}_1\dots u^{p_n}_1w'$
	where $w'$ is obtained from $t_1\dots t_k$ by substituting each
	$t^p_i$ by $\get{t^p_i}u^p_{i+1}$.
	Observe that every $t_j$ from the sequence $t_1\dots t_k$ corresponds to
	exactly one $t_i^p$, as the sets $\Towns_{p_1},\dots,\Towns_{p_n}$ are disjoint.
	
	All $u^p_{i}$ are "neutral", hence after executing $u^{p_1}_1\dots u^{p_n}_1$ all locks are free. 
	Let $t^p_i \in T_p$, suppose furthermore that all
        $\get{t^q_j}u^q_{j+1}$ with $t^q_j \prec t^p_i$ have been
        executed after $u^{p_1}_1\dots u^{p_n}_1$.
	Then the set of taken locks is $\set{t^q_j \mid t^q_j \prec t^p_i}$. 
	As $\preceq$ is compatible with all $\preceq^p$, all locks $t$ used in $\get{t^p_i}u^p_{i+1}$ are such that $t^p_i \preceq t$.
	Moreover, since all $t^q_j$ that were taken before are such that $t^q_j \prec t^p_i$, 
	the run $\get{t^p_i}u^p_{i+1}$  uses only locks that are free and can therefore 
	be executed.
	
	To sum up, $w$ can be executed. It ends in a "deadlock" as
        $\bigcup_p\Twants_p \subseteq \bigcup_{p}\Towns_p$.
\end{proof}

\begin{lem}
  \label{lem:checkpatternNTSS}
  Given a "nested-locking" "LSS" $\tss$, a process $p \in \Proc$ and a set of
  "patterns@@nest" $\PP_p$, we can check in polynomial time in
  $\size{\aut_p}$ and $2^{\size{T}\log(\size{T})}$ whether there exists a "strategy"
  $\sigma$ 
  with $\PP^\s_p \subseteq \PP_p$.
\end{lem}
\begin{proof}
	Fix a process $p$.
	We extend the states of $p$ to keep track of the set of locks held by $p$ as well as the order $\fleq$ induced  by the "stair
	pattern" of the run seen so far (as in Definition~\ref{def:nested-pattrerns}).
	This increases the number of states by the factor
        $|T|!\cdot2^{|T|}$.
        
	As the set of locks owned by $p$ is now a function of the current state, 
	this also allows us to eliminate all non-realizable transitions which acquire a 
	lock that $p$ owns or release one it does not have.

	Consider a state $s$ where all outgoing transitions have a
        lock acquisition as operation. 
	Thanks to the previous paragraph, $s$ determines the set of locks $\Towns(s)$
	and an order $\prec_s$ such that every local run ending in $s$ has a pattern
	$(\Towns^s,B,\prec_s)$, where $B$ depends on the choices a strategy for $p$
	makes in $s$.
	We mark $s$ bad if none of these possible patterns is in \nolinebreak $\PP_p$.
	
	We iteratively delete all bad states and all their ingoing transitions, as we need to ensure that we never reach them. If we delete an "uncontrollable" transition then we mark its source state as bad because
	reaching that state would make the environment able to reach a bad state.
	If this process marks the initial state bad then there is no local strategy
	with patterns  included in $\PP_p$.
	Otherwise, we look for new bad states as in the previous paragraph.
	Indeed, a state may satisfy the conditions of the previous paragraph after
	removing some of its outgoing transitions, for example a transition not
	accessing locks. 
	If some new state is marked bad then we repeat the whole procedure. 

	When this double loop stabilizes and if the initial state is not marked bad, then the
	remaining transitions form a  strategy for $p$ with all patterns in $\PP_p$.
\end{proof}

\begin{prop}
	\label{prop:innexptime}
	The "deadlock avoidance control problem" is decidable for
	"nested-locking" "lock-sharing systems" in non-deterministic exponential time.
\end{prop}

\begin{proof}
	First we apply the first step of Lemma~\ref{lem:checkpatternNTSS}
        so that every state encodes which locks are taken, and in
        which order.
        In this way we ensure that every release is applied in nested
        manner.
        
	The decision procedure for the existence of a winning strategy guesses a "behavior@@nest" $\PP_p$ for each process
	$p$.
	The size of the guess is at most $\pow{2\size{T}} \cdot \size{T}!$, hence $\pow{O(\size{T}
		\log(\size{T})}$.
	Then it checks if there exist "local strategies" yielding subsets of those "behaviors@@nest".
	This takes exponential time by Lemma~\ref{lem:checkpatternNTSS}.  
	If the result is negative then the procedure rejects.
	Otherwise, it checks if some condition from Lemma~\ref{lem:winiffnopatterns} does
	not hold.
	It it finds one then it accepts, otherwise it rejects.
	
	Clearly, if there is a "winning" "strategy" then the procedure can accept
	by guessing the family of "behaviors@@nest" corresponding to this "strategy". 
	For these "behaviors@@nest" the check from Lemma~\ref{lem:checkpatternNTSS} does not fail,
	and one of the conditions of Lemma~\ref{lem:winiffnopatterns} must be violated.
	
	Conversely, if the decision procedure concludes that there exists a
	"winning" "strategy", then let $(\PP_p)_{p \in \Proc}$ be the
	guessed family of "behaviors@@nest".
	We know that  there exists a "strategy" $\sigma$ with "behaviors@@nest" $(\PP'_p)_{p \in \Proc}$ such that $\PP'_p \subseteq \PP_p$ for all $p \in \Proc$. 
	Furthermore, as there are no "patterns@@nest" in $(\PP_p)_{p\in \Proc}$ satisfying the
	requirements of Lemma~\ref{lem:winiffnopatterns}, there cannot be any in
	the $\PP'_p$ either.
	Hence $\sigma$ is a "winning" "strategy".
\end{proof}

\nestedNEXP*

\begin{proof}
	The upper bound is given by Proposition~\ref{prop:innexptime}.
	For the lower bound, we reduce from the domino tiling problem 
	over an exponential grid. 
	In this problem, we are given an alphabet $\Sigma$ with a special letter
	$b$,
	an integer $n$ (in unary) 
	and a set $D$ of dominoes, 
	each domino $d$ being a 4-tuple $(\up_d, \down_d, \rgt_d, \lft_d)$
	of letters of $\Sigma$. 
	The question is whether there exists a mapping
	$t : \set{0, \ldots, 2^{n}-1}^2 \to D$ representing a valid
        tiling of the grid, i.e.~such that 
	for all $x,y,x',y' \in \set{0, \ldots, 2^{n}-1}$: 
	\begin{itemize}
		\item if $x'=x$ and $y'=y+1$ then $\up_{t(x,y)} = \down_{t(x',y')}$
		\item if $x' = x+1$ and $y'=y$ then $\rgt_{t(x,y)} = \lft_{t(x',y')}$
		\item if $x = 0$ then $\lft_{t(x,y)}= b$
		\item if $x = 2^n-1$ then $\rgt_{t(x,y)} = b$
		\item if $y = 0$ then $\down_{t(x,y)} = b$
		\item if $y = 2^n-1$ then $\up_{t(x,y)} = b$
	\end{itemize}
	
	The above problem is well-known to be \NEXPTIME-complete.
	
	Let $n, \Sigma, D, b$ be an instance of the tiling problem. 
	We construct a \lss\ as follows:
	We have three processes $p$, $\bar{p}$ and $q$. 
	Process $p$ uses locks from $\set{0^x_i, 1^x_i, 0^y_i$, $1^y_i
          \mid 1 \leq i \leq
	n}$, together with a lock $t_d$ for each domino $d \in D$, and an extra lock called
	simply $\ell$. 
	Process $\bar{p}$ will use similar locks but with a bar:
	$\bar{0^x_i}$, 	$\bar{1^x_i}$, $\bar{0^y_i}$, $\bar{1^y_i}$,
	$\bar{t_d}$, $\bar{\ell}$.
	Process $q$ will use all the locks of $p$ and $\bar{p}$.

	Let us describe process $q$ represented in Figure~\ref{fig:proc-q}.
	In the initial state the environment can choose between several
	actions:  $\equ$, $\ver$, $\hor$, $b_{\lft}$, $b_{\rgt}$, $b_{\up}$ and $b_{\down}$.
	Each of these actions leads to a different transition system, but the
	principle behind all the systems is the same.
	In the first phase, for each $1 \leq i \leq n$, the environment can choose to take 
	either lock $0^x_i$ or $1^x_i$, and then take either $\bar{0^x_i}$ or
	$\bar{1^x_i}$.
	In the second phase the same happens for $y$ locks.
	After these two phases the environment has chosen two pairs of $n$-bit numbers, call
	them $\#x,\#y$ and $\#\bx,\#\bary$. 
	Where the three systems differ is how the choice of $\bx$'s and $\bary$'s is
	limited in these two phases. 
	This  depends on the first action done by the environment:
	\begin{itemize}
		\item If it is $\equ$ then  $\#x=\#\bx$ and $\#y=\#\bary$.
		\item If it is $\ver$, then $\#x=\#\bx$ and $\#y+1=\#\bary$.
		\item If it is $\hor$, then $\#x+1 =\#\bx$ and $\#y =\#\bary$.
		\item If it is $b_{\lft}$ (resp.~$b_{\rgt}$) then $\#x=0$ (resp.~$\#x=2^n-1$).
		\item If it is $b_{\down}$ (resp.~$b_{\up}$) then $\#y=0$ (resp.~$\#y=2^n-1$).
	\end{itemize}
	All these constraints are easily implemented.
	For example, after $\equ$  the environment must take the same bits
	for $\bar{x}$ as for $x$ (similarly for $y$). 
	
	In the third phase, process $q$ has to take and then
        immediately release locks $\ell$ 
	and $\bar{\ell}$, before it reaches a state called $\domi$.
	Note that every state in the three phases before $\domi$ has a  loop on it, 
	meaning that $q$ cannot "deadlock" while being in one of these states.
	In state $\domi$, the system chooses to take two dominoes $d$ and $\bar{d}$ such that:
	\begin{itemize}
		\item If the environment has chosen $\equ$ then $d=\bar{d}$.
		
		\item If it has chosen $\ver$ then $\up_d = \down_{\bar{d}}$.
		
		\item If it has chosen $\hor$ then $\rgt_d = \lft_{\bar{d}}$.
		
		\item If it has chosen $b_{\lft}$ (resp.~$b_{\rgt}, b_{\up}, b_{\down}$) then $\lft_d = b$ (resp.~$\rgt_d, \up_d, \down_d$).
	\end{itemize}
	
	Each choice leads to a different state $s_{d, \bar{d}}$.
	From there transitions force the system to take every lock $t_{d'}$ and
	$\bar{t_{d'}}$,  except for $t_{d}$ and $t_{\bar{d}}$, in order
        to reach a state called
	$\win$ with a local loop on it and no other outgoing transitions.  
	
	We now describe process $p$ represented in Figure~\ref{fig:proc-q}.
	It starts by taking the lock $\ell$, which it never releases. 
	Then the environment chooses to take one of $0^x_i$ and $1^x_i$ and 
	one of $0^y_i$ and $1^y_i$ for all $1 \leq i \leq n$. 
	Finally, the system chooses a domino $d$ and takes the lock $t_d$ 
	before reaching a state with no outgoing transitions.	 
	Process $\bar{p}$ behaves identically, but uses locks with a bar.

	\begin{figure}[ht]
	\begin{tikzpicture}[xscale=1.5,yscale=1.3,AUT style]
		\node[state,initial, fill=blue!5!white] (0) at (0,0) {};
		\node[state, fill=blue!5!white] (1) at (1.5,0) {};
		\node[state, fill=blue!5!white] (2) at (3,0) {};
		\node[state, fill=blue!5!white] (3) at (4.5,0) {};
		\node[state, fill=blue!5!white] (4) at (6,0) {};
		\node[state, fill=blue!5!white] (5) at (7.5,0.5) {$d_1$};
		\node[state, fill=blue!5!white] (5') at (7.5,-0.5) {$d_m$};
		\node (6) at (3.75,0) {\huge $\cdots$};
		\node (7) at (7.5,0.1) {\huge $\vdots$};

		\path[->] (0) edge node[below] {$\get{\ell}$} (1);
		\path[->, bend left] (1) edge node[above] {$\get{0^x_1}$} (2);
		\path[->, bend right] (1) edge node[below] {$\get{1^x_1}$} (2);
		\path[->, bend left] (3) edge node[above] {$\get{0^y_n}$} (4);
		\path[->, bend right] (3) edge node[below] {$\get{1^y_n}$} (4);
		\path[->, bend left, dashed] (4) edge node[above] {$\get{t_{d_1}}$} (5);
		\path[->, bend right, dashed] (4) edge node[below] {$\get{t_{d_m}}$} (5');
	\end{tikzpicture}
	\caption{Transition system for process $p$ for the proof of
		Theorem~\ref{thm:nested} (with $D = \set{d_1, \ldots,
			d_m}$). Dashed arrows
		are controlled by the system.}
	\label{fig:proc-p}
\end{figure} 
	
	\begin{figure}[ht]
	\begin{tikzpicture}[xscale=1.3,yscale=1,AUT style]
		
		\node[state,initial, fill=blue!5!white] (0) at (-0.3,0.5) {};
		\node[state, fill=blue!5!white] (1) at (1.5,0.5) {};
		\node[state, fill=blue!5!white] (2) at (3,0.5) {};
		\node[state, fill=blue!5!white] (3) at (4.3,0.5) {};
		\node[state, fill=blue!5!white] (4) at (5.8,0.5) {};
		\node[state, fill=blue!5!white] (5) at (7.1,1.5) {$1$};
		\node[state, fill=blue!5!white] (k) at (7.1,0.5) {$k$};
		\node[state, fill=blue!5!white] (5') at (7.1,-0.5) {$n$};
		\node (6) at (3.65,0.5) {\large $\cdots$};
		\node (7) at (7.1,1.1) {\large $\vdots$};
		\node (7) at (7.1,0.1) {\large $\vdots$};
		\node (eq) at (1.3,1.4) {};
		\node (hor) at (-0.3,1.5) {};
		\node (bdir) at (-0.3,-0.6) {};
		\node (c3) at (1.5,1.4) {\large $\cdots$};
		\node (c4) at (-0.3,1.8) {\large $\vdots$};
		\node (c5) at (-0.3,-0.7) {\large $\vdots$};
		
		\path[->, bend left] (0) edge node[above, yshift=.2cm] {$hor.$} (eq);
		\path[->] (0) edge node[left] {$eq.$} (hor);
		\path[->] (0) edge node[left] {$b_{\lft}$} (bdir);
		
		\path[->] (0) edge node[below] {$\ver$} (1);
		\path[->, bend left] (1) edge node[above] {$\get{}\set{0^x_1, \bar{0^x_1}}$} (2);
		\path[->, bend right] (1) edge node[below] {$\get{}\set{1^x_1, \bar{1^x_1}}$} (2);
		\path[->, bend left] (3) edge node[above] {$\get{}\set{0^x_n, \bar{0^x_n}}$} (4);
		\path[->, bend right] (3) edge node[below] {$\get{}\set{1^x_n, \bar{1^x_n}}$} (4);
		\path[->, bend left] (4) edge (5);
		\path[->] (4) edge node[below] {} (k);
		\path[->, bend right] (4) edge (5');

		\node[state, fill=blue!5!white] (1y) at (7.4,-2) {};
		\node[state, fill=blue!5!white] (2y) at (5.2,-2) {};
		\node[state, fill=blue!5!white] (4y) at (3.8,-2) {};
		\node[state, fill=blue!5!white] (5y) at (1.7,-2) {};
		\node[state, fill=blue!5!white] (6y) at (0.2,-2) {};
		\node[state, fill=blue!5!white] (7y) at (-0.85,-2) {};
		\node[state, fill=blue!5!white] (8y) at (-2.3,-2) {};
		\node (c1) at (4.5,-2) {\large $\cdots$};
		\node (c2) at (-0.3,-2) {\large $\cdots$};

		\path[->, bend left=30] (k) edge (1y);	
		\path[->] (1y) edge node[above] {$\get{}\set{1^y_1, \bar{0^y_1}}$} (2y);
		\path[->] (4y) edge node[above] {$\get{}\set{0^y_k, \bar{1^y_k}}$} (5y);
		\path[->, bend right] (5y) edge node[above] {$\get{}\set{0^y_{k+1}, \bar{0^y_{k+1}}}$} (6y);
		\path[->,bend left] (5y) edge node[below] {$\get{}\set{1^y_{k+1}, \bar{1^y_{k+1}}}$} (6y);
		\path[->, bend right] (7y) edge node[above] {$\get{}\set{0^y_{n}, \bar{0^y_{n}}}$} (8y);
		\path[->,bend left] (7y) edge node[below] {$\get{}\set{1^y_{n}, \bar{1^y_{n}}}$} (8y);

		\node[state, fill=blue!5!white] (l1) at (-2.3, -4) {};
		\node[state, fill=blue!5!white] (l1') at (-0.7, -4) {};
		\node[state, fill=blue!5!white] (l2) at (0.9, -4) {};
		\node[state, fill=blue!5!white] (dom) at (2.5, -4) {$\domi$};
		\node[state, fill=blue!5!white] (dd) at (4, -4) {$d, \bar{d}$};
		\node[state, fill=blue!5!white] (win) at (6.5, -4) {$win$};
		\node (d1) at (4,-3.2) {};
		\node (dm) at (4,-4.8) {};
		
		\path[->, dashed] (dom) edge (dm);
		\path[->, dashed] (dom) edge (d1);
		\path[->, dashed] (dom) edge (dd);
		\path[->] (dd) edge node[above] {\small $\get{}$ $(D\setminus\set{d,\bar{d}})$} (win);
		\path[->,loop right] (win) edge (win);
		
		\path[->, bend right] (8y) edge node[left] {$\get{\ell}$} (l1);
		\path[->] (l1) edge node[above] {$\get{\bar{\ell}}$} (l1');
		\path[->] (l1') edge node[above] {$\rel{\bar{\ell}}$} (l2);
		\path[->] (l2) edge node[above] {$\rel{\ell}$} (dom);

	\end{tikzpicture}
	\caption{Transition system for process $q$ in the proof of
		Theorem~\ref{thm:nested}. Dashed arrows are controllable, every
		state before $\domi$ has a self-loop (not drawn) and $\get{}$ $S$
		means a sequence of forced transitions with the operations
		$\get{t}$ for each $t \in S$ (in some order). For
		simplicity only the $\ver$ case is shown.}
	\label{fig:proc-q}
\end{figure} 

	We need to show that if there is a tiling $t:\set{0, \ldots, 2^{n}-1}^2 \to
	D$ then there is a "winning" "strategy".
	The "strategy" for $q$ is to respond with the correct tiles: if the environment chooses
	$\#x$, $\#y$, $\#\bx$, $\#\bary$ the strategy chooses locks corresponding to
	$d_1$ and $\bar{d_2}$ with $d_1=t(\#x,\#y)$ and $d_2=t(\#\bx,\#\bary)$.
	The strategy of $p$ does the same but uses inverse encoding of numbers:
	considers $0$ as $1$, and $1$ as $0$.
	Similarly for $\bp$.
	
	Assume for contradiction that the "strategy" is not "winning", so we have a run
	leading to a "deadlock". 
	First, observe that the environment must have process $q$ go
        through state $\domi$ before $p$ and $\bp$ start running, because all states before
	$\domi$ have a self-loop, so $q$ cannot block there. 
	If either $p$ or $\bp$ starts before $q$ has reached
        $\domi$, then $q$ can never reach it, as one of the locks $\ell, \bar{\ell}$ will never be available again.
	
	If $q$ reached state  $\domi$ then process $p$ has no choice but to
	take $\ell$, and then the remaining locks among $x$, $y$.
	Similarly for $\bp$.
	At this stage the "strategy" $\s$ is defined so that the three processes will never
	take the same lock. 
	So $q$ cannot be blocked before reaching state $win$. 
	Thus deadlock is impossible. 
	
	For the other direction, suppose there is a "winning" "strategy" $\s$ for the
	system. 
	Observe that the "strategy" $\s_p$ for process $p$  decides which domino to
	take after the environment has decided which $x$ and $y$ locks to take. 
	So $\s_p$ defines a function $t:\set{0, \ldots, 2^{n}-1}^2 \to D$.
	Similarly $\s_{\bp}$ defines $\bar{t}$.
	
	We first show that $t(i,j)=\bar{t}(i,j)$ for all
	$i,j\in\set{0,\dots,2^n-1}$. 
	If not then consider for example the run where the environment chooses $\equ$ and then $x$,
	$\bx$ to be the representations of $i$, and $y$, $\bary$ to be representations
	of $j$.
	Suppose we have a run where process $q$ reaches state $\domi$, and
assume that $q$'s strategy tells to go to state $(d,\bar{d})$.
	Next the environment makes processes $p$ and $\bp$ reach the
        states where they chose their dominoes, $t(i,j)$ and
        $\bar{t}(i,j)$ respectively. 
	The two processes $p$ and $\bp$ then reach a "deadlock" state.
	Since we assumed that $t(i,j)\not=\bar{t}(i,j)$, process $q$ cannot reach state $\win$ from any
	state $s_{d,\bar{d}}$.
	Hence we have a "deadlock" run, a contradiction. 
	
	Once we know that the "strategies" $\s_p$ and $\s_{\bp}$ define the same
	tiling function it is easy to see that in order to be "winning" when the environment
	chooses one of the actions $\ver$, $\hor$ or $b_{\lft}$, $b_{\rgt}$, $b_{\down}$, 
	$b_{\up}$, the tiling function must be
	correct. 
\end{proof}


\section{Undecidability in the general case}\label{sec:undec}
In this section we show that the "deadlock avoidance control problem" is undecidable.
With a more involved proof we show undecidability using only 4 locks per process in~\cite{deadlockarxiv}. The case of 3 locks per process remains open.
In this section we present a lightweight proof, where processes use a
larger (but still fixed) number of locks.  

\undecfourlocks*

We begin by showing  undecidability under the assumption that
processes can already hold some locks in the initial configuration. 
We then reduce this problem to the "deadlock avoidance control
problem", where all processes start holding no lock.

\subsection{Deadlock avoidance with initialization}

\begin{figure}[h]
	\begin{tikzpicture}[AUT style,
		block/.style = {rectangle, draw, text centered, minimum height=4.2cm, minimum width=7.5cm},
		env/.style = {ellipse, draw, text centered, minimum height=1.5cm, minimum width=2.5cm},
		spec/.style = {rectangle, draw, dashed, fill=yellow!10, text centered, minimum height=1.5cm, minimum width=2.5cm},
		arrow/.style = {->, thick, >=latex'},
		scale=0.7
		]
		
		\node[env, thick] (system) {$P$};
		\node[env, right=4cm of system] (environment) {$\bP$};
		\node[right=2cm of system] (spec) {};
		\node[block, below =2cm of spec, thick] (halting) {};
		\node[above right = -1cm of halting] {$C$};
		
		\node[state, rectangle, initial, left = -1.2cm of halting] (A) {} ;
		
		\node[state, rectangle, anchor = west, above = 0.3cm of A, xshift=3.5cm] (B) {~check $b_1 b_2 \cdots = \bb_1 \bb_2 \cdots$} ;
		\node[draw, anchor = west, below = 0cm of A, xshift= 3.8cm, align=left] (C) {
			choose $i_1 i_2 \cdots$ and\\ 
			check $\a_{i_1} \a_{i_2} \cdots = b_{1} b_{2} \cdots$ and\\ 
			check $\b_{i_1} \b_{i_2} \cdots = \bb_{1} \bb_{2} \cdots$}  ;

		\path[->, bend left] (A) edge (B.west);
		\path[->, bend right] (A) edge (C.west);
		
		\draw[arrow] (system) -- (halting) node[midway, left] {$b_1 b_2 \cdots$};
		\draw[arrow]  (environment) -- (halting)  node[midway, right] {$\bb_1 \bb_2 \cdots$};
		
		
		%
	\end{tikzpicture}
	\caption{High-level view of the undecidability proof.} 
	\label{fig:undec-prnciple}
\end{figure} 

The input for the "deadlock avoidance control problem" with
initialization is a "lock-sharing system" $\tss = ((\aut_p)_{p\in \Proc}, \S^s, \S^e, T)$ and an initial configuration
$\Cinit=(\init_p,I_p)_{p\in \Proc}$ with pairwise disjoint sets $I_p
\subseteq T_p$.
The question  is whether there exists a "strategy" that
guarantees that no run from $\Cinit$ yields a global deadlock.
It turns out that this generalization of the "deadlock avoidance control problem"
is not more difficult than our original problem, as we will later see in Lemma~\ref{lem:initial}.

\begin{thm}\label{th:undec-init}
	The control problem for "LSS" with initial configuration and at most 7
	locks per process is undecidable.
\end{thm}

The proof of Theorem~\ref{th:undec-init} follows a well-known schema.
We reduce from the question whether a PCP instance has an infinite
solution.

Two processes $P$ and $\bP$ send independently sequences of bits $b_1, b_2, \ldots$ and $\bb_1, \bb_2, \ldots$ to process $C$. 
	The environment  asks $C$  either to check that $b_1 b_2 \dots =\bb_1 \bb_2 \dots$ or choose a sequence of indices $i_1, i_2, \ldots$ and check 
	that $\a_{i_1} \a_{i_2} \cdots = b_{1} b_{2} \cdots$ and
        $\b_{i_1} \b_{i_2} \cdots = \bb_{1} \bb_{2} \cdots$.
 	Since $P$ and $\bP$ do not know what is being checked, they
        have to send sequences of letters and indices that satisfy
        both conditions, i.e., an infinite PCP solution.
	The difficulty here is that $P$ and $\bP$ use only locks to
        communicate.

Formally, let $(\alpha_i,\beta_i)_{i=1}^m$ be a PCP instance with
$\alpha_i,\beta_i \in \set{0,1}^*$.
We construct a system with three processes $P, \bP, C$, using locks from the set
\[\set{c, s_0,s_1,p,\bs_0,\bs_1,\bp}\,.
\]
Process $P$ will use  locks from $\set{c, s_0,s_1,p}$, process $\bP$
locks from $\set{c, \bs_0,\bs_1,\bp}$, and $C$  all seven locks.
For the initial configuration we assume that $I_p=\set{p}$, $I_{\bP}=
\set{\bp}$ and $I_C= \set{c, s_0,s_1,\bs_0,\bs_1}$.

We describe now the three processes $P,\bP,C$.
Define first for $b=0,1$:
\begin{eqnarray*}
	\run_P(b) &=& \get{s_b} \rel{p}\, \get{c} \; \rel{s_b} \get{p}
	\rel{c}\\
	\run_{\bP}(b) &=& \get{\bs_b} \rel{\bp}\, \get{c}\; \rel{\bs_b} \get{\bp}
	\rel{c}
\end{eqnarray*}

\begin{figure}
	\begin{tikzpicture}[AUT style]
		\node[initial, state, fill=blue!5!white, minimum size=0.8cm] (A) {$P$};
		\node [left=1cm of A]  (AP) {$\Aa_P$};
		\node[initial, state, fill=blue!5!white, minimum size=0.8cm, right=3cm of A] (B) {$\bP$};
		\node[ right=1cm of A] (B') {};
		\node [right=1cm of B]  (ABP) {$\Aa_{\bP}$};
		\node[state, fill=blue!5!white, initial above, below=3cm of B'] (Cinit)
		{};
		\node [below=1cm of Cinit] (AC) {$\Aa_C$};
		\node[state, fill=blue!5!white, right=of Cinit] (Cind) {};
		\node[state, fill=blue!5!white, left=1.5cm of Cinit] (Clet) {};

		\node[state, fill=blue!5!white, right=of Cind] (I1) {\small $i$};
		\node[state, fill=blue!5!white, above=of I1] (1) {\small $1$};
		\node[state, fill=blue!5!white, below=of I1] (M) {\small $m$};
		\node[ above=0.2 cm of I1] (1') {\textbf{\vdots}};
		\node[ below=0.1 cm of I1] (M') {\textbf{\vdots}};
		\node[ right=2cm of I1] (I2) {\Large\textbf{$\cdots$}};
		\node[state, fill=blue!5!white,right=1.7cm of I2] (I2') {};
		\node[state, fill=blue!5!white,below right=0.5cm of I2'] (I2'') {};
		\node[left=2cm of I2''] (I3) {\Large\textbf{$\cdots$}};
		\node[left=0cm of I3] (I4) {};
		\node[state, fill=blue!5!white, left=1.7cm of I4] (I5) {};
		
		\node[state, fill=blue!5!white, below=of Clet] (L1) {};
		\node[state, fill=blue!5!white, above=of Clet] (L2) {};
		
		\path[->, thick, bend left, dashed] (Clet) edge node[left] {\small$\run_C(P, 0)$} (L2);
		\path[->, thick, bend left, dashed] (Clet) edge node[right] {\small$\run_C(P, 1)$} (L1);
		\path[->, thick, bend left] (L2) edge node[right] {\small$\run_C(\bP, 0)$} (Clet);
		\path[->, thick, bend left] (L1) edge node[left] {\small$\run_C(\bP, 1)$} (Clet);
		
		\path[->, thick, loop above, dashed] (A) edge node[above] {\small $\run_{P}(0)$} (A);
		\path[->, thick, loop above, dashed] (B) edge node[above] {\small$\run_{\bP}(0)$} (B);
		\path[->, thick, loop below, dashed] (A) edge node[below] {\small$\run_{P}(1)$} (A);
		\path[->, thick, loop below, dashed] (B) edge node[below] {\small$\run_{\bP}(1)$} (B);
		
		\path[->, thick] (Cinit) edge (Clet);
		\path[->, thick] (Cinit) edge (Cind);
		\path[->, thick, dashed] (Cind) edge (I1);
		\path[->, thick, dashed, bend left=15] (Cind) edge (1);
		\path[->, thick, dashed, bend right=15] (Cind) edge (M);
		\path[->, thick] (I1) edge node[above] {\scriptsize $\run_C(P,\a_i[0])$} (I2);
		\path[->, thick] (I2) edge node[above] {\scriptsize $\run_C(P,\a_i[k_i])$} (I2');
		\path[->, thick, bend left] (I2') edge node[above] {} (I2'');
		\path[->, thick] (I2'') edge node[above] {\scriptsize$\run_C(\bP,\b_i[0])$} (I3);
		\path[->, thick] (I4) edge node[above] {\scriptsize $\run_C(\bP,\b_i[k'_i])$} (I5);
		\path[->, thick, bend left=10] (I5) edge node[above] {} (Cind);

	\end{tikzpicture}
	\caption{The system used for the undecidability proof. The letters
		$\a_i[0], \ldots, \a_i[k_i]$ and $\b_i[0], \ldots, \b_i[k'_i]$ are
		defined so that $\a_i[0] \cdots \a_i[k_i] = \a_i$ and $\b_i[0]
		\cdots \b_i[k'_i] = \b_i$. Dashed transitions are controlable.}
\end{figure}

The automaton  $\aut_P$ ($\aut_{\bP}$, resp.) allows all possible action sequences from
$(\run_P(0)+\run_P(1))^\omega$ ($( \run_{\bP}(0)+ \run_{\bP}(1))^\omega$, resp.).
If e.g.~process $P$ manages to execute a sequence $\run_P(b_1) \run_P(b_2) \dots$
then this will mean that $C,P$ synchronize over the sequence
$b_1,b_2,\dots$, as we show below.

Process $C$'s behavior for checking word equality consists in
repeating the following procedure: she chooses a bit $b\in \set{0,1}$ through a
controllable action, then tries to execute $\run_C(P,b) \, \run_C(\bP,b)$, where: 
\begin{eqnarray*}
	\run_C(P,b) &=& \rel{s_b}\, \get{p} \, \rel{c} \; \get{s_b} \rel{p}
	\, 
	\get{c} \\
	\run_C(\bP,b) &=& \rel{\bs_b} \get{\bp} \, \rel{c} \; \get{\bs_b}
	\rel{\bp}\, 
	\get{c} 
\end{eqnarray*}
For index equality $C$'s behavior is similar: she chooses an index
$i \in\set{1,\dots,m}$ and then tries to do $\run_C(P,b_1) \dots \run_C(P,b_k) \run_C(\bP,b'_1) \dots
\run_C(\bP,b'_r)$, where $\alpha_i=b_1 \dots b_k$, $\beta_i=b'_1 \dots b'_r$.

Let us now prove that there is a "winning" "strategy" if and only if there is an infinite solution to the PCP instance.

We start by formalizing the intuition that the sequences $\run_P(b)$
and $\run_C(P,b)$ make the processes $P,C$ synchronize over bit $b$.

\begin{lem}
	\label{key-lem-undec}
	Let $\rho$ be a finite global run between two configurations
        $\gamma$ and $\g'$, such that the sequence of operations of
        $C$ in $\rho$ is $u_C(P,b)$, and  $C$ holds $ s_0, s_1, \bs_0, \bs_1$ in $\gamma$ .
	Then the sequence of operations executed by $P$ in $\rho$ is $u_P(b)$ and $\bP$ stays idle in $\rho$. Furthermore $C$ holds $ s_0, s_1, \bs_0, \bs_1$ in $\gamma'$.
\end{lem}

\begin{proof}
	Let us start with $\bP$. At the start it cannot be holding $c, \bs_0$ or $\bs_1$ as $C$ holds all of them. This implies that it is in its initial state, and not in one of the loops $u_{\bP}(0)$ or $u_{\bP}(1)$.
	Therefore its next action can only be to acquire $\bs_0$ or
        $\bs_1$. As those locks are never released by $C$ in  $u_C(P,b)$, $\bP$ has to stay idle.
	
	It is easy to see from $C$'s sequence of actions that in $\g'$ it holds $ s_0, s_1, \bs_0, \bs_1$.
	
	Concerning $P$, for the same reason it has to be in its initial state in $\g$.
From $\g$, process	$C$ releases $s_b$ and acquires $p$, meaning that $P$ has started executing the $u_P(b)$ loop, acquired $s_b$ and released $p$ and is waiting for $c$.
	Then $C$ releases $c$ and acquires $s_b$, which means that $P$ has taken $c$ and released $s_b$, and is waiting for $p$.
	Finally $C$ releases $p$ and acquires $c$, which implies that $P$ has acquired $p$ and released $c$, and is stuck in its initial state as $C$ holds both $s_0$ and $s_1$.
	Therefore, $P$ has executed precisely $u_P(b)$.
\end{proof}

Assume that there is a "winning" "strategy" for the problem with initialization. 
We can observe that $P$ has no incentive to allow both $u_P(0)$ and $u_P(1)$ at any point, since this leaves the choice to the environment.
On the other hand, if $P$ disallows both choices, then he keeps $p$ forever, thus $C$ will eventually be blocked as it needs to acquire $p$ infinitely often. 
Hence the lock $c$ will be held indefinitely by $C$.
Then $\bP$ will also be blocked since it needs to acquire $c$ infinitely often.
As a consequence, we can assume that $P$ uses a strategy that allows exactly one of $u_P(0), u_P(1)$ each time.

Therefore, the strategy of $P$ boils down to choosing a sequence of bits $b_0 b_1 b_2 \cdots$ and executing $u_P(b_0) u_P(b_1) u_P(b_2) \cdots$.
Similarly, $\bP$ chooses a sequence of bits $\bb_0 \bb_1 \bb_2 \cdots$ and executes $u_{\bP}(\bb_0) u_{\bP}(\bb_1)u_{\bP}(\bb_2) \cdots$.
Also, if the environment makes $C$ verify word equality, then $C$ chooses a sequence of bits $b''_0 b''_1 b''_2 \cdots$ and executes $u_P(b''_0) u_{\bP}(b''_0) u_P(b''_1) u_{\bP}(b''_1) \cdots$. 
Otherwise, $C$ chooses a sequence of indices $i_0 i_1 i_2\cdots$ and executes an interleaving of $u_P(b'_0) u_P(b'_1) u_P(b'_2) \cdots$  with $b'_0 b'_1 \cdots = \a_{i_0} \a_{i_1} \cdots$ and $u_{\bP}(\bb'_0) u_P(\bb'_1) u_P(\bb'_2) \cdots$ with $\bb'_0 \bb'_1 \cdots = \b_{i_0} \b_{i_1} \cdots$.

Let us now observe the relations between those sequences.
First of all note that if any process gets blocked forever, then so do
the other two, by Lemma~\ref{key-lem-undec}.
Thus a winning strategy should ensure that all processes run forever.
To do so, by Lemma~\ref{key-lem-undec}, the case of checking word equality implies that we should have $b_0 b_1 \cdots = b''_0 b''_1 \cdots =  \bb_0 \bb_1 \cdots$.

Moreover, the case of index equality imposes that $b_0 b_1 \cdots = \a_{i_0} \a_{i_1} \cdots$ and $\bb_0 \bb_1 \cdots = \b_{i_0} \b_{i_1} \cdots$.
As a result, we must have $\a_{i_0} \a_{i_1} \cdots = \b_{i_0} \b_{i_1} \cdots$, hence the PCP instance has an infinite solution.

Let us now show the other direction. Suppose there are indices $i_0 i_1 \cdots$ such that $\a_{i_0} \a_{i_1} \cdots = \b_{i_0} \b_{i_1} \cdots$. Let $b_0 b_1 \cdots = \a_{i_0} \a_{i_1} \cdots$
A "winning" "strategy" is to make $P$ and $\bP$ choose that same sequence of bits $b_0 b_1 \cdots$.
If $C$ has to check word equality, it chooses the sequence $b_0 b_1 \cdots$, otherwise it chooses indices $i_0 i_1 \cdots$.
In the following lemma we say that a process wants to execute a sequence of operations if those are the operations of the next transitions chosen by its "strategy".
\begin{lem}
	Assume that $C$ owns $\{s_0,s_1,c, \bs_0, \bs_1\}$, $P$ owns $\{p\}$, $C$ wants to execute
	$u_C(P,b)$, $P$ wants to execute $u_P(b)$ and $\bP$ wants to execute $u_{\bP}(b')$.
	Then $C$
	and $P$ finish executing $\run_C(P,b)$ and $\run_P(b)$ without encountering a global deadlock, $\bP$ stays idle, and the lock ownership is the same as before the
	execution.
\end{lem}
\begin{proof}
	Locks $\bs_0, \bs_1$ will never be released in the sequence we describe, thus $\bP$ has to stay idle.
	
	It suffices to observe that the environment has no choice for the sequence of operations. At first every lock is taken, and $C$ is the only process which can release one ($s_b$), so it does.
	Then process $P$ is the only one which can move, by taking $s_b$, and then releasing $p$, and so on.	
	Eventually $P$ will have executed $u_P(b)$ and $C$ will have executed $u_C(P,b)$.
\end{proof}

By construction of the strategy, no matter if the environment chooses to check word equality or index equality, the sequence of operations of $C$ is an interleaving of the sequences $u_C(P,b_0)u_C(P,b_1) \cdots$ and $u_C(\bP,b_0)u_C(\bP,b_1) \cdots$.
The sequences of operations on $P$ and $\bP$ are respectively $u_P(b_0) u_P(b_1)\cdots $ and $u_{\bP}(b_0) u_{\bP}(b_1)\cdots $.
As a consequence of this lemma, we obtain that the system cannot reach a global deadlock.

We have shown that there is a "winning" "strategy" if and only if the PCP instance has an infinite solution.

\subsubsection{Removing the initialization}

We aim to prove the following lemma:

\begin{lem}
	\label{lem:initial}
	There is a polynomial-time reduction from the "deadlock avoidance control problem" for "lock-sharing
	systems" with initialization to the control problem where all locks are initially free. 
	The reduction adds one process and $|\Proc|+1$ new locks in total.
\end{lem}

\begin{proof}
	The system $\Ss=((\aut_p)_{p \in\Proc},\S^s,\S^e,T)$ with initial ownership $(I_p)_{p \in\Proc}$
	is transformed
	into a new  system $\Ss_\es$ with one extra process and additional locks.
	The transformation introduces one extra lock for each process $p$,
	denoted $k_p$ and called \emph{the key of $p$}.
	The extra process is called $q$ and also has a key $k_q$.
	Each process $p \in \Proc$ uses in addition to $T_p$ the locks $k_p$ and $k_q$.
	
	The automaton $\Aa_q$  of $q$ consists of a sequence of states
        connected with uncontrollable transitions where $q$ first
        acquires $k_q$, then acquires and releases each $k_p$ in some arbitrary, fixed order.
	Additionally, every state except the last one has an uncontrollable $\nop$ self-loop.
	This is to make sure that $q$ must execute the full sequence
        in any run leading to a global deadlock. 
	
	The automaton $\aut_{p}$ of process $p$
	is extended by new states and transitions,
	which define a specific finite run
	called the \emph{init sequence}.
	The new states and transitions can occur only during the init sequence.
	When a process $p$ completes his init sequence in $\Ss_\es$,
	he  owns precisely all locks in $I_p$,
	plus the key $k_p$,
	and has reached his initial state $\init_p$ in $\aut_p$.
	After that, further actions and transitions played in $\Ss_\es$
	are actions and transitions of $\Ss$, unchanged.
	All the new actions are uncontrollable,
	thus there is no strategic decision to make for the controller of a
	process $p$ until 
	his init sequence is completed.

	\paragraph*{The init sequence.}
	For process $p$, the init sequence $\Iseq_p$ consists of three steps.
	
	\begin{enumerate}
		\item
		First, $p$ takes one by one  (in a fixed arbitrary order)
		all locks in $I_p$.
		
		\item
		Second, $p$  takes and releases $k_q$.
		
		\item
		Finally, $p$ acquires its key $k_p$ and reaches the
                initial state $\init_{p}$ of $\aut_{p}$. 
              \end{enumerate}
              In addition, an uncontrollable $\nop$ self-loop labels
              every state of this sequence (except for $\init_p$).
            The uncontrollable self-loops on every state of
        $\Iseq_p$ guarantee that  a deadlock may  occur only after all processes
	have fully completed their init sequences.
	
	\paragraph*{Linking runs in  $\Ss_\es$ and $\Ss$.}
	
	We establish that there is a "winning" "strategy" in one
        system if and only if there is one in the other. The reason
        for this is that essentially, in order to reach a deadlock in
        $\Ss_\es$ the environment is forced to execute the init
        sequences of all processes and then continue with an execution
        of $\Ss$. 
	
	\begin{claim}\label{cl:spade}
		If there is a "winning" "strategy" in $\Ss_\es$, 
		then there is one in $\Ss$.		
	\end{claim}

	\begin{proof}
          	Let $\s' = (\s'_p)_{p \in \Proc \cup \set{q}}$ be a
                "winning" "strategy" in $\Ss_\es$. 
	We define a "strategy" $\s = (\s_{p})_{p \in \Proc}$
                in $\Ss$ by letting $\s_{p}(u) = \s'_{p}(\Iseq_p\,u)$ for
                every local run $u$ of $p$ in $\Ss$.
                Since all transitions in $\Iseq_p$ are uncontrollable,
                $\s_p$ is well-defined.
		
		Suppose by contradiction that there is a $\s$-run
                leading to a global deadlock in $\Ss$.
	We can execute the first two steps of $\Iseq_p$ for
        each process $p$, one by one, then let $q$ execute all its transitions
        (acquire $k_q$, then acquire and release each $k_p$).
        At this point $q$ is deadlocked.
        Finally we execute the third step of $\Iseq_p$, for each
        process (acquire $k_p$). 
		We can then execute the $\s$-run leading to a global
                deadlock in $\Ss$, which also leads to a global
                deadlock in $\Ss_\es$. 
		This contradicts the assumption that $\s'$ is "winning".
	\end{proof}
	
	\begin{claim}\label{cl:diamond}
	If there is a "winning" "strategy" in $\Ss$, then there is one  in $\Ss_\es$.		
	\end{claim}
	\begin{proof}
          	Let $\s = (\s_p)_{p \in \Proc}$ be a "winning" "strategy" in $\Ss$.
		We define $\s' = (\s'_p)_{p \in \Proc}$ such that
                $\s_{p}(u) = \s'_{p}(\Iseq_p\,u)$ for every local run
                $u$ of $p$ in $\Ss$.  
		
		Suppose by contradiction that we have a $\s'$-run
                leading to a global deadlock in $\Ss_\es$.
		As every state along the init sequence has a $\nop$ self-loop, they must all have executed their init sequence in full.
		Similarly, $q$ must have entirely executed its sequence of operations.
		Each $p$ must hence have executed steps (1) and (2) of
                $\Iseq_p$, and this before $q$ has taken $k_q$.
		On the other hand, each $p \in \Proc$ must have taken $k_p$ after $q$ has taken and released it.
		
		As a consequence, there is a point in the run at which
                each $p$ has taken all locks in $I_p$,  but none of them has reached $\init_p$ yet.
		Consider the rest of the run from that point and remove every action from the init sequences and from $q$.
		We obtain a $\s$-run of $\Ss$ leading to a global deadlock.
		This contradicts the assumption that $\s$ is "winning".
	\end{proof}
	
	The two claims above prove that there is a "winning" "strategy" in one system if and only if there is one in the other.
	This concludes the reduction.
\end{proof}

We obtain Theorem~\ref{thm:undecfour} from Theorem~\ref{th:undec-init} and Lemma~\ref{lem:initial}.

\section{Conclusions}
Motivated by a recent undecidability result for distributed control
synthesis of Zielonka automata~\cite{Gimbert22} we have considered a
simpler model, for which the
problem has not been investigated yet. 
With hindsight it is strange that the well-studied model using lock synchronization
has not been considered in the context of distributed synthesis. 
One reason may be the non-monotone nature of the synthesis
problem:  for a less expressive class of systems the problem is not
necessarily easier because the controllers get less powerful, too. 

The two decidable classes of lock-sharing systems presented here are rather promising. 
Especially because the low complexity results cover already non-trivial problems. 
All our algorithms are based on analyzing lock patterns. 
While in this article we consider only finite state processes, the same method
applies to more complex systems, as long as solving the centralized control problem in the style of
Lemma~\ref{lem:polycheckpatterns} is decidable. 
This is for example the case for pushdown systems.

There are numerous directions that need to be investigated further.
We have focused on deadlock avoidance because this is a central property,
and deadlocks are difficult to discover by means of testing or verification. 
Another option is partial deadlock, where some, but not all, processes are
blocked. 
The concept of $Z$-"deadlock scheme" should help
here, but the complexity results may be different. 
Reachability, and repeated reachability properties need to be
investigated, too.

We do not know if the upper bound from Theorem~\ref{thm:2LSSNP} is tight. 
The algorithm for verifying if there is a "deadlock" in a given "lock graph",
Algorithm~\ref{alg:completealgo}, is already quite complicated, and
it is not clear how to proceed when a "strategy" is not given.

Another research direction is to consider probabilistic controllers. 
It is well known that there are no symmetric solutions to the dining philosophers
problem but there is a randomized one~\cite{LehRab81,Lyn96}.
Symmetric solutions are quite important for resilience issues as it is preferable that
every process runs the same code.  
The Lehmann-Rabin algorithm is essentially the system presented in
Figure~\ref{fig:flexible-philosophers} where the 
choice between $\lft$ and $\rgt$ is made randomly. 
This is one of the examples where randomized strategies are essential.
Distributed synthesis has a potential here because it is even more difficult to construct
distributed randomized systems and prove them correct.

\medskip

\emph{Acknowledgements.} We thank the \emph{LMCS} reviewers for the
thorough reading and their numerous and helpful comments.

\bibliographystyle{alphaurl}
\bibliography{distr.bib}
\end{document}